\documentclass[english,openany,12pt]{article}

\usepackage[cp1251]{inputenc}
\usepackage{indentfirst}
\usepackage{url}

\usepackage{cmap}
\usepackage[T2A]{fontenc}

\usepackage[english]{babel}
\usepackage{amsfonts,amssymb,amsmath,amscd,amsthm}

\usepackage{cases}

\usepackage{pgf,tikz,circuitikz}\usetikzlibrary{arrows}
\usepackage{wrapfig}
\usepackage{geometry}
\usepackage{titling}

\usepackage{color}
\usepackage{xcolor}

\long\def\tobeadded#1\endtobeadded{}
\newcommand{\tobereplaced}[2]{#1}
\newcommand{\mscomm}[1]{\begingroup\endgroup}

\newcommand{\blueit}[1]{#1}
\newcommand{\bluevar}[1]{#1}
\newcommand{\remove}{}

\usepackage{graphicx}

\usepackage[matrix,arrow,curve]{xy}
\usepackage{pdfpages}

\usepackage{diagbox}
\usepackage{makecell}
\usepackage{faktor}
\usepackage{mathtools}
\usepackage[breaklinks,pdfstartview=FitH,linktocpage,bookmarks=false]{hyperref}

\makeatletter
\newcommand{\l@myshrink}[2]{\vspace{-0.4cm}}
\makeatother

\long\def\comment#1\endcomment{}
\long\def\tmpcomment#1\endtmpcomment{}

\theoremstyle{theorem}
\newtheorem{theorem}{Theorem}
\newtheorem{lemma}{Lemma}
\newtheorem{corollary}{Corollary}
\newtheorem{proposition}{Proposition}
\newtheorem{example}{Example}

\newtheorem*{theorem2prime*}{Theorem~5$'$}
\newtheorem*{consistencyprinciple*}{Consistency Principle}

\theoremstyle{remark}
\newtheorem{remark}{Remark}

\theoremstyle{definition}
\newtheorem{definition}{Definition}
\newtheorem*{definition*}{Definition}
\newtheorem*{definitionsketch*}{Definition Sketch}
\newtheorem*{notation*}{Notation}
\newtheorem*{physicalinterpretation*}{Physical interpretation}

\newtheorem{pr}{Problem}

\newenvironment {th*}[1]
    {\gdef\thname{#1} \begin{thn}}%
    {\end{thn}}
\newtheorem*{thn}{\thname}

\begin{document}

\title{Feynman checkers:
\\ lattice quantum field theory with real time}

\author{M. Skopenkov and A. Ustinov}

\date{}

\maketitle


\begin{abstract}
We present a new completely elementary model \blueit{that} describes
\blueit{the} creation, annihilation, and motion of non-interacting electrons and positrons along a line.
It is a modification of the model known under the names
Feynman checkers \blueit{or} one-dimensional quantum walk.
\blueit{It can be viewed as a six-vertex model with certain complex weights of the vertices.}
The discrete model is consistent with the continuum quantum field theory,
namely, reproduces the known expected charge density as the lattice step
tends to zero. It is exactly solvable in terms of hypergeometric functions.
We introduce interaction resembling Fermi\blueit{'s} theory
and establish perturbation expansion.


\end{abstract}

\footnotetext{The work was prepared within the Russian Science Foundation grant N22-41-05001, \url{https://rscf.ru/en/project/22-41-05001/}. \mscomm{Check wording!!!}}



\tableofcontents

\begin{figure}[htbp]
  \centering
  \begin{tabular}{cccccc}
  \includegraphics[width=0.10\textwidth]{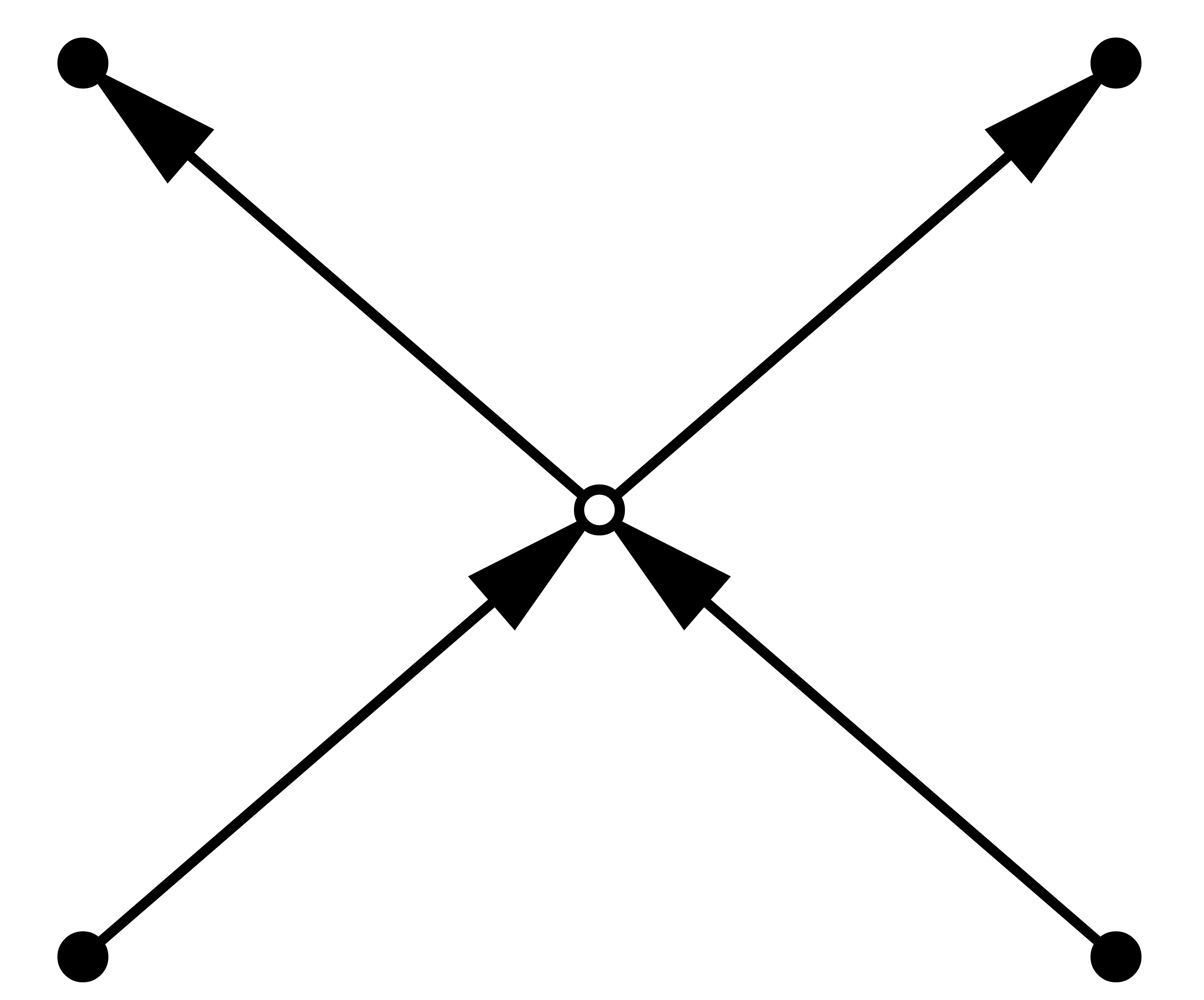}&
  \includegraphics[width=0.10\textwidth]{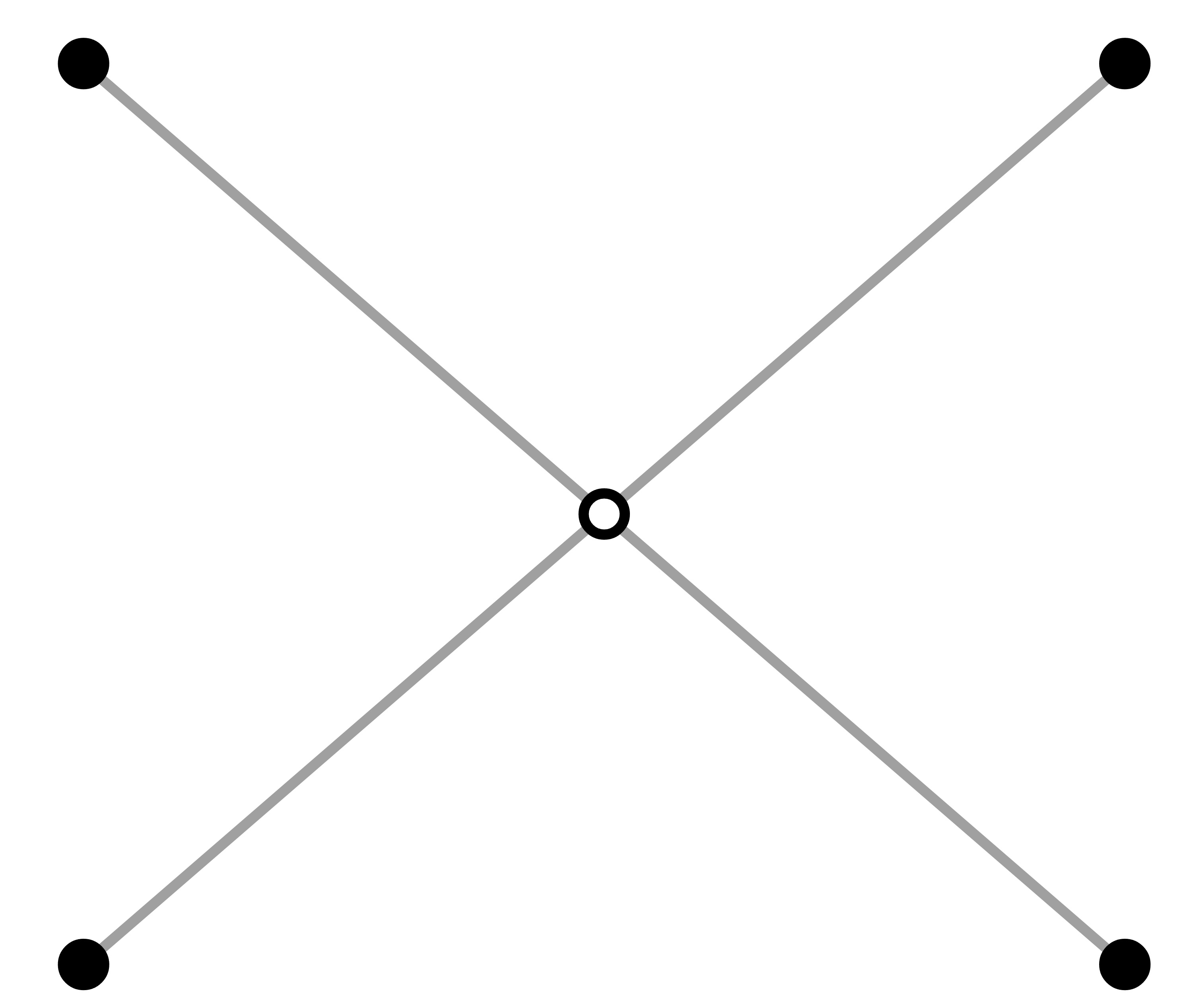}&
  \includegraphics[width=0.10\textwidth]{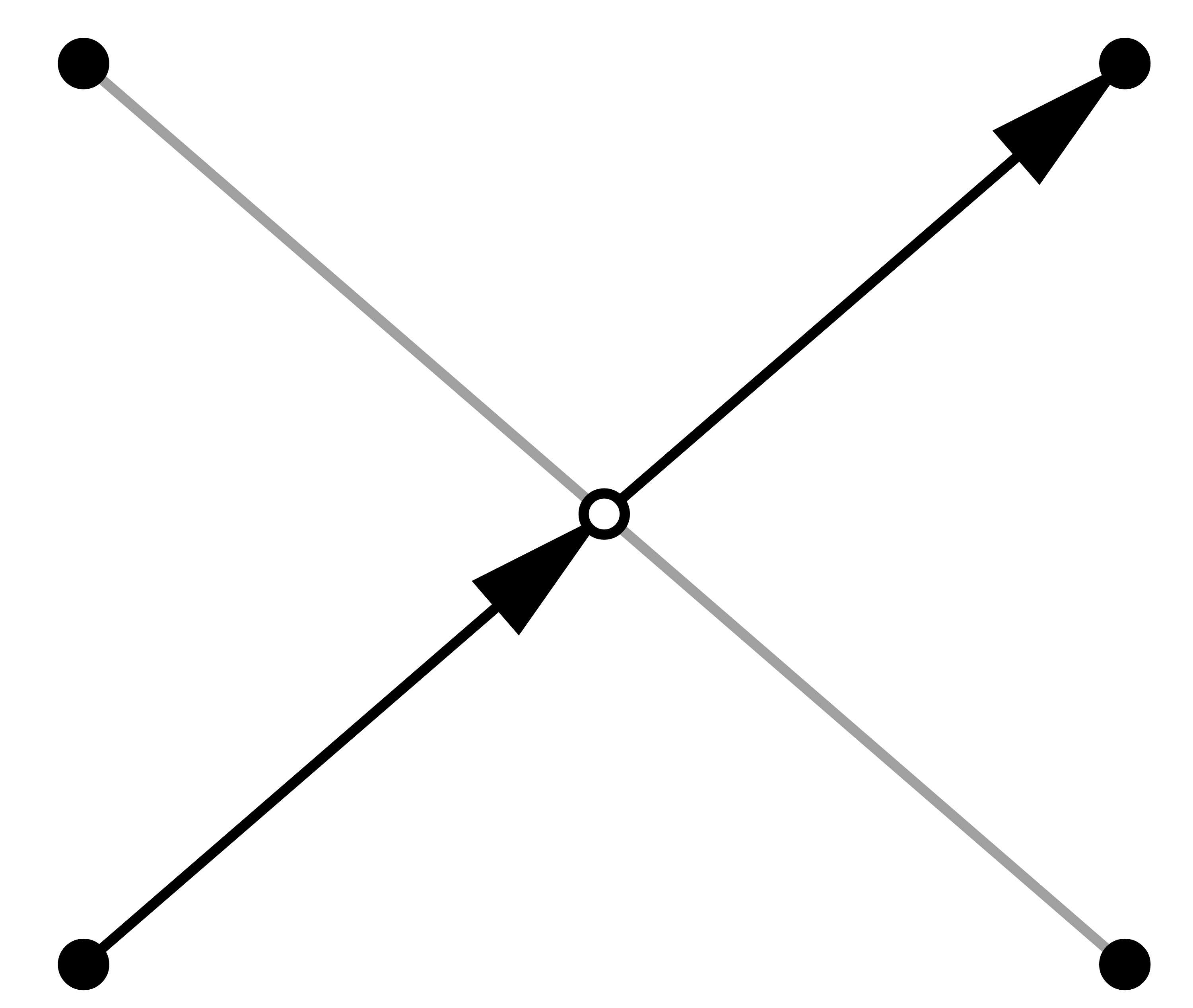}&
  \includegraphics[width=0.10\textwidth]{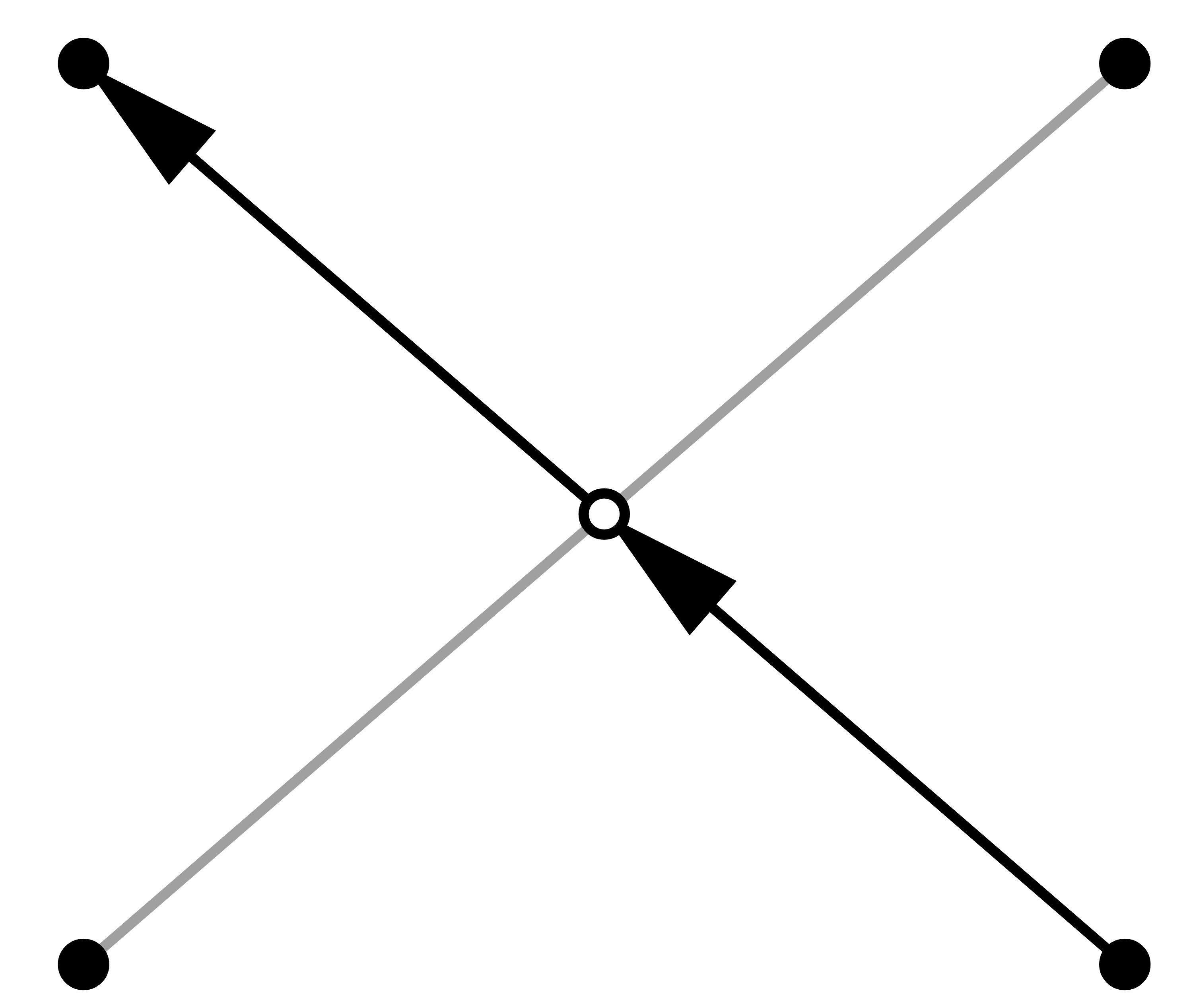}&
  \includegraphics[width=0.10\textwidth]{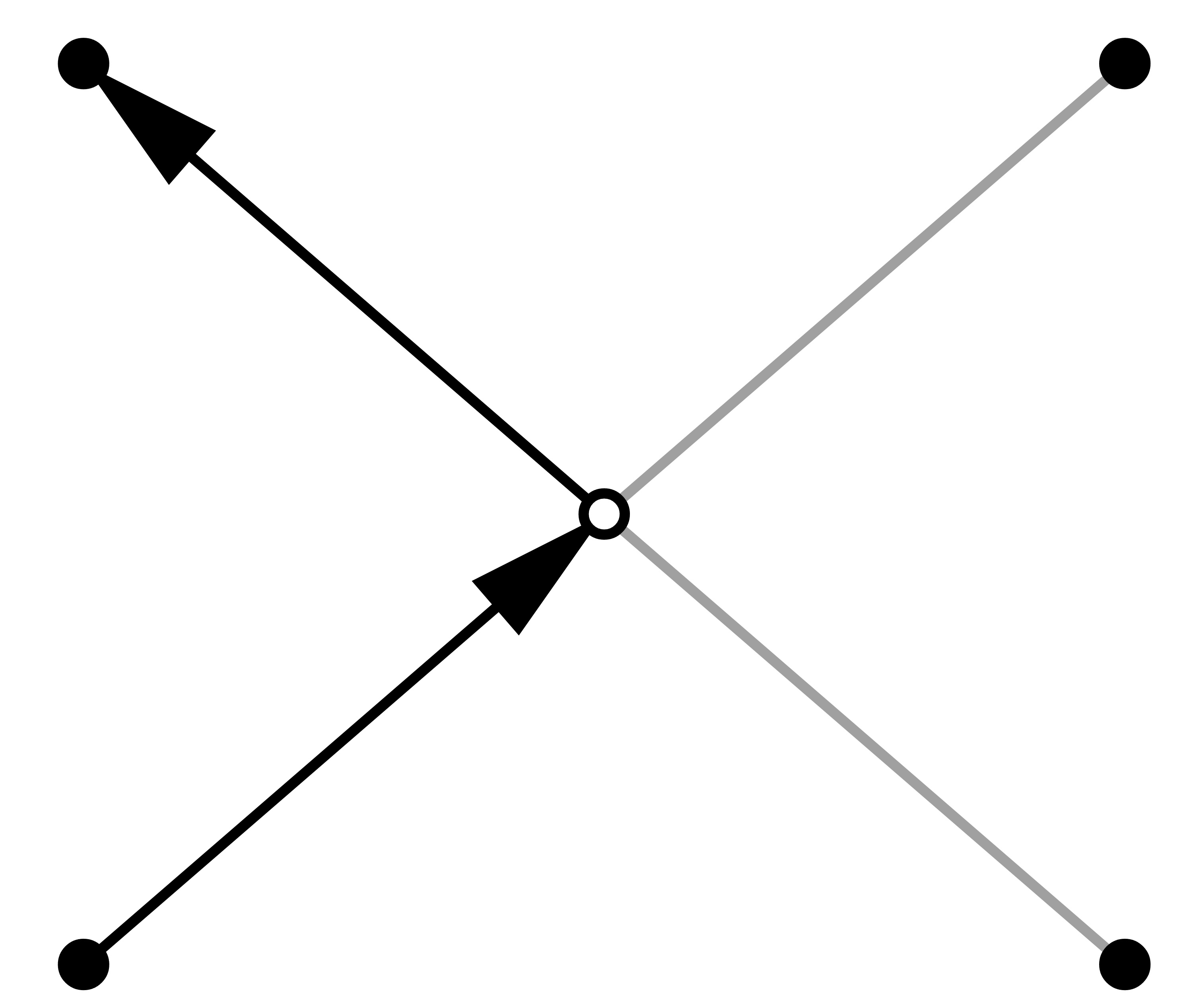}&
  \includegraphics[width=0.10\textwidth]{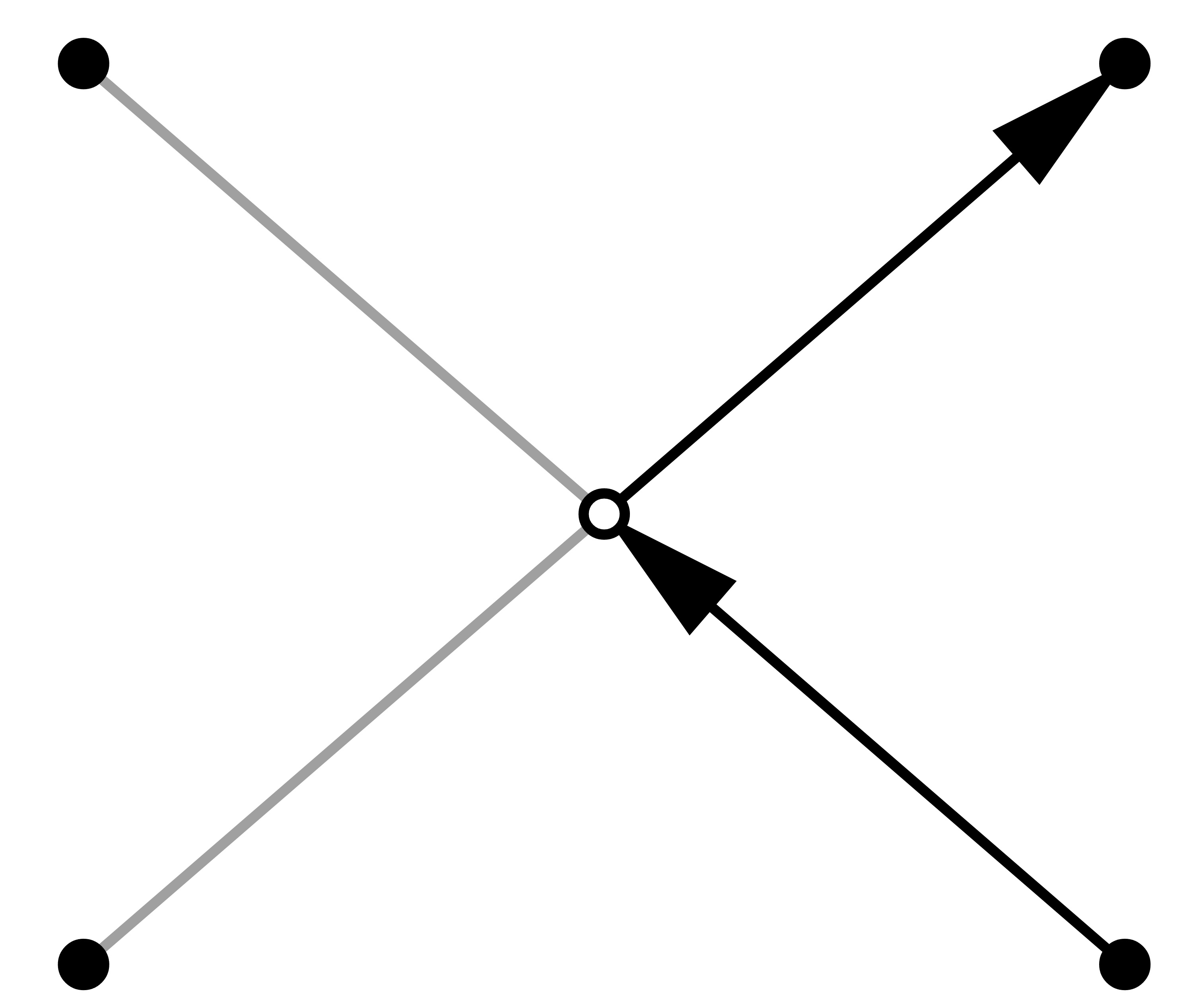}\\
  \multicolumn{2}{c}{$1$} &
  \multicolumn{2}{c}{$1/\sqrt{1+m^2\varepsilon^2}$} &
  \multicolumn{2}{c}{$-im\varepsilon/\sqrt{1+m^2\varepsilon^2}$} \\[0.3cm]
  \includegraphics[width=0.10\textwidth]{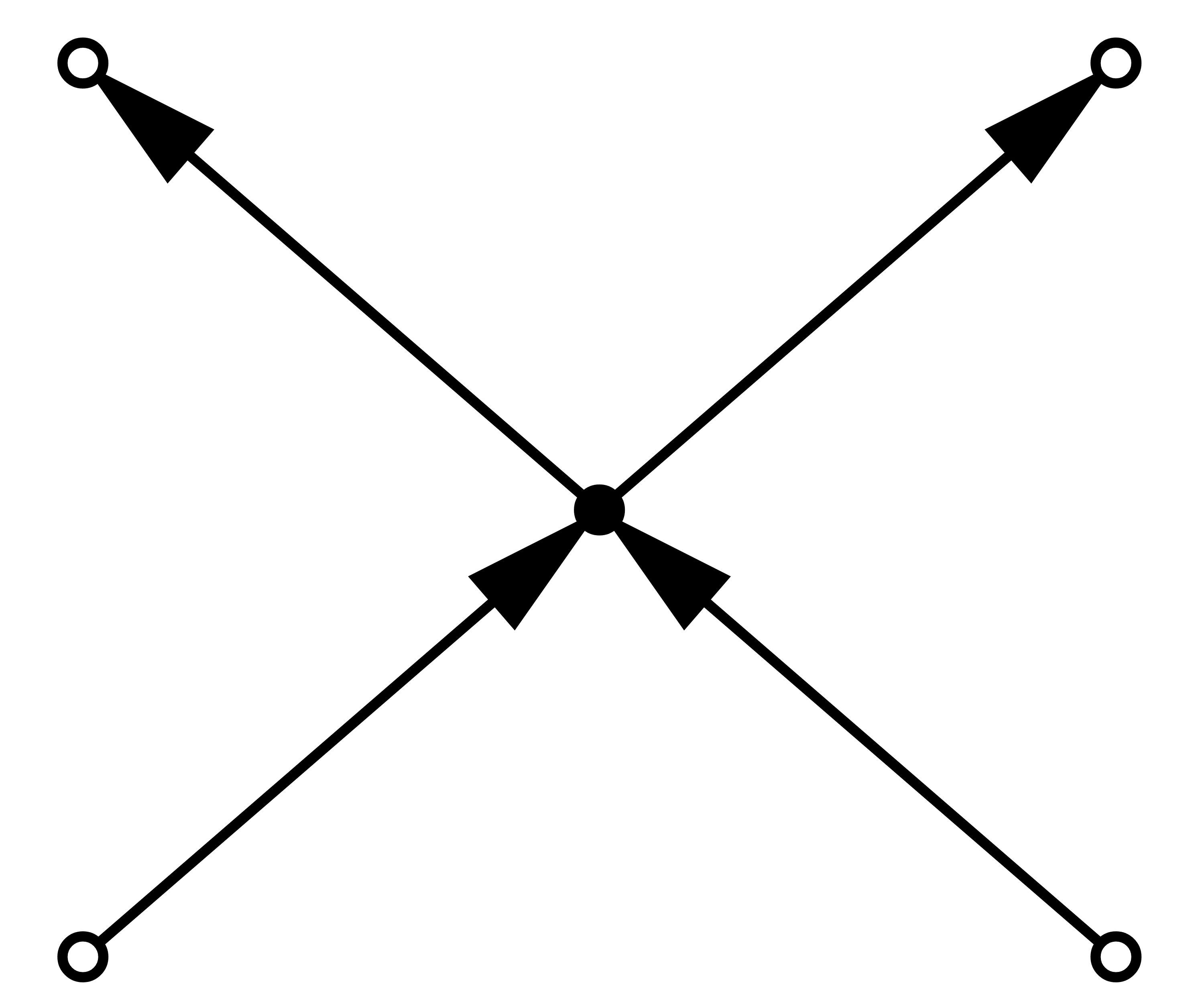}&
  \includegraphics[width=0.10\textwidth]{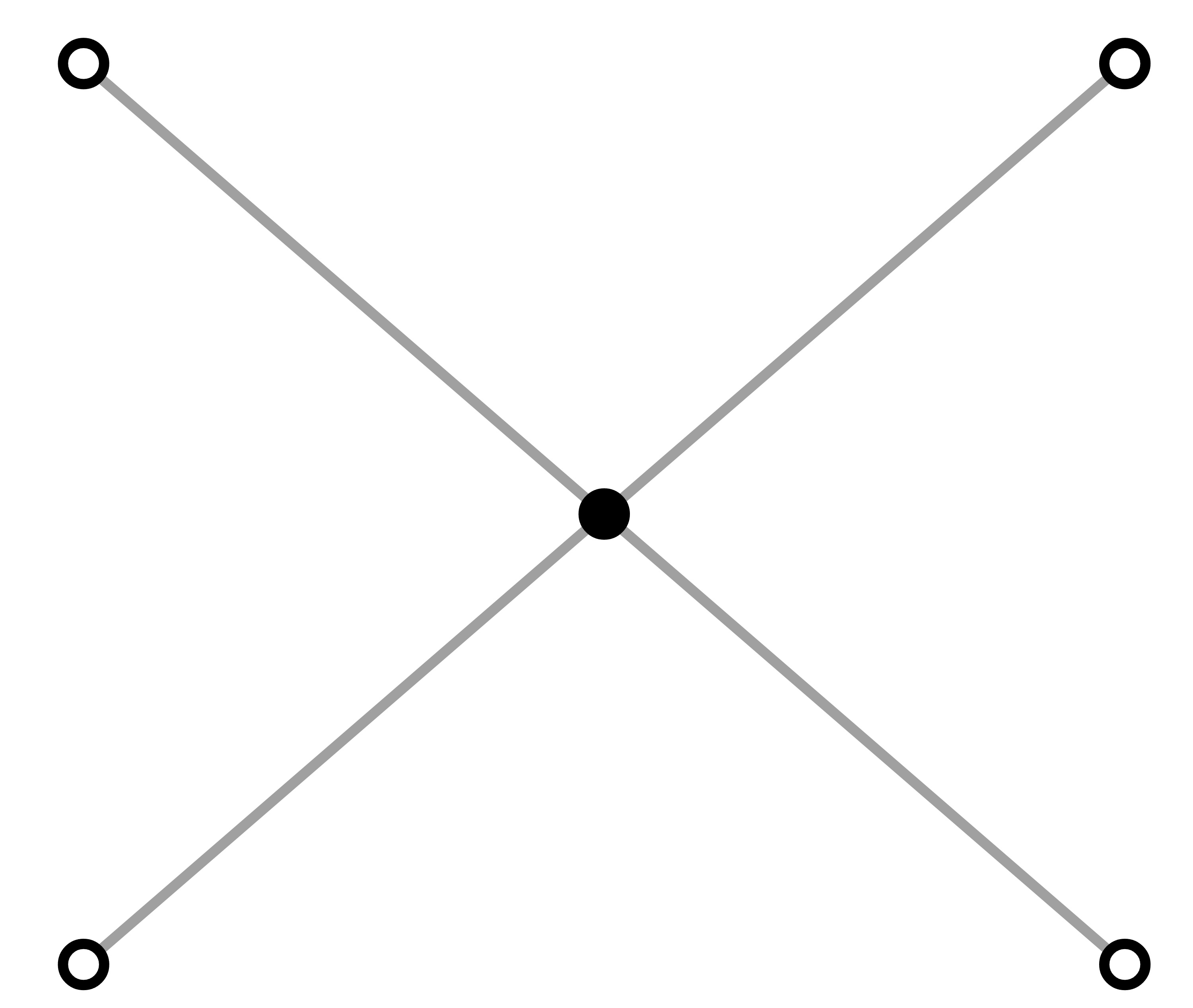}&
  \includegraphics[width=0.10\textwidth]{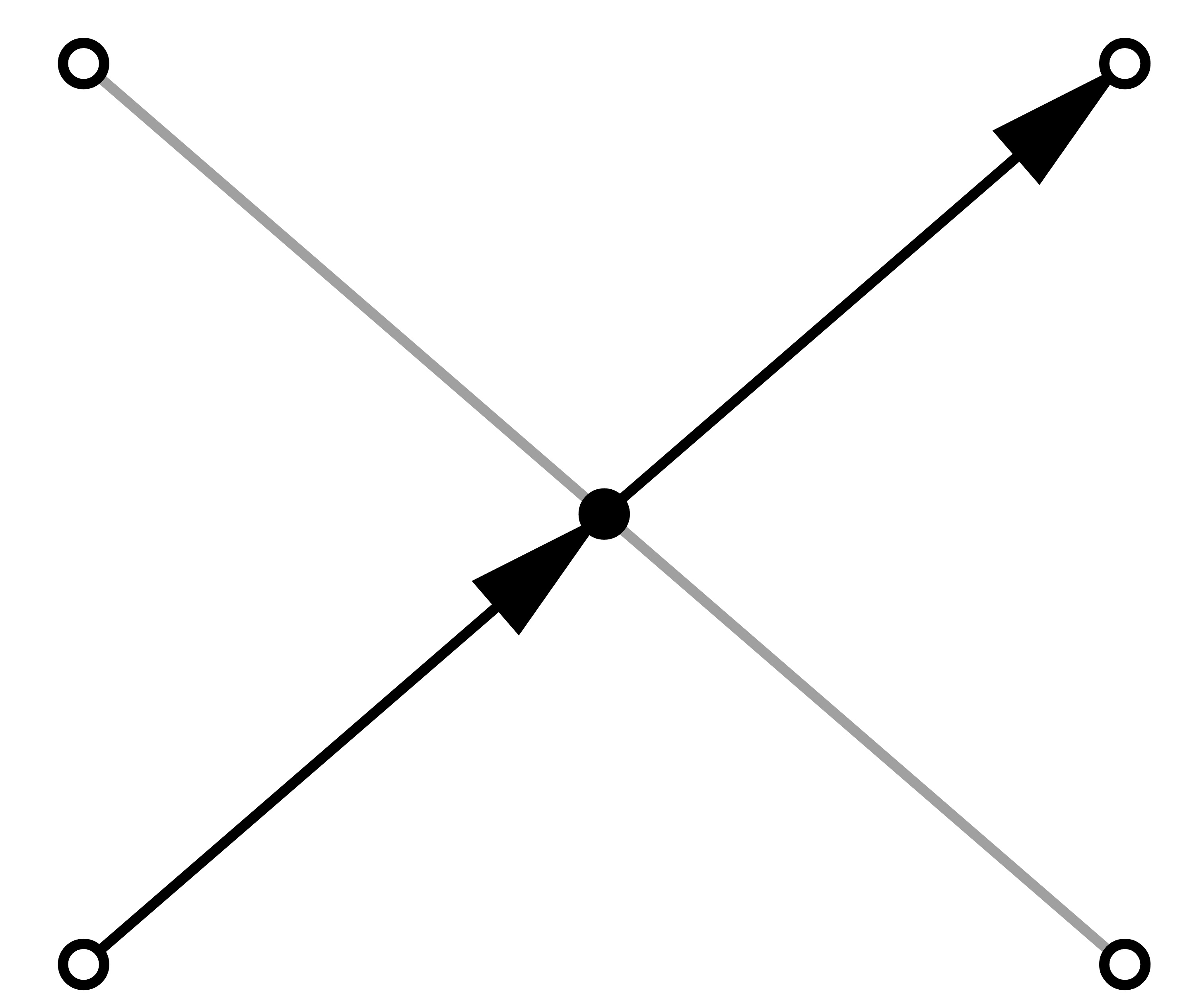}&
  \includegraphics[width=0.10\textwidth]{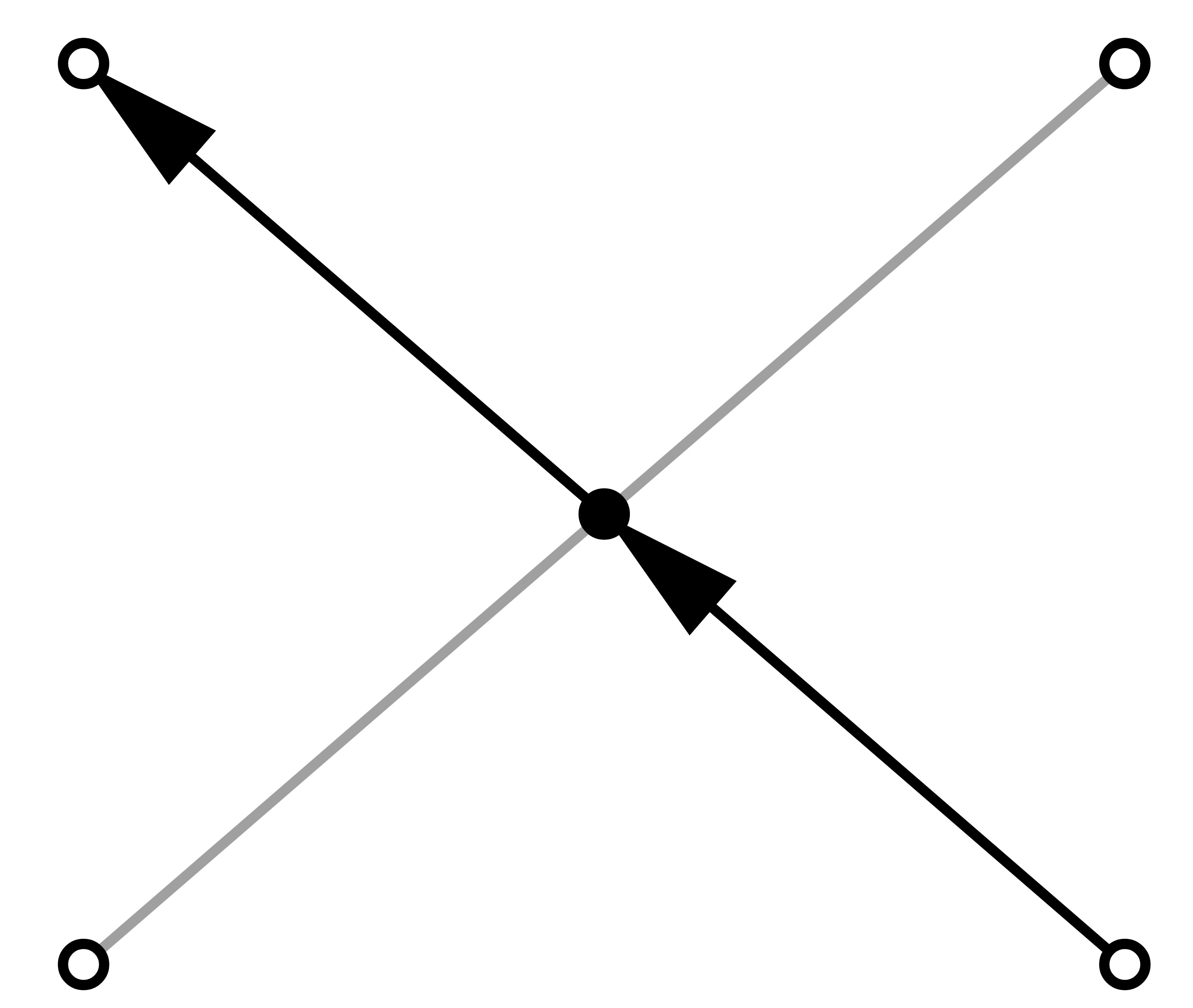}&
  \includegraphics[width=0.10\textwidth]{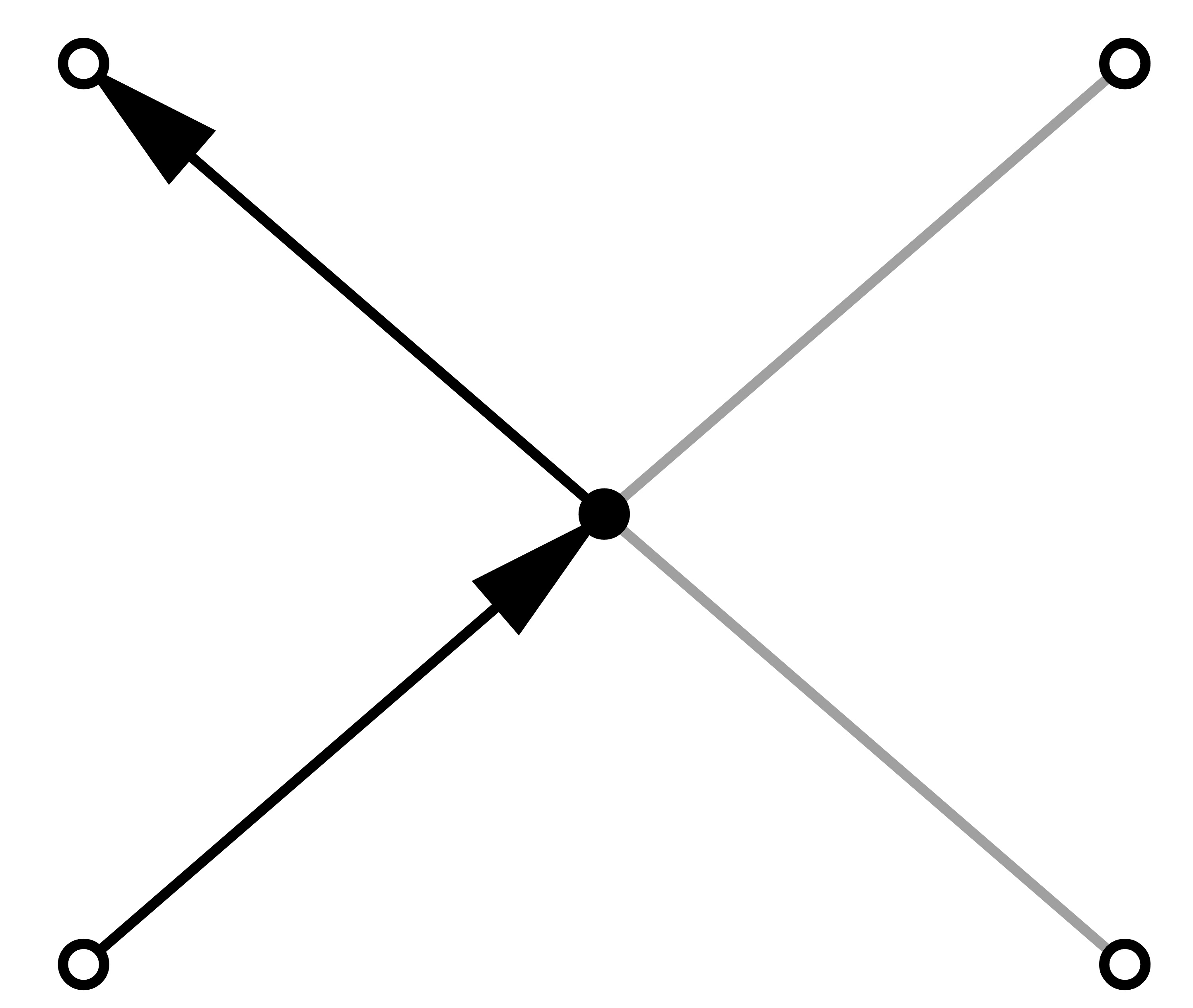}&
  \includegraphics[width=0.10\textwidth]{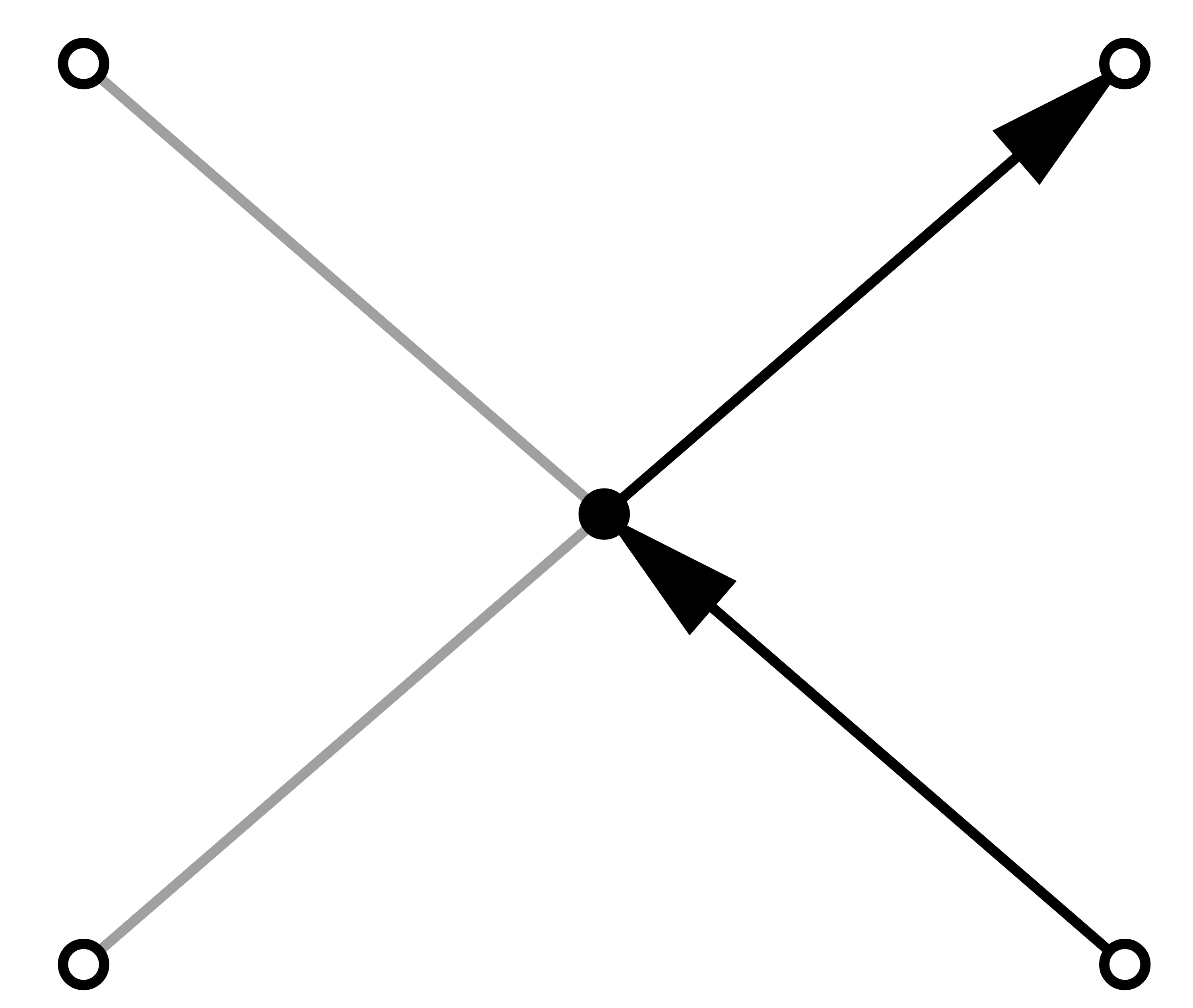}\\
  \multicolumn{2}{c}{$1$} &
  \multicolumn{2}{c}{$1/\sqrt{1-\delta^2\varepsilon^2}$} &
  \multicolumn{2}{c}{$-\delta\varepsilon/\sqrt{1-\delta^2\varepsilon^2}$}
  \end{tabular}
  \caption{\bluevar{Feynman (anti)checkers as a six-vertex model. Weights of odd (top) and even (bottom) vertices depend on mass $m$, lattice step $\varepsilon$, and small imaginary mass $\delta$. Each configuration has also an overall sign defined globally in terms of a loop decomposition.
  See Definition~\ref{def-anti-subgraphs}.}}\label{fig-6v}
\end{figure}


\addcontentsline{toc}{myshrink}{}
\addcontentsline{toc}{myshrink}{}

\section{Introduction}

We present a new completely elementary model \blueit{that} describes
\blueit{the} creation, annihilation, and motion of non-interacting electrons and positrons along a line (see Definitions~\ref{def-anti-alg}, \ref{def-anti-combi}, and~\ref{def-multipoint}).
It is a modification of the model known 
as Feynman checkers, one-dimensional quantum walk, or Ising model at \blueit{an} imaginary temperature
(see Definition~\ref{def-mass}
and surveys~\cite{
Konno-20,SU-22,Venegas-Andraca-12}).
\blueit{It can be viewed as a six-vertex model with complex vertex weights (see Figure~\ref{fig-6v} and Definition~\ref{def-anti-subgraphs}).}

This modification preserves known identities (see \S\ref{ssec-identities})
and Fourier integral representation (see Proposition~\ref{th-equivalence})
but adds antiparticles (and thus is called \emph{Feynman anticheckers}).
The discrete model is consistent with the continuum quantum field theory,
namely, reproduces the known expected charge density as the lattice step tends to zero
(see Figure~\ref{fig-approx-q} and Corollary~\ref{cor-anti-uniform}).
It is exactly solvable via 
hypergeometric functions
(see Proposition~\ref{p-mass5}) and is described asymptotically
by Bessel and trigonometric functions (see Theorems~\ref{th-continuum-limit}--\tobereplaced{\ref{th-anti-ergenium}}{\ref{th-anti-Airy}}).
We introduce interaction resembling Fermi\blueit{'s}  theory
and get perturbation expansion (see Definition~\ref{def-fermi}
and Proposition~\ref{p-perturbation}).

\begin{figure}[bhtp]
  \centering
\includegraphics[width=0.3\textwidth]{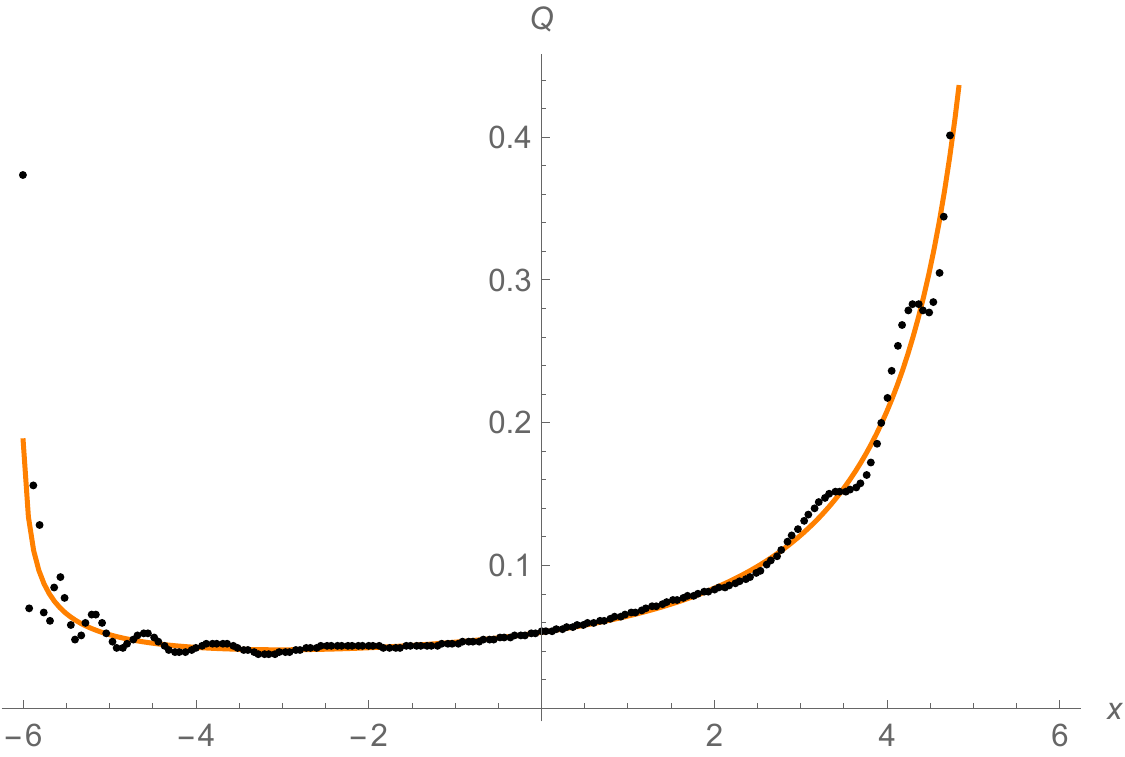}
\includegraphics[width=0.3\textwidth]{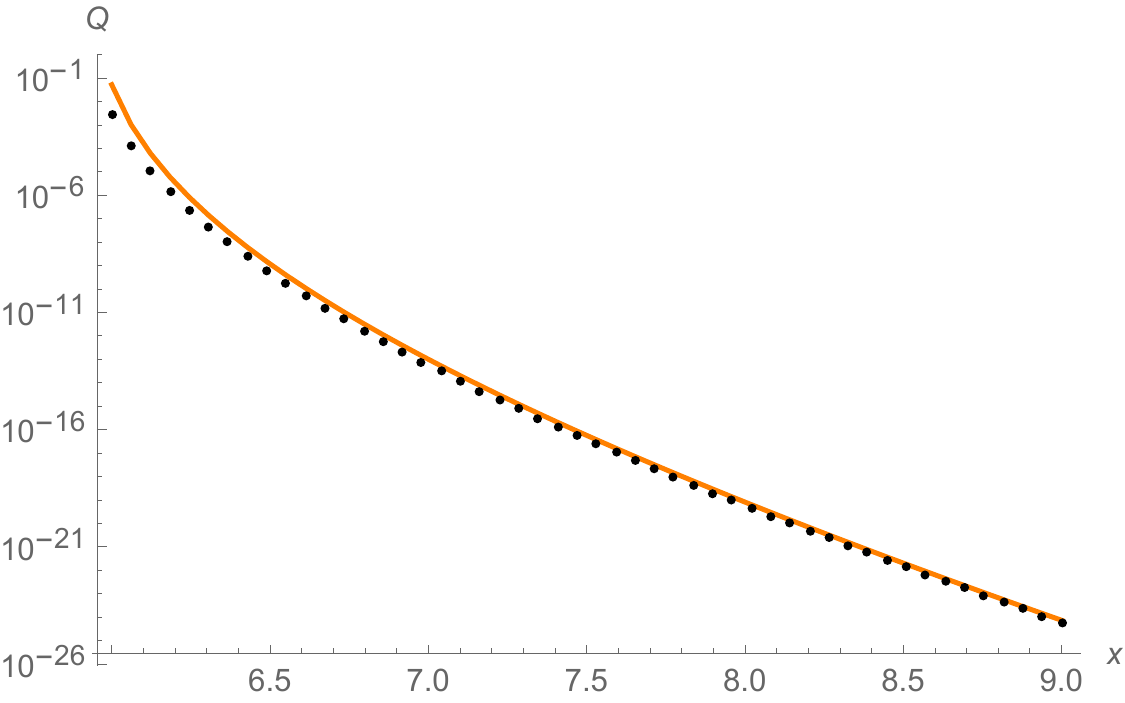}
\vspace{-0.4cm}
\caption{Normalized expected charge density
in the discrete model (dots) and continuum quantum field theory (curves).
Plots depict the left (dots) and the right side (curves) of~\eqref{eq-cor-anti-uniform}
for mass $m=4$, lattice step $\varepsilon=0.03$, time $t=6$, and position $x$ being an even multiple of $\varepsilon$.
}
\vspace{-0.2cm}
\label{fig-approx-q}
\end{figure}

\subsection{Background}


One of the main open problems in mathematics is a rigorous definition of quantum field theory. For instance,
the case of $4$-dimensional Yang--Mills theory is a Millennium Problem.


A perspective approach to the problem is \emph{constructive field theory}, which constructs a continuum theory as a limit of discrete ones \cite{Glimm-Jaffe-12}.
This leads to the \emph{consistency} question of whether
the discrete objects indeed approximate the desired continuum ones.

Constructive field theory is actually as old as quantum field theory itself.
The most elementary model of electron motion known as \emph{Feynman checkers}
or \emph{quantum walk} was introduced by R.Feynman in 1940s
and first published in 1965 \cite{Feynman-Gibbs}.
Consistency with continuum quantum mechanics was posed as a problem there
\cite[Problem~2.6]{Feynman-Gibbs}; it was solved mathematically
only recently 
\cite{SU-22}. See also surveys~\cite{
Konno-20,SU-22,Venegas-Andraca-12}
on Feynman's model and its generalizations.

In \blueit{the} 1970s F.~Wegner and K.~Wilson introduced \emph{lattice gauge theory} as a computational tool for gauge theory describing all known interactions (except gravity); see \cite{Maldacena-16} for a popular-science introduction.
\blueit{Electrons and more general \emph{fermions} are also incorporated in this theory.}
This theory is \emph{Euclidean} in the sense that it involves \emph{imaginary} time.
Euclidean lattice field theory became one of the main
computational tools \cite{Rothe-12} and culminated in determining the proton mass theoretically with \blueit{an} error less than $2\%$ in a sense.
There were developed methods to establish consistency, such as \emph{Reisz power counting theorem}
\cite[\S13.3]{Rothe-12}. This \blueit{led} to some rigorous constructions of field theories in dimensions 2 and 3 \cite{Glimm-Jaffe-12}.

Euclidean 
theories are related to 
statistical physics via the equivalence between imaginary
time and temperature \cite[\S V.2]{Zee-10}.
For instance, Feynman checkers can be viewed as an Ising model at \blueit{an} imaginary
temperature \cite[\S2.2]{SU-22}, whereas \emph{Euclidean} 
\blueit{fermions emerge, for instance, in the critical or near-critical Ising, dimer, and six-vertex models} at real temperature.
S.~Smirnov and coauthors 
proved consistency \blueit{(in the above sense)} for a wide class of
$2$-dimensional models including the Ising one and some loop $O(n)$ models \cite{Smirnov-10, KSS-23}. \blueit{The consistency for the six-vertex model is a hot research topic being developed by H.~Duminil-Copin and others \cite{Duminil-Copin-etal-22}.}
\remove{}
Euclidean 
theories 
suffer from \emph{fermion doubling}, unavoidable
by the \emph{Nielsen--Ninomiya no-go theorem},
and often making 
them inconsistent \cite[\S4.2]{Rothe-12}. 
\mscomm{Contradiction?}

A promising modern direction is \emph{Minkowskian} lattice quantum field theory \cite{Alexandru-etal-22},
where time is \emph{real}
and fermion doubling is possibly avoided
%
\cite[\S4.1.1]{Foster-Jacobson-17}, \cite[\S VI]{Arnault-Cedzich-22}.
\blueit{One feature of theories with real time is the emergence of \emph{oscillating integrals} instead of exponentially damped ones; such oscillating integrals can be analyzed by well-developed number-theoretic tools.}
Feynman checkers is a reasonable starting point \blueit{in this direction.}
It was an old dream to incorporate 
creation and annihilation of electron-positron pairs in it
(see \cite[p.~481--483]{Schweber-86}, \cite{Jacobson-85}), celebrating a passage from quantum mechanics to quantum field theory (\emph{second quantization}).
One looked for a combinatorial model reproducing \emph{Feynman propagator}~\eqref{eq-feynman-propagator} 
in the continuum limit. This cannot be achieved by any kind of quantum walks or \emph{quantum cellular automata}
(discussed e.g. in the work \cite{Arrighi-et-al} by P.~Arrighi--C.~B\'eny--T.~Farrelly and in the references there) because the Feynman propagator does not vanish for negative time (``allows propagation backward in time'').
Other known constructions (such as \emph{hopping expansion} \cite[\S12]{Rothe-12}) did not lead to the Feynman propagator because of fermion doubling.
In the \emph{massless} case, a \emph{non-combinatorial} construction 
was given by C.~Bender--L.~Mead--K.~Milton--D.~Sharp in \cite[\S9F]{Bender-etal-94} and~\cite[\S IV]{Bender-etal-85}. 

The desired combinatorial construction 
is finally given in this paper (realizing two steps of the program
outlined in \cite[\S\S8--9]{SU-22}).
It follows a classical approach known from
Kirchhoff matrix-tree theorem, the Kasteleyn and Kenyon theorems \cite{Levine-11,Kenyon-02}.
In this approach, a physical system (an electrical network, a moving electron, etc)
is described by a difference equation on the lattice (lattice Laplace equation, lattice Dirac equation
from Feynman checkers, etc).  The solution is expressed through
determinants, interpreted combinatorially
via loop expansion \cite[\S12.3]{Rothe-12}, and computed explicitly by
Fourier transform. In our setup, the solution is not unique
and is \emph{regularized} first by introducing a small imaginary mass ``living'' on the dual lattice.
In the proof of consistency, we overcome the \emph{sign problem} \cite{Alexandru-etal-22} by a set of number-theoretic tools.


\subsection{Organization of the paper}

In \S\ref{sec-statements} we define the new model and list its properties.
In \S\ref{sec-variations} we discuss its generalizations and
in \S\ref{sec-proofs} we give proofs. In Appendices~\ref{ssec-proofs-alternative} and~\ref{sec-axioms}
we give alternative definitions/proofs and put the model in the general framework
of quantum field theory respectively.

The paper assumes no background in physics. The definitions in
\S\ref{sec-statements}--\ref{sec-variations} are completely elementary
(in particular, we use neither Hilbert spaces nor Grassman variables).

The paper is written in a mathematical level of rigor, in the sense that all the definitions, conventions, and theorems (including corollaries, propositions, \blueit{and} lemmas) should be understood literally. Theorems remain true, even if cut out from the text. The proofs of theorems use the statements but not the proofs of the other ones. Most statements are much less technical than the proofs; hence the proofs are kept in a separate section (\S\ref{sec-proofs}) and long computations are kept in~\cite{SU-3}. Remarks are informal and usually not used elsewhere (hence skippable). Text outside definitions, theorems, \blueit{and} proofs is neither used formally.

\addcontentsline{toc}{myshrink}{}

\section{Feynman anticheckers: statements}
\label{sec-statements}

\subsection{Construction outline}
\label{ssec-outline}

Let us recall the definition of Feynman's original model,
and then outline how it is modified.

\begin{definition} \label{def-mass} 
Fix $\varepsilon>0$ and $m\ge 0$ called \emph{lattice step} and \emph{particle mass} respectively.
Consider the lattice $\varepsilon\mathbb{Z}^2
=\{\,(x,t):x/\varepsilon,t/\varepsilon\in\mathbb{Z}\,\}$
(see Figure~\ref{fig-paths} to the left).
A \emph{checker path} $s$ is a finite sequence of lattice points such that the vector from each point (except the last one) to the next one equals either $(\varepsilon,\varepsilon)$ or $(-\varepsilon,\varepsilon)$. Denote by $\mathrm{turns}(s)$ the number of points in $s$ (not the first and not the last one) such that the vectors from the point to the next and to the previous ones are orthogonal.
To the path $s$, assign the complex number
$$
a(s):=(1+m^2\varepsilon^2)^{1-n/2}\,i
(-im\varepsilon)^{\mathrm{turns}(s)},
$$
where $n$ is the total number of points in $s$.
For each $(x,t)\in\varepsilon\mathbb{Z}^2$, where $t>0$,
denote by 
\begin{equation*}
{a}(x,t,m,\varepsilon)
:=\sum_s a(s)
\end{equation*}
the sum over all checker paths $s$ from $(0,0)$ to $(x,t)$ containing $(\varepsilon,\varepsilon)$.
An empty sum is set to be zero by definition.
Denote
$$
a_1(x,t,m,\varepsilon):=\mathrm{Re}\,a(x,t,m,\varepsilon),
\quad
a_2(x,t,m,\varepsilon):=\mathrm{Im}\,a(x,t,m,\varepsilon).
$$
\end{definition}

\begin{physicalinterpretation*} One interprets $|{a}(x,t,m,\varepsilon)|^2$ as the probability to find
an electron of mass $m$ in the interval of length $\varepsilon$
around the point $x$ at the time $t>0$, if the electron
was emitted from the origin at the time $0$. Hereafter we work in a natural system of units where $\hbar=c=e=1$
(setting all those constants to $1$ is possible for
vacuum permittivity $\varepsilon_0\ne 1$).
\end{physicalinterpretation*}

\begin{figure}[htbp]
  \centering
\includegraphics[height=2.4cm]{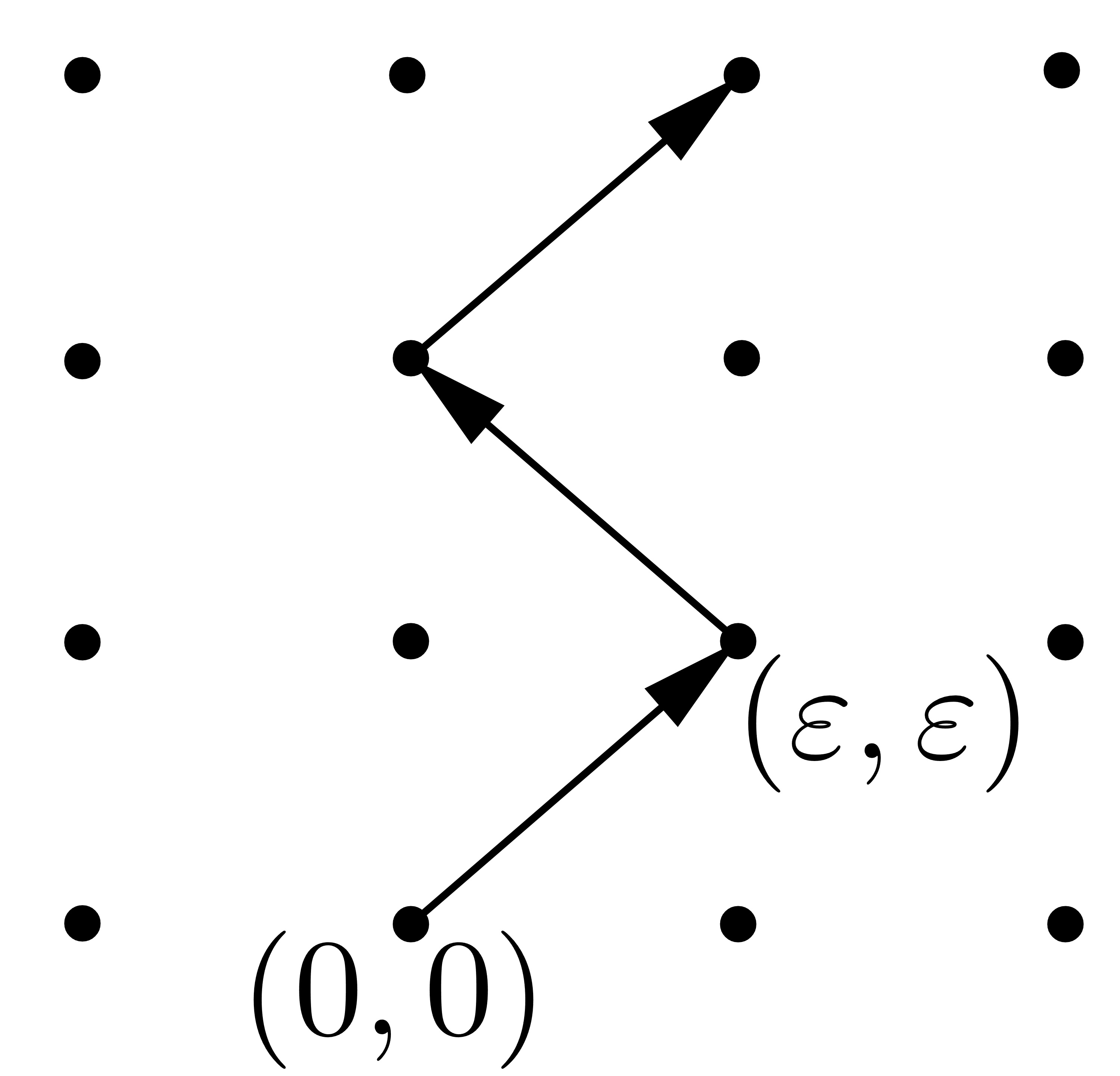}
\qquad
\includegraphics[height=2.4cm]{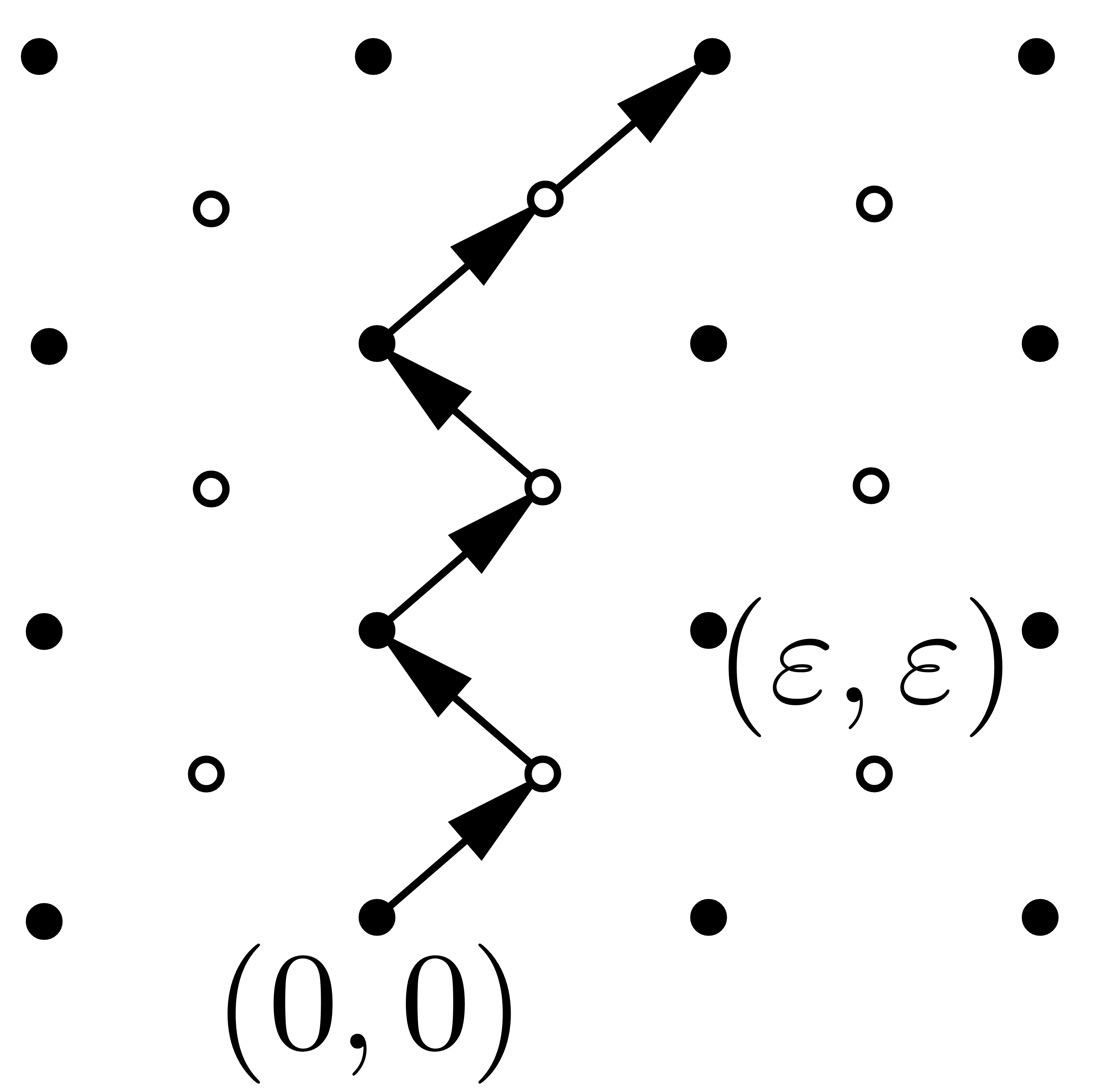}
\caption{A checker path (left). A generalized checker path (right).
See Definitions~\ref{def-mass} and~\ref{def-anti-combi}.} 
  \label{fig-paths}
\end{figure}

We have the following recurrence relation
called \emph{lattice Dirac equation} \cite[Proposition~5]{SU-22}:
\begin{align*}
a_1(x,t,m, \varepsilon) &= \frac{1}{\sqrt{1+m^2\varepsilon^2}}
(a_1(x+\varepsilon,t-\varepsilon,m, \varepsilon)
+ m \varepsilon\, a_2(x+\varepsilon,t-\varepsilon,m, \varepsilon)),\\
a_2(x,t,m, \varepsilon) &= \frac{1}{\sqrt{1+m^2\varepsilon^2}}
(a_2(x-\varepsilon,t-\varepsilon,m, \varepsilon)
- m \varepsilon\, a_1(x-\varepsilon,t-\varepsilon,m, \varepsilon)).
\end{align*}


Informally, the new model \blueit{(see Definition~\ref{def-anti-alg})}
is obtained by the following modification of lattice Dirac equation:
\begin{description}
\item[Step 0:] the functions $a_1$ and $a_2$ are extended to the \emph{dual} lattice,
shifts by $\pm\varepsilon$ are replaced by shifts by $\pm\varepsilon/2$ in their arguments, and
a term vanishing outside the origin is added;
\item[Step 1:] the particle mass acquires small imaginary part which we eventually tend to zero;
\item[Step 2:] on the dual lattice, the mass is replaced by its imaginary part.
\end{description}
This makes lattice Dirac equation uniquely solvable in $L^2$,
and the solution is much different from $(a_1,a_2)$:
we get two complex functions rather than components of one complex function.

The elementary combinatorial definition (see Definition~\ref{def-anti-combi}) is obtained from Feynman's one (see Definition~\ref{def-mass}) by slightly more involved modification
starting with the same Step~1:
\begin{description}
\item[Step 2${}'$:] just like the real mass is ``responsible'' for turns at the points of the lattice, the imaginary one allows turns at the points of the \emph{dual} lattice (see Figure~\ref{fig-paths} to the right);
\item[Step 3${}'$:] the infinite lattice is replaced by a torus with the size eventually tending to infinity;
\item[Step 4${}'$:] the sum over checker paths is replaced by a ratio of sums over loop configurations.
\end{description}
The resulting loops are interpreted as the \emph{Dirac sea} of electrons filling the whole space, and the edges not in the loops form paths of holes in the sea, that is, positrons.

In this informal outline, Steps~2 and~2${}'$ are completely new whereas the other ones are standard. The former reflect a general principle that the real and the imaginary part of a quantity should be always put on dual lattices.
Thus in what follows we consider a new lattice which is the disjoint union
of the initial lattice $\varepsilon\mathbb{Z}^2$ and its dual; the latter become sublattices.

\blueit{Each of Steps~0-3${}'$ is necessary. Omitting any one 
makes the model ill-defined or equivalent to 
Feynman's one (see Remarks~\ref{rem-delta-zero} and~\ref{rem-z}).
As for Step~4${}'$, the result of omitting it is unknown.}

\begin{remark}
\blueit{Both the old model and the new one can be viewed as \emph{six-vertex models} with certain complex weights. For instance, in Definition~\ref{def-mass}, the number $a(s)$ can be alternatively defined as the product of the vertex weights shown in Figure~\ref{fig-6v} to the top, over all the lattice vertices. (Here the thin gray segments are the ones not contained in the checker path $s$, and the leftmost vertex type cannot actually occur in the setup of Definition~\ref{def-mass}.) In the new model, there is an additional contribution from the vertices of the dual lattice with weights shown in Figure~\ref{fig-6v} to the bottom, and also a globally defined overall sign. We return to this topic in Definition~\ref{def-anti-subgraphs}.}
\end{remark}


\subsection{Axiomatic definition}
\label{ssec-alg-def}

\begin{definition} \label{def-anti-alg}
Fix $\varepsilon,m,\delta>0$ called \emph{lattice step, particle mass}, and \emph{small imaginary mass} respectively. Assume $\delta\varepsilon<1$.
For two elements $x,y$ of the same set
denote
\begin{equation*}
\delta_{xy}:=
\begin{cases}
1, &\text{for }x=y,\\
0, &\text{for }x\ne y.
\end{cases}
\end{equation*}
Define a pair of complex-valued functions $A_k(x,t)=A_k(x,t,m,\varepsilon,\delta)$,
where $k\in\{1,2\}$, on the set
$\{\,(x,t)\in\mathbb{R}^2:
2x/\varepsilon,2t/\varepsilon,(x+t)/\varepsilon\in\mathbb{Z}\,\}$
by the following $3$ conditions:
\begin{description}
\item[axiom~1:] for each $(x,t)$ with $2x/\varepsilon$ and $2t/\varepsilon$ even,
\begin{equation*}
\begin{aligned}
A_1(x,t
) &= \frac{1}{\sqrt{1+m^2\varepsilon^2}}
\left(A_1\left(x+\frac{\varepsilon}{2},t-\frac{\varepsilon}{2}
\right)
+ m \varepsilon\, A_2\left(x+\frac{\varepsilon}{2},t-\frac{\varepsilon}{2}
\right)\right),\\
A_2(x,t
) &= \frac{1}{\sqrt{1+m^2\varepsilon^2}}
\left(A_2\left(x-\frac{\varepsilon}{2},t-\frac{\varepsilon}{2}
\right)
- m \varepsilon\, A_1\left(x-\frac{\varepsilon}{2},t-\frac{\varepsilon}{2}
\right)\right)
+2\delta_{x0}\delta_{t0};
\end{aligned}
\end{equation*}
\item[axiom~2:]
for each $(x,t)$ with $2x/\varepsilon$ and $2t/\varepsilon$ odd,
\begin{equation*}
\begin{aligned}
A_1(x,t
) &= \mathrlap{
\frac{1}{\sqrt{1-\delta^2\varepsilon^2}}
\left(A_1\left(x+\frac{\varepsilon}{2},t-\frac{\varepsilon}{2}
\right)
-i\delta\varepsilon\, A_2\left(x+\frac{\varepsilon}{2},t-\frac{\varepsilon}{2}
\right)\right),
}
\hphantom{\frac{1}{\sqrt{1+m^2\varepsilon^2}}
\left(A_2\left(x-\frac{\varepsilon}{2},t-\frac{\varepsilon}{2}
\right)
- m \varepsilon\, A_1\left(x-\frac{\varepsilon}{2},t-\frac{\varepsilon}{2}
\right)\right)
+2\delta_{x0}\delta_{t0};
}
\\
A_2(x,t
) &= \frac{1}{\sqrt{1-\delta^2\varepsilon^2}}
\left(A_2\left(x-\frac{\varepsilon}{2},t-\frac{\varepsilon}{2}
\right)
+i\delta\varepsilon\, A_1\left(x-\frac{\varepsilon}{2},t-\frac{\varepsilon}{2}
\right)\right);
\end{aligned}
\end{equation*}
\item[axiom~3:]
$\sum_{(x,t)\in\varepsilon\mathbb{Z}^2
}
\left(\left|A_1(x,t
)\right|^2+\left|A_2(x,t
)\right|^2\right)<\infty.$
\end{description}

For each $k\in\{1,2\}$ and $(x,t)\in\varepsilon\mathbb{Z}^2$
define the \emph{lattice propagator} to be the limit
\begin{equation}\label{eq-def-anti-alg}
\widetilde{A}_k(x,t):=\widetilde{A}_k(x,t,m,\varepsilon)
:=\lim_{\delta\searrow0}
A_k(x,t,m,\varepsilon,\delta).
\end{equation}
\end{definition}

\blueit{As usual for axiomatic definitions, the price for conciseness is that even the existence of the defined object is not obvious. This is going to be our first theorem.}

\begin{theorem}[Consistency of the axioms and concordance to Feynman's model]
\label{th-well-defined}
The functions $A_k(x,t,m,\varepsilon,\delta)$ and the lattice propagator
$\widetilde{A}_k(x,t,m,\varepsilon)$ are well-defined, that is,
there exists a unique pair of functions satisfying axioms~1--3, and
limit~\eqref{eq-def-anti-alg} exists
for each 
$(x,t)\in\varepsilon \mathbb{Z}^2$ and $k\in\{1,2\}$.
For $(x+t)/\varepsilon+k$ even, the limit is real and given by
\begin{align*}
\widetilde{A}_1(x,t,m,\varepsilon)
&=a_1(x,|t|+\varepsilon,m,\varepsilon),
&&\text{for $(x+t)/\varepsilon$ odd,}\\
\widetilde{A}_2(x,t,m,\varepsilon)&=
\pm a_2(\pm x+\varepsilon,|t|+\varepsilon,m,\varepsilon),
&&\text{for $(x+t)/\varepsilon$ even,}
\end{align*}
where the minus signs are taken when $t<0$.
For $(x+t)/\varepsilon+k$ odd, limit~\eqref{eq-def-anti-alg} is purely imaginary.
\end{theorem}

Again 
real and imaginary parts appear on dual sublattices (cf.~\emph{staggered fermions} \cite[\S4.4]{Rothe-12}).

\begin{physicalinterpretation*} One interprets
\begin{equation}\label{eq-q}
Q\left(x,t,m,\varepsilon\right):=\frac{1}{2}\left|\widetilde{A}_1\left(x,t,m,\varepsilon\right)\right|^2+
\frac{1}{2}\left|\widetilde{A}_2\left(x,t,m,\varepsilon\right)\right|^2
\end{equation}
as \emph{minus the expected charge} in the interval of length $\varepsilon$ around the point $x$ at the time $t>0$, if an electron of mass $m$ was emitted from the origin at the time $0$ (or a positron is absorbed there). Unlike the original Feynman checkers, this value cannot be interpreted as probability \cite[\S9.2]{SU-22}:
virtual electrons and positrons also contribute to the charge.
\end{physicalinterpretation*}


\subsection{Exact solution}
\label{ssec-analytic}



Now we state a result which reduces investigation of the new model to analysis of certain integral. One can use it as an alternative definition of the model.
The integral coincides with the one arising in the original Feynman's model \cite[Proposition~12]{SU-22} but now is computed for arbitrary parity of $(x+t)/\varepsilon$.


\begin{proposition}[Fourier integral]\label{th-equivalence}
For each $m,\varepsilon>0$ and $(x,t)\in\varepsilon\mathbb{Z}^2$ we have
\begin{equation*}
  \begin{aligned}
  \widetilde{A}_1(x,t,m,\varepsilon)&=
  \pm\frac{im\varepsilon^2}{2\pi}
  \int_{-\pi/\varepsilon}^{\pi/\varepsilon}
  \frac{e^{i p x-i\omega_pt}\,dp}
  {\sqrt{m^2\varepsilon^2+\sin^2(p\varepsilon)}};\\
  \widetilde{A}_2(x,t,m,\varepsilon)&=
  \pm\frac{\varepsilon}{2\pi}\int_{-\pi/\varepsilon}^{\pi/\varepsilon}
  \left(1+
  \frac{\sin (p\varepsilon)} {\sqrt{m^2\varepsilon^2+\sin^2(p\varepsilon)}}\right) e^{ipx-i\omega_pt}\,dp,
  \end{aligned}
\end{equation*}
where the minus sign in the expression for $\widetilde{A}_k$ is taken for $t<0$ and $(x+t)/\varepsilon+k$ even, and
\begin{equation}\label{eq-omega}
\omega_p:=\frac{1}{\varepsilon}\arccos\left(\frac{\cos p\varepsilon}{\sqrt{1+m^2\varepsilon^2}}\right).
\end{equation}
\end{proposition}


\begin{physicalinterpretation*} Fourier integral represents a wave emitted by a point source as a superposition of waves of wavelength $2\pi/p$ and frequency~$\omega_p$.
\emph{Plank} and \emph{de Broglie relations} assert that $\omega_p$ and $p$ are the energy and the momentum \blueit{associated to} the waves.
As $\varepsilon\to 0$, 
the energy $\omega_p$ tends to the expression $\sqrt{m^2+p^2}$ (see~Lemma~\ref{l-g-difference}).
The latter expression is standard; it generalizes \emph{Einstein formula} $\hbar\omega_0=mc^2$
relating particle energy and mass (recall that $\hbar=c=1$ in our units).
Fourier integral resembles the \emph{spin-$1/2$ Feynman propagator}
describing \blueit{the} creation, annihilation and motion of non-interacting electrons and positrons along a line
in quantum field theory
\cite[(6.45),(6.50)--(6.51)]{Folland}:
\begin{equation}\label{eq-feynman-propagator-fourier}
G^F(x,t,m)=
\frac{1}{4\pi}
\int_{-\infty}^{+\infty}
\begin{pmatrix}
\frac{im}{\sqrt{m^2+p^2}} & 1+\frac{p}{\sqrt{m^2+p^2}} \\
-1+\frac{p}{\sqrt{m^2+p^2}} & \frac{im}{\sqrt{m^2+p^2}}
\end{pmatrix}
e^{i p x-i\sqrt{m^2+p^2}t}\,dp
\qquad\text{for }t>0.
\end{equation}
Here the integral is understood as Fourier transform of
matrix-valued tempered distributions. 
(We do not use~\eqref{eq-feynman-propagator-fourier} in this paper
and hence do not define those notions. Cf.~\eqref{eq-feynman-propagator})
\end{physicalinterpretation*}

\begin{example}[Massless and heavy particles]\label{p-massless}
\textup{(Cf.~\cite{Bender-etal-85})}
For each $(x,t)\in \varepsilon\mathbb{Z}^2$ we have
\begin{align}
\notag
\widetilde{A}_k(x,t,\infty,\varepsilon)&:=
\lim_{m \to +\infty}\widetilde{A}_k(x,t,m,\varepsilon)=
\delta_{x0}(-i)^{|t/\varepsilon+k-1|-1}, 
\\
\notag
\widetilde{A}_1(x,t,0,\varepsilon)&:=
\lim_{m \searrow 0}\widetilde{A}_1(x,t,m,\varepsilon)=0,\\
\label{eq-massless-lattice}
\widetilde{A}_2(x,t,0,\varepsilon)&:=
\lim_{m \searrow 0}\widetilde{A}_2(x,t,m,\varepsilon)=
\begin{cases}
1, &\text{for }x=t\ge 0,\\
-1, &\text{for }x=t<0,\\
{2i\varepsilon}/{\pi(x-t)},&\text{for }
{(x+t)}/{\varepsilon}\text{ odd},\\
0, &\text{otherwise.}
\end{cases}
\end{align}
Physically this means that an infinitely heavy particle
stays at the origin forever, and a massless particle forms
a charged ``cloud'' moving with the speed of light.
The massless lattice propagator is proportional to
the \emph{massless spin-$1/2$ Feynman
propagator} defined by
\begin{equation}\label{eq-massless}
G^{F}(x,t,0):=
\left(\begin{smallmatrix}
   0 & i/2\pi(x-t) \\
   i/2\pi(x+t) & 0
 \end{smallmatrix}\right)
\qquad\text{for }|x|\ne |t|.
\end{equation}
\end{example}

\begin{example}[Unit mass and lattice step]
\label{ex-simplest-values}
The value
\begin{align}\label{eq-Gauss}
\widetilde{A}_1(0,0,1,1)/i&=\Gamma(\tfrac{1}{4})^2/(2\pi)^{3/2}=
\tfrac{2}{\pi}K(i)=:G\approx 0.83463
\intertext{is the Gauss constant and}
\label{eq-lemniscate}
\widetilde{A}_2(1,0,1,1)/i&=2\sqrt{2\pi}/\Gamma(\tfrac{1}{4})^2
=\tfrac{2}{\pi}(E(i)-K(i))=1/\pi G=:L'\approx 0.38138
\end{align}
is the inverse lemniscate constant (cf.~\textup{\cite[\S6.1]{Finch-03}}), where $K(z):=\int_{0}^{\pi/2}d\theta/\sqrt{1-z^2\sin^2\theta}$ and
$E(z):=\int_{0}^{\pi/2}\sqrt{1-z^2\sin^2\theta}\,d\theta$ are the complete elliptic integrals of the 1st and 2nd kind respectively. The other values are
even more complicated irrationalities (see Table~\ref{table-a4}).
\end{example}

\begin{proposition}[Rational basis] \label{p-basis}
For each $k\in\{1,2\}$ and $(x,t)\in \mathbb{Z}^2$
the values $2^{|t|/2}\mathrm{Re}\,\widetilde{A}_k(x,t,1,1)$ are integers and
$2^{|t|/2}\mathrm{Im}\,\widetilde{A}_k(x,t,1,1)$ are rational linear combinations
of numbers~\eqref{eq-Gauss} and~\eqref{eq-lemniscate}.
\end{proposition}

\begin{table}[ht]
  \centering
  $\widetilde{A}_1(x,t,1,1)$ \\[0.1cm]
\begin{tabular}
{|c*{8}{|>{\centering\arraybackslash}p{55pt}}|}
\hline
$2$& $0$ & ${i\frac{2G-3L'}{3}}$& $\frac{1}{2}$ & ${-iL'}$& $\frac{1}{2}$ & ${i\frac{2G-3L'}{3}}$& $0$ \\
\hline
$1$ & ${i\frac{7G-15L'}{3\sqrt{2}}}$ & $0$ & ${i\frac{G-L'}{\sqrt{2}}}$& $\frac{1}{\sqrt{2}}$ & ${i\frac{G-L'}{\sqrt{2}}}$ & $0$ & ${i\frac{7G-15L'}{3\sqrt{2}}}$ \\
\hline
$0$& $0$ & ${i(G-2L')}$& $0$ & ${iG}$& $0$ & ${i(G-2L')}$& $0$ \\
\hline
$-1$& ${i\frac{7G-15L'}{3\sqrt{2}}}$ & $0$ & ${i\frac{G-L'}{\sqrt{2}}}$& $\frac{1}{\sqrt{2}}$ & ${i\frac{G-L'}{\sqrt{2}}}$& $0$ & ${i\frac{7G-15L'}{3\sqrt{2}}}$\\
\hline
$-2$ & $0$ & ${i\frac{2G-3L'}{3}}$& $\frac{1}{2}$ & ${-iL'}$& $\frac{1}{2}$ & ${i\frac{2G-3L'}{3}}$& $0$ \\
\hline
\diagbox[dir=SW,height=25pt]{$t$}{$x$}&$-3$&$-2$&$-1$&$0$&$1$&$2$&$3$\\
\hline
\end{tabular}
%
\\[0.2cm]
   $\widetilde{A}_2(x,t,1,1)$\\[0.1cm]
\begin{tabular}{|c*{8}{|>{\centering\arraybackslash}p{55pt}}|}
\hline
$2$& $i\frac{5G-12L'}{15}$ & $0$ & $-i\frac{G}{3}$ & $-\frac{1}{2}$ & $-iG$ & $\frac{1}{2}$ & $i\frac{-5G+12L'}{3}$ \\
\hline
$1$& $0$ & $i\frac{G-3L'}{3\sqrt{2}}$ & $0$ & $-i\frac{G+L'}{\sqrt{2}}$ & $\frac{1}{\sqrt{2}}$ & $i\frac{-G+3L'}{\sqrt{2}}$ & $0$ \\
\hline
$0$ & $i\frac{4G-9L'}{3}$ & $0$ & $-iL'$ & $1$ & $iL'$ & $0$ & $i\frac{-4G+9L'}{3}$ \\
\hline
$-1$& $0$ & $i\frac{G-3L'}{\sqrt{2}}$ & $-\frac{1}{\sqrt{2}}$ & $i\frac{G+L'}{\sqrt{2}}$ & $0$ & $i\frac{-G+3L'}{3\sqrt{2}}$ & $0$  \\
\hline
$-2$& $i\frac{5G-12L'}{3}$ & $-\frac{1}{2}$ & $iG$ & $\frac{1}{2}$ & $i\frac{G}{3}$ & $0$ &  $i\frac{-5G+12L'}{15}$\\
\hline
\diagbox[dir=SW,height=25pt]{$t$}{$x$}&$-3$&$-2$&$-1$&$0$&$1$&$2$&$3$\\
\hline
\end{tabular}
\qquad
  \caption{The values $\widetilde{A}_1(x,t,1,1)$ and $\widetilde{A}_2(x,t,1,1)$ for small $x,t$ (see Definition~\ref{def-anti-combi} and Example~\ref{ex-simplest-values})}
  \label{table-a4}
\end{table}

In general, the 
propagator is ``explicitly'' expressed through the Gauss hypergeometric function.
Denote by ${}_2F_{1}(p,q;r;z)$ the principal branch of the function
defined by \cite[9.111]{Gradstein-Ryzhik-63}.


\begin{proposition}[Exact solution] \label{p-mass5}
For each $m,\varepsilon>0$ and $(x,t)\in\varepsilon\mathbb{Z}^2$
we have
\begin{align*}
\widetilde{A}_1(x,t,m,\varepsilon)
&=\pm i
\frac{(-im\varepsilon)^{\frac{t-|x|}{\varepsilon}}}
{\left({1+m^2\varepsilon^2}\right)^{\frac{t}{2\varepsilon}}}
\binom{\frac{t+|x|}{2\varepsilon}-\frac{1}{2}}{|x|/\varepsilon}
{}_2F_1\left(\frac{1}{2} + \frac{|x|-t}{2\varepsilon}, \frac{1}{2} + \frac{|x|-t}{2\varepsilon}; 1+\frac{|x|}{\varepsilon}; -\frac{1}{m^2\varepsilon^2}\right),\\
\widetilde{A}_2(x,t,m,\varepsilon)
&=\pm
\frac{(-im\varepsilon)^{\frac{t-|x|}{\varepsilon}}}
{\left({1+m^2\varepsilon^2}\right)^{\frac{t}{2\varepsilon}}}
\binom{\frac{t+|x|}{2\varepsilon}-1+\theta(x)}{|x|/\varepsilon}
{}_2F_1\left(\frac{|x|-t}{2\varepsilon}, 1+\frac{|x|-t}{2\varepsilon}; 1+\frac{|x|}{\varepsilon}; -\frac{1}{m^2\varepsilon^2}\right),
\end{align*}
where the minus sign in the expression for $\widetilde{A}_k$ is taken for $t<0$ and  $(x+t)/\varepsilon+k$ even, and
\begin{align*}
\theta(x)&:=
\begin{cases}
                            1, & \mbox{if } x\ge0, \\
                            0, & \mbox{if } x<0;
\end{cases}
&
\binom{z}{n}&:=\frac{1}{n!}\prod_{j=1}^{n}(z-j+1).
\end{align*}
\end{proposition}

\begin{remark} \label{rem-jacobi-functions}
Depending on the parity of $(x+t)/\varepsilon$,
those expressions can be rewritten as \emph{Jacobi polynomials} (see~\cite[Remark~3]{SU-22}) or as
\emph{Jacobi functions of the 2nd kind} of half-integer order (see the definition in \cite[(4.61.1)]{Szego-39}). For instance, for each $(x,t)\in\varepsilon\mathbb{Z}^2$ such that $|x|>t$ and $(x+t)/\varepsilon$ is even we have
$$
\widetilde{A}_1(x,t,m,\varepsilon)
=\frac{2m\varepsilon}{\pi}
\left({1+m^2\varepsilon^2}\right)^{t/2\varepsilon}
Q_{(|x|-t-\varepsilon)/2\varepsilon}^{(0,t/\varepsilon)}(1+2m^2\varepsilon^2).
$$
\end{remark}

\begin{remark} \label{rem-probability-to-die} For even $x/\varepsilon$, the value $\widetilde{A}_1(x,0,m,\varepsilon)$ equals $i(1+\sqrt{1+m^2\varepsilon^2})/m\varepsilon$ times the probability that a planar simple random walk over the sublattice
$\{(x,t)\in\varepsilon\mathbb{Z}^2:(x+t)/\varepsilon\text{ even}\}$ dies at
$(x,0)$, if it starts at $(0,0)$ and dies with the probability $1-1/\sqrt{1+m^2\varepsilon^2}$ before each step.
\blueit{Such \emph{massive random walk} is archetypical in the Makarov--Smirnov program for studying off-critical models \cite{Makarov-Smirnov-09, Chelkak-Wan-21}. Thus we observe a Euclidean theory for $t=0$, which is natural because $t=0$ is both real and purely imaginary.}
Nothing like that is known \blueit{nor expected} for $t\ne 0$.
\end{remark}

\subsection{Asymptotic formulae}
\label{ssec-asymptotic}


The main result of this work is consistency of the model with continuum quantum field theory, that is, the convergence of the lattice propagator to the continuum one as the lattice becomes finer and finer (see Figure~\ref{fig-approx-b} and Theorem~\ref{th-continuum-limit}). More precisely, the former propagator converges to the real part of the latter on certain sublattice,
and it converges to the imaginary part on the dual sublattice.
The limit involves Bessel functions $J_n(z)$ and $Y_n(z)$ of the 1st and 2nd kind respectively, and 
modified Bessel functions $K_n(z)$ of the 2nd kind defined in
\cite[\S8.40]{Gradstein-Ryzhik-63}.

\begin{theorem}[Asymptotic formula in the continuum limit] \label{th-continuum-limit}
For each $m,\varepsilon>0$ and $(x,t)\in\varepsilon\mathbb{Z}^2$ such that
$|x|\ne |t|$
we have
\begin{align}
\label{eq1-th-continuum-limit}
\hspace{-0.2cm}\widetilde{A}_1\left(x,t,{m},{\varepsilon}\right)
&=
\begin{cases}
m\varepsilon \left(J_0(ms)+O\left(\varepsilon\Delta\right)\right),
&\text{for }|x|<|t|\text{ and }(x+t)/\varepsilon\text{ odd};\\
-im\varepsilon \left(Y_0(ms)+O\left(\varepsilon\Delta\right)\right),
&\text{for }|x|<|t|\text{ and }(x+t)/\varepsilon\text{ even};\\
0,
&\text{for }|x|>|t|\text{ and }(x+t)/\varepsilon\text{ odd};\\
\mathrlap{2im\varepsilon \left(K_0(ms)+O\left(\varepsilon\Delta\right)\right)/{\pi},}
\hphantom{2im\varepsilon(t+x)\left(K_1(ms)+O\left(\varepsilon\Delta\right)\right)/\pi s,}
&\text{for }|x|>|t|\text{ and }(x+t)/\varepsilon\text{ even};
\end{cases}
\end{align}
\begin{align}
\label{eq2-th-continuum-limit}
\hspace{-0.2cm}\widetilde{A}_2\left(x,t,{m},{\varepsilon}\right)
&=
\begin{cases}
-m\varepsilon(t+x)\left(J_1(ms)+O\left(\varepsilon\Delta\right)\right)/{s},
&\text{for }|x|<|t|\text{ and }(x+t)/\varepsilon\text{ even};\\
im\varepsilon(t+x)\left(Y_1(ms)+O\left(\varepsilon\Delta\right)\right)/{s},
&\text{for }|x|<|t|\text{ and }(x+t)/\varepsilon\text{ odd};\\
0,
&\text{for }|x|>|t|\text{ and }(x+t)/\varepsilon\text{ even};\\
2im\varepsilon(t+x)\left(K_1(ms)+O\left(\varepsilon\Delta\right)\right)/\pi s,
&\text{for }|x|>|t|\text{ and }(x+t)/\varepsilon\text{ odd}.
\end{cases}
\end{align}
$$
\text{where}\qquad
s:=\sqrt{\left|t^2-x^2\right|}\qquad\text{and}\qquad
\Delta:=\frac{1}{\left|\,|x|-|t|\,\right|}+m^2(|x|+|t|).
$$
\end{theorem}


Recall that 
$f(x,t,{m},{\varepsilon})=g(x,t,{m},{\varepsilon})+{O}\left(h(x,t,{m},{\varepsilon})\right)$ means that there is a constant $C$ (not depending on $x,t,{m},{\varepsilon}$) such that for each $x,t,{m},{\varepsilon}$ satisfying the assumptions of the theorem we have $|f(x,t,{m},{\varepsilon})-g(x,t,{m},{\varepsilon})|\le C\,h(x,t,{m},{\varepsilon})$.

\begin{physicalinterpretation*} The theorem means that in the continuum limit, the model reproduces the spin-$1/2$ Feynman propagator
(cf.~\eqref{eq-feynman-propagator-fourier})
\begin{equation}\label{eq-feynman-propagator}
\hspace{-0.2cm}G^F(x,t,m):=
\begin{cases}
\dfrac{m}{4}\,
\begin{pmatrix}
J_0(ms)-iY_0(ms) &
-{\frac{t+x}{s}}\left(J_1(ms)-iY_1(ms)\right) \\ {\frac{t-x}{s}}\left(J_1(ms)-iY_1(ms)\right) &
J_0(ms)-iY_0(ms)
\end{pmatrix}, &\text{if }|x|<|t|;\\
\dfrac{im}{2\pi}\,
\begin{pmatrix}
K_0(ms) & {\frac{t+x}{s}}\,K_1(ms)  \\
{\frac{x-t}{s}}\,K_1(ms) & K_0(ms)
\end{pmatrix}, &\text{if }|x|>|t|,
\end{cases}
\end{equation}
where again
$s:=\sqrt{\left|t^2-x^2\right|}$.
(A common definition involves also a generalized function supported on the lines $t=\pm x$ which we drop
because those lines are excluded anyway.)
The value $|G_{11}^F(x,t,m)|^2+|G_{12}^F(x,t,m)|^2$
is minus the expected charge density at the point $x$ at the moment~$t$.
Recall that Feynman's original model reproduces 
\emph{retarded propagator} \cite[Theorem~5]{SU-22}.
\end{physicalinterpretation*}

The asymptotic formulae in Theorem~\ref{th-continuum-limit} were known earlier (with a slightly weaker error estimates) on the sublattice, where the new model coincides with Feynman's original one \cite[Theorem~5]{SU-22}. 
Extension to the dual sublattice has required different methods.

In the following corollary, we approximate a point $(x,t)\in\mathbb{R}^2$
by the lattice point
\begin{equation}\label{eq-lattice-approx}
(x_\varepsilon,t_\varepsilon):=\left(2\varepsilon\!\left\lceil \frac{x}{2\varepsilon}\right\rceil,2\varepsilon\!\left\lceil \frac{t}{2\varepsilon}\right\rceil\right).
\end{equation}

\begin{corollary}[Uniform continuum limit; see Figures~\ref{fig-approx-q} and~\ref{fig-approx-b}] \label{cor-anti-uniform}  For each fixed $m\ge 0$ we have
\begin{gather} \notag
\begin{aligned}
  \frac{1}{4\varepsilon}\,\tilde{A}_1\left(x_\varepsilon+\varepsilon,t_\varepsilon,{m},{\varepsilon}\right)
&\rightrightarrows \mathrm{Re}\,G^F_{11}(x,t,m); &
  \frac{1}{4\varepsilon}\,\tilde{A}_1\left(x_\varepsilon,t_\varepsilon,{m},{\varepsilon}\right)
&\rightrightarrows i\mathrm{Im}\,G^F_{11}(x,t,m);\\
  \frac{1}{4\varepsilon}\,\tilde{A}_2\left(x_\varepsilon+\varepsilon,t_\varepsilon,{m},{\varepsilon}\right)
&\rightrightarrows i\mathrm{Im}\,G^F_{12}(x,t,m); &
  \frac{1}{4\varepsilon}\,\tilde{A}_2\left(x_\varepsilon,t_\varepsilon,{m},{\varepsilon}\right)
&\rightrightarrows \mathrm{Re}\,G^F_{12}(x,t,m);
\end{aligned}\\
\label{eq-cor-anti-uniform}
\frac{1}{8\varepsilon^2}\,\left(Q\left(x_\varepsilon,t_\varepsilon,{m},{\varepsilon}\right)
+Q\left(x_\varepsilon+\varepsilon,t_\varepsilon,{m},{\varepsilon}\right)\right)
\rightrightarrows |G^F_{11}(x,t,m)|^2+|G^F_{12}(x,t,m)|^2
\end{gather}
as $\varepsilon\to 0$ uniformly on compact subsets of $\mathbb{R}^2\setminus\{|t|=|x|\}$, 
under notation~\eqref{eq-lattice-approx},\eqref{eq-feynman-propagator},\eqref{eq-q},\eqref{eq-massless-lattice},\eqref{eq-massless}.
\end{corollary}


\begin{figure}[htbp]
  \centering
\includegraphics[width=0.36\textwidth]{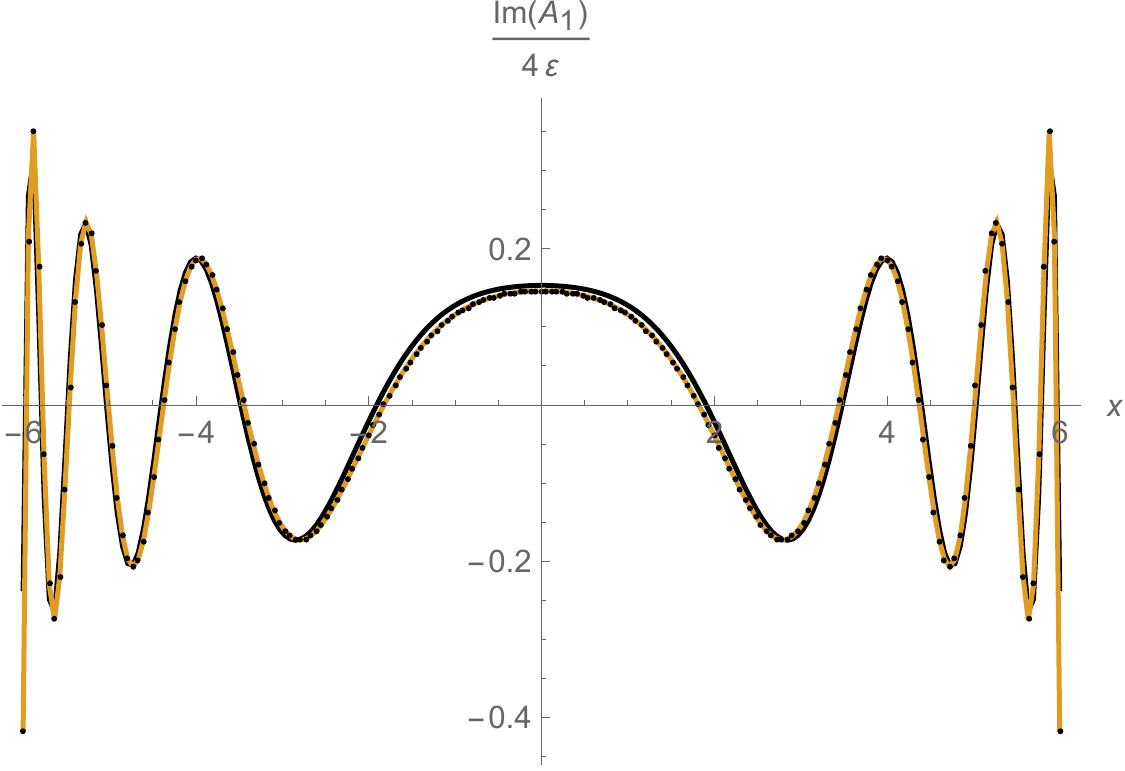}
\includegraphics[width=0.36\textwidth]{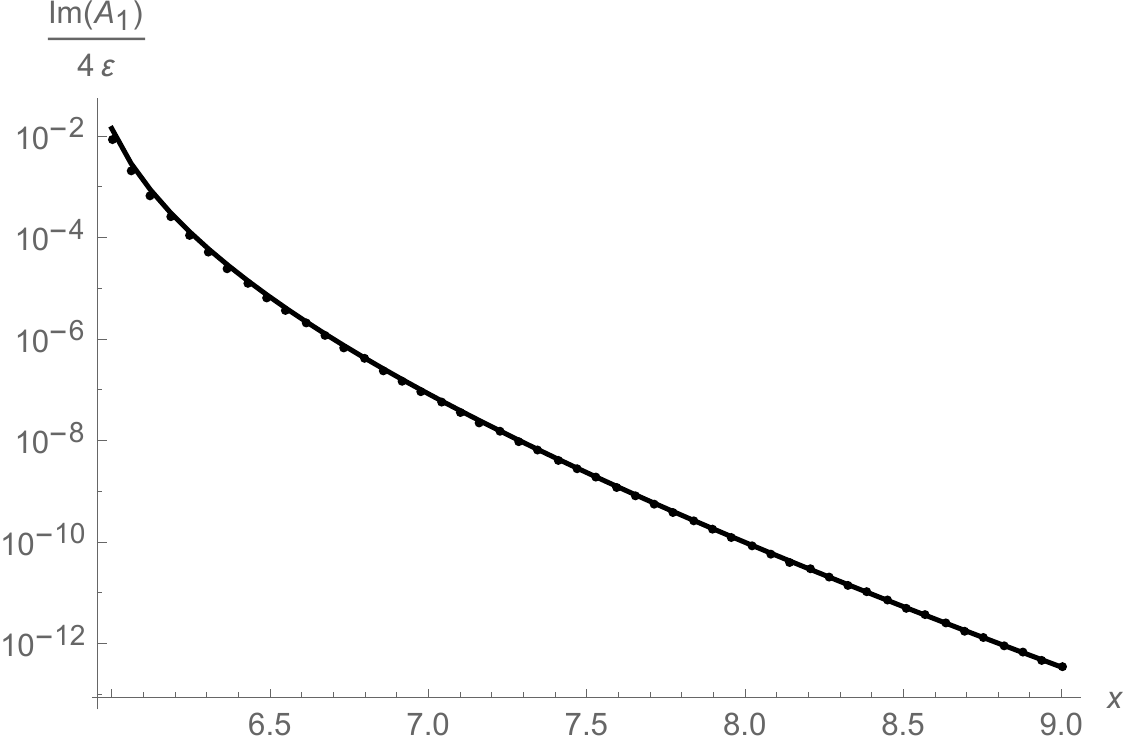}
\\
\includegraphics[width=0.36\textwidth]{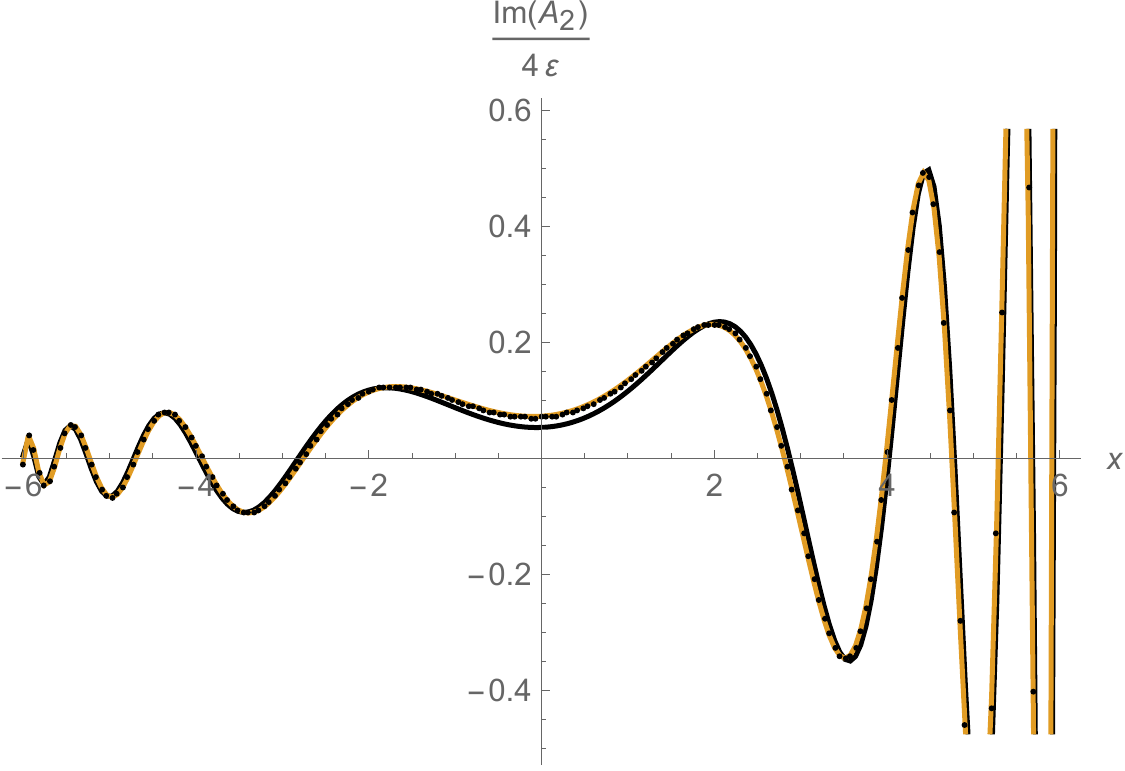}
\includegraphics[width=0.36\textwidth]{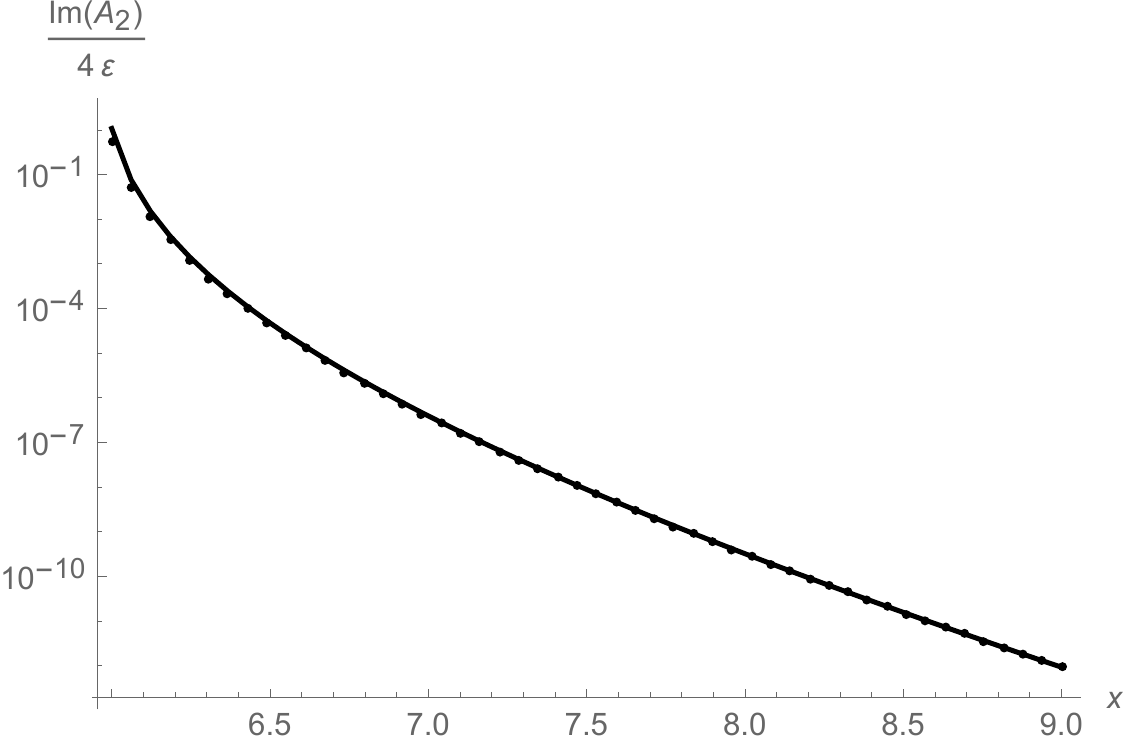}
  \caption{Plots of (the normalized imaginary part of)
  the lattice propagators
  $\mathrm{Im}\,\widetilde{A}_1(x,6,4,0.03)/0.12$ (top, dots) and
  $\mathrm{Im}\,\widetilde{A}_2(x,6,4,0.03)/0.12$ (bottom, dots)
  for $x/0.06\in\mathbb{Z}$ and $x/0.06+1/2\in\mathbb{Z}$ respectively, their
  analytic approximations from Theorems~\ref{th-continuum-limit} (dark curve) and~\tobereplaced{\ref{th-anti-ergenium}}{\ref{th-anti-Airy}} (light curve). The former approximation is the imaginary part of the spin-$1/2$ Feynman propagator
  $\mathrm{Im}\,G^F_{11}(x,6,4)$ (top, dark) and $\mathrm{Im}\,G^F_{12}(x,6,4)$ (bottom, dark) given by~\eqref{eq-feynman-propagator}.} 
  \label{fig-approx-b}
\end{figure}

The following result follows from
Proposition~\ref{th-equivalence} and~\cite[Theorems~2 and~7]{SU-22}; cf.~\cite{Drmota-etal, Ozhegov-21, Zakorko-21}.

\begin{theorem}
[Large-time asymptotic formula between the peaks; see Figure~\ref{fig-approx-b}]
\label{th-anti-ergenium}
For each $\Delta>0$ there is $C_\Delta>0$ such that for each $m,\varepsilon>0$ and each $(x,t)\in\varepsilon\mathbb{Z}^2$ satisfying
\begin{equation}\label{eq-case-A}
 |x|/t<1/\sqrt{1+m^2\varepsilon^2}-\Delta, \qquad
 \varepsilon\le 1/m, \qquad t>C_\Delta/m,
\end{equation}
we have
\begin{align*}
\widetilde{A}_1\left(x,t,m,\varepsilon\right)
&=
\begin{cases}
{\varepsilon}\sqrt{\frac{2m}{\pi}}
\frac{\sin (\pi/4-\bluevar{t\mathcal{L}}(x/t,m,\varepsilon))}
{\left(t^2-(1+m^2\varepsilon^2)x^2\right)^{1/4}}
+O_\Delta\left(\frac{\varepsilon}{m^{1/2}t^{3/2}}\right),
&\text{for }(x+t)/\varepsilon\text{ odd};\\
\mathrlap{i{\varepsilon}\sqrt{\frac{2m}{\pi}}
\frac{\cos (\pi/4-\bluevar{t\mathcal{L}}(x/t,m,\varepsilon))}
{\left(t^2-(1+m^2\varepsilon^2)x^2\right)^{1/4}}
+O_\Delta\left(\frac{\varepsilon}{m^{1/2}t^{3/2}}\right),}
\hphantom{-i{\varepsilon}\sqrt{\frac{2m}{\pi}}
\sqrt{\frac{t+x}{t-x}}\frac{\sin (\pi/4-t\mathcal{L}(x/t,m,\varepsilon))}
{\left(t^2-(1+m^2\varepsilon^2)x^2\right)^{1/4}}
+O_\Delta\left(\frac{\varepsilon}{m^{1/2}t^{3/2}}\right),}
&\text{for }(x+t)/\varepsilon\text{ even};
\end{cases}
\\
\widetilde{A}_2\left(x,t,m,\varepsilon\right)
&=
\begin{cases}
{\varepsilon}\sqrt{\frac{2m}{\pi}}
\sqrt{\frac{t+x}{t-x}}\frac{\cos (\pi/4-\bluevar{t\mathcal{L}}(x/t,m,\varepsilon))}
{\left(t^2-(1+m^2\varepsilon^2)x^2\right)^{1/4}}
+O_\Delta\left(\frac{\varepsilon}{m^{1/2}t^{3/2}}\right),
&\text{for }(x+t)/\varepsilon\text{ even};\\
-i{\varepsilon}\sqrt{\frac{2m}{\pi}}
\sqrt{\frac{t+x}{t-x}}\frac{\sin (\pi/4-\bluevar{t\mathcal{L}}(x/t,m,\varepsilon))}
{\left(t^2-(1+m^2\varepsilon^2)x^2\right)^{1/4}}
+O_\Delta\left(\frac{\varepsilon}{m^{1/2}t^{3/2}}\right),
&\text{for }(x+t)/\varepsilon\text{ odd},
\end{cases}
\intertext{where}
\bluevar{\mathcal{L}(v,m,\varepsilon)}&\bluevar{:=
-\frac{1}{\varepsilon}\arcsin
\frac{m\varepsilon} {\sqrt{\left(1+m^2\varepsilon^2\right)\left(1-v^2\right)}}
+\frac{v}{\varepsilon}\arcsin
\frac{m\varepsilon v}{\sqrt{1-v^2}}.}
\end{align*}
\end{theorem}

Here the notation $f(x,t,m,\varepsilon)
=g(x,t,m,\varepsilon)+{O}_{\Delta}\left(h(x,t,m,\varepsilon)\right)$
means that there is a constant $C(\Delta)$ (depending on $\Delta$ but \emph{not} on $x,t,m,\varepsilon$) such that for each $x,t,m,\varepsilon$ satisfying the assumptions of the theorem we have $|f(x,t,m,\varepsilon)-g(x,t,m,\varepsilon)|\le C(\Delta)\,h(x,t,m,\varepsilon)$.


\tobeadded

The following result is an immediate consequence of Proposition~\ref{th-equivalence} 
and \cite[Remark after Theorem~3]{Zakorko-21}; cf.~\cite{Ozhegov-21}. \mscomm{!!! Update references and theorem numbers !!!}


\begin{theorem}[Large-time asymptotic formula]  \label{th-anti-Airy} \textup{(See \cite[Remark after Theorem~3]{Zakorko-21}.)}
For each $m,\varepsilon>0$ and each $(x,t)\in\varepsilon\mathbb{Z}^2$ satisfying
$0<|x/t|<1/\sqrt{1+m^2\varepsilon^2}$
we have
\begin{align*}
\tilde{A}_1\left(x,t,m,\varepsilon\right)
&=
\frac{i^{\frac{|x|-|t|+\varepsilon}{\varepsilon}}\varepsilon\sqrt{m}\,\theta(x,t,m,\varepsilon)^{1/2}}
{\sqrt{3}(t^2-(1+m^2\varepsilon^2)x^2)^{1/4}}
\left(J_{ 1/3}\left(\theta(x,t,m,\varepsilon)\right)+
      J_{-1/3}\left(\theta(x,t,m,\varepsilon)\right)\right)+
      O_{m,\varepsilon}\left(\frac{1}{|t|}\right),\\
\tilde{A}_2\left(x,t,m,\varepsilon \right)
&= \sqrt{\frac{t + x}{t - x}}
\frac{i^{\frac{|x|-|t|}{\varepsilon}}\varepsilon\sqrt{m}\,\theta(x,t,m,\varepsilon)^{1/2}}
{\sqrt{3}(t^2-(1+m^2\varepsilon^2)x^2)^{1/4}}
\left(J_{ 1/3}\left(\theta(x,t,m,\varepsilon)\right)+
      J_{-1/3}\left(\theta(x,t,m,\varepsilon)\right)\right)+
      O_{m,\varepsilon}\left(\frac{1}{|t|}\right),\\
\intertext{where}
\theta(x,t,m,\varepsilon)&:=
 \frac{t}{\varepsilon}\arctan\frac{\sqrt{t^2-(1+m^2\varepsilon^2)x^2}}{m \varepsilon t}
-\frac{x}{\varepsilon}\arctan\frac{\sqrt{t^2-(1+m^2\varepsilon^2)x^2}}{m \varepsilon x}.
\end{align*}
\end{theorem}

Here the notation $f(x,t,m,\varepsilon)
=g(x,t,m,\varepsilon)+{O}_{m,\varepsilon}\left(h(x,t)\right)$
means that there is a constant $C(m,\varepsilon)$ (depending on $m,\varepsilon$ but \emph{not} on $x,t$) such that for each $x,t,m,\varepsilon$ satisfying the assumptions of the theorem we have $|f(x,t,m,\varepsilon)-g(x,t,m,\varepsilon)|\le C(m,\varepsilon)\,h(x,t)$.

\endtobeadded

\begin{physicalinterpretation*}
\blueit{One interpretes $\mathcal{L}(v,m,\varepsilon)$ as the \emph{Lagrangian}.
As $\varepsilon\to 0$, it tends to the Lagrangian
$-m\sqrt{1-v^2}$ of a free relativistic particle. For each $\varepsilon>0$,}
the well-known relation between energy~\eqref{eq-omega}, momentum $p$, and Lagrangian $\mathcal{L}$ holds:
$\omega_p=pv-\mathcal{L}$
for $p=\partial\mathcal{L}/\partial v$.
\end{physicalinterpretation*}

\subsection{Identities}
\label{ssec-identities}


Now we establish the following informal assertion
(see Propositions~\ref{p-mass}--\ref{p-triple} for formal ones).

\begin{consistencyprinciple*}
The new model satisfies 
the same identities as 
Feynman's one.
\end{consistencyprinciple*}

In particular, the identities 
in this subsection are known (and easy to deduce from
Definition~\ref{def-mass}) for the sublattice,
where the new model coincides with the original one \cite[Propositions~5--10]{SU-22}.
For the dual sublattice, these results are not so easy to prove. Further, for the former sublattice,
there are a few ``exceptions'' to the identities;
but on the dual lattice, the imaginary part
$b_k(x,t,m,\varepsilon)$ defined in \cite[Definition~5]{SU-22}
satisfies known identities \cite[Propositions~5--8 and 10]{SU-22}
literally for both $t>0$ and $t\le 0$.

In what follows we fix $m,\varepsilon>0$ and
omit the arguments $m,\varepsilon$ of the propagators.

\begin{proposition}[Dirac equation]\label{p-mass}
For each $(x,t)\in \varepsilon\mathbb{Z}^2$
we have
\begin{align}\label{eq-Dirac-source1}
\tilde A_1(x,t
) &= \frac{1}{\sqrt{1+m^2\varepsilon^2}}
(\tilde A_1(x+\varepsilon,t-\varepsilon
)
+ m \varepsilon\, \tilde A_2(x,t-\varepsilon
)),\\
\label{eq-Dirac-source2}
\tilde A_2(x,t
) &= \frac{1}{\sqrt{1+m^2\varepsilon^2}}
(\tilde A_2(x-\varepsilon,t-\varepsilon
)
- m \varepsilon\, \tilde A_1(x,t-\varepsilon
))+2\delta_{x0}\delta_{t0}.
\end{align}
\end{proposition}

This is slightly different from \cite[Eq.~(9)]{Arrighi-et-al},
where the coefficients are $\cos(m\varepsilon)$ and $\sin(m\varepsilon)$.

In the limit $\varepsilon\to 0$, this reproduces the \emph{Dirac equation in $1$ space- and $1$ time-dimension}
$$\begin{pmatrix}
m  & \partial/\partial x-\partial/\partial t \\
\partial/\partial x+\partial/\partial t & m
\end{pmatrix}
\begin{pmatrix}
G^F_{12}(x,t
) \\ 
G^F_{11}(x,t
)
\end{pmatrix}
=
\begin{pmatrix}
0\\
\delta(x)\delta(t)
\end{pmatrix}
.$$

\begin{proposition}[Klein--Gordon equation] \label{p-Klein-Gordon-mass} For each $k\in\{1,2\}$ and each $(x,t)\in \varepsilon\mathbb{Z}^2$, where $(x,t)\ne(0,0)$ for $k=1$ and $(x,t)\ne(-\varepsilon,0),(0,-\varepsilon)$ for $k=2$, we have
\begin{align*}
\sqrt{1+m^2\varepsilon^2}\,\widetilde{A}_k(x,t+\varepsilon
)
+\sqrt{1+m^2\varepsilon^2}\,\widetilde{A}_k(x,t-\varepsilon
)
-\widetilde{A}_k(x+\varepsilon,t
)
-\widetilde{A}_k(x-\varepsilon,t
)=0.
\end{align*}
\end{proposition}

In the limit $\varepsilon\to 0$, this gives the \emph{Klein--Gordon equation}
$\left(\tfrac{\partial^2}{\partial t^2}-\tfrac{\partial^2}{\partial x^2}+m^2\right)G^F_{1k}(x,t)=0$.

The infinite-lattice propagator has the same reflection symmetries
as the continuum one.

\begin{proposition}[Skew-symmetry]\label{p-symmetry-mass} 
For each $(x,t)\in \varepsilon\mathbb{Z}^2$, where $(x,t)\ne(0,0)$, we have
\begin{gather*}
  \widetilde{A}_1(x,t
  )=\widetilde{A}_1(-x,t
  )=\widetilde{A}_1(x,-t
  )=\widetilde{A}_1(-x,-t
  ),\\
  \widetilde{A}_2(x,t
  )=-\widetilde{A}_2(-x,-t
  ),
  \qquad\qquad
  (t-x)\,\widetilde{A}_2(x,t
  )
  =(t+x)\,\widetilde{A}_2(-x,t
  ).
\end{gather*}
\end{proposition}

\begin{proposition}[Charge conservation] \label{p-mass2}
For each $t\in\varepsilon\mathbb{Z}$,
$\sum\limits_{x\in\varepsilon\mathbb{Z}}
\dfrac{|\widetilde{A}_1\left(x,t
\right)|^2+
|\widetilde{A}_2\left(x,t
\right)|^2}{2}=1$.
\end{proposition}

There are two versions of Huygens' principle
(cf.~\cite[Proposition~9]{SU-22}).

\begin{proposition}[Huygens' principle]\label{p-Huygens2}
For each $x,t,t'\in\varepsilon\mathbb{Z}$, where $t\ge t'\ge 0$, we have
\begin{align*}
\widetilde{A}_1(x,t)
&=\frac{1}{2}\sum\limits_{\substack{x'\in\varepsilon\mathbb{Z}}}
\left(\widetilde{A}_2(x',t')\widetilde{A}_1(x-x',t-t')
+\widetilde{A}_1(x',t')\widetilde{A}_2(x'-x,t-t')\right),\\
\widetilde{A}_2(x,t)
&=\frac{1}{2}\sum\limits_{\substack{x'\in\varepsilon\mathbb{Z}}}
\left(
\widetilde{A}_2(x',t')\widetilde{A}_2(x-x',t-t')
-\widetilde{A}_1(x',t')\widetilde{A}_1(x'-x,t-t')\right).
\end{align*}
\end{proposition}

In the following version of Huygens' principle,
there are finitely many nonzero summands.

\begin{proposition}[Huygens' principle]\label{p-Huygens}
For each $x,t,t'\in\varepsilon\mathbb{Z}$, where $t\ge t'\ge 0$, we have
\begin{align*}
\widetilde{A}_1(x,t)
&=\sum\limits_{\substack{x'\in\varepsilon\mathbb{Z}:\\
(x+x'+t+t')/\varepsilon \text{ odd}}}
\widetilde{A}_2(x',t')\widetilde{A}_1(x-x',t-t')
+\sum\limits_{\substack{x'\in\varepsilon\mathbb{Z}:\\
(x+x'+t+t')/\varepsilon \text{ even}}}
\widetilde{A}_1(x',t')\widetilde{A}_2(x'-x,t-t'),\\
\widetilde{A}_2(x,t)
&= \sum\limits_{\substack{x'\in\varepsilon\mathbb{Z}:\\
(x+x'+t+t')/\varepsilon \text{ even}}}
\widetilde{A}_2(x',t')\widetilde{A}_2(x-x',t-t')
 - \sum\limits_{\substack{x'\in\varepsilon\mathbb{Z}:\\
(x+x'+t+t')/\varepsilon \text{ odd}}}
 \widetilde{A}_1(x',t')\widetilde{A}_1(x'-x,t-t').
\end{align*}
\end{proposition}




\begin{proposition}[Equal-time mixed recurrence]\label{l-mean}
For each $(x,t)\in \varepsilon\mathbb{Z}^2$ we have
\begin{align}\label{eq-mean1}
2m\varepsilon x\widetilde{A}_1(x,t)
&=(x-t-\varepsilon)\widetilde{A}_2(x-\varepsilon,t)
 -(x-t+\varepsilon)\widetilde{A}_2(x+\varepsilon,t),
\\ \label{eq-mean2}
2m\varepsilon x\widetilde{A}_2(x,t)
&=(x+t)\widetilde{A}_1(x-\varepsilon,t)-(x+t)\widetilde{A}_1(x+\varepsilon,t).
\end{align}
\end{proposition}

In particular, $\widetilde{A}_2(x,0)
=\left(\widetilde{A}_1(x-\varepsilon,0)-\widetilde{A}_1(x+\varepsilon,0)\right)/2m\varepsilon$
for $x\ne 0$.

\begin{proposition}[Equal-time recurrence relation]\label{p-triple} For each $(x,t)\in\varepsilon\mathbb{Z}^2$ we have
\begin{multline*}
(x+\varepsilon)((x-\varepsilon)^2-t^2)
\widetilde{A}_1(x-2\varepsilon,t
) +(x-\varepsilon)((x+\varepsilon)^2-t^2)
\widetilde{A}_1(x+2\varepsilon,t
)=
\\
=2x \left((1+2m^2\varepsilon^2)(x^2-\varepsilon^2)-t^2\right)
\widetilde{A}_1(x,t
),
\end{multline*}
\vspace{-1.0cm}
\begin{multline*}
(x+\varepsilon)((x-\varepsilon)^2-(t+\varepsilon)^2)
\widetilde{A}_2(x-2\varepsilon,t
) +(x-\varepsilon)((x+\varepsilon)^2-(t-\varepsilon)^2)
\widetilde{A}_2(x+2\varepsilon,t
)=
\\
=2x\left((1+2m^2\varepsilon^2)(x^2-\varepsilon^2)-t^2+\varepsilon^2\right)
\widetilde{A}_2(x,t
).
\end{multline*}
\end{proposition}

See \cite{Kuyanov-Slizkov-22} for an application.
Analogous identities can be written for any $3$ neighboring lattice points
by means of Proposition~\ref{p-mass5} and Gauss contiguous relations \cite[9.137]{Gradstein-Ryzhik-63}.


\begin{figure}[htbp]
  \centering
\includegraphics[height=2.2cm]{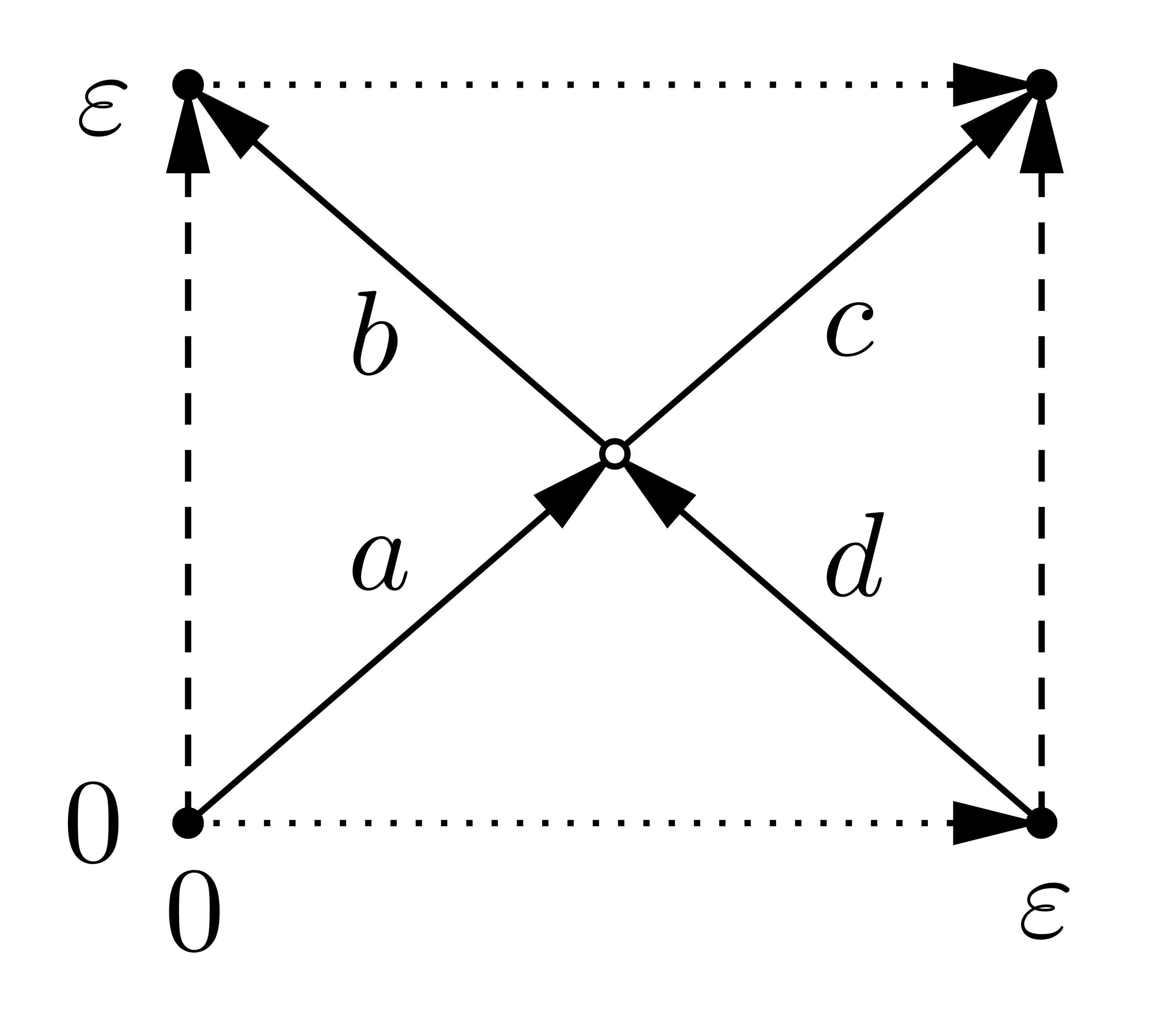}
\includegraphics[height=2.3cm]{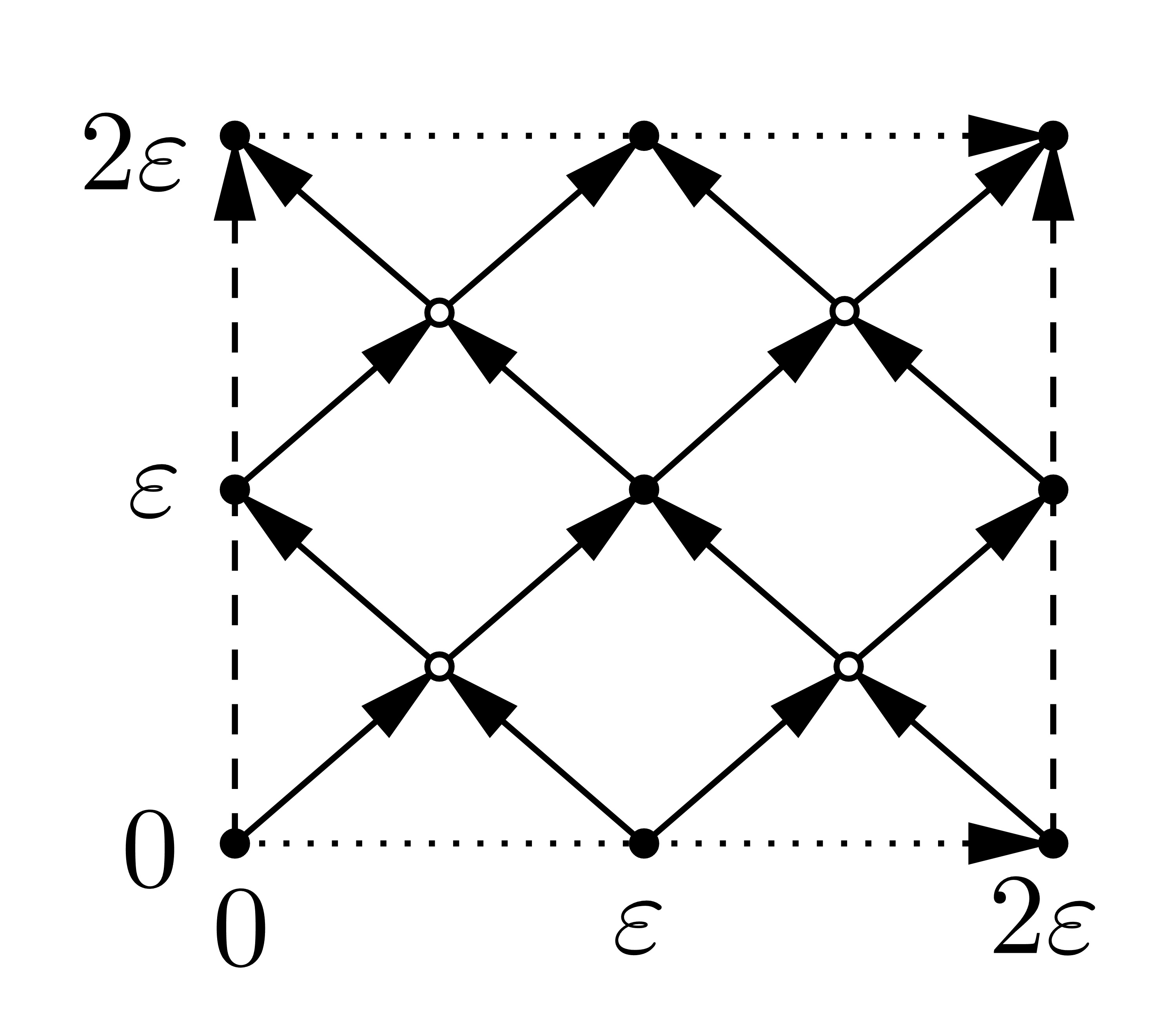}
\includegraphics[height=2.09cm]{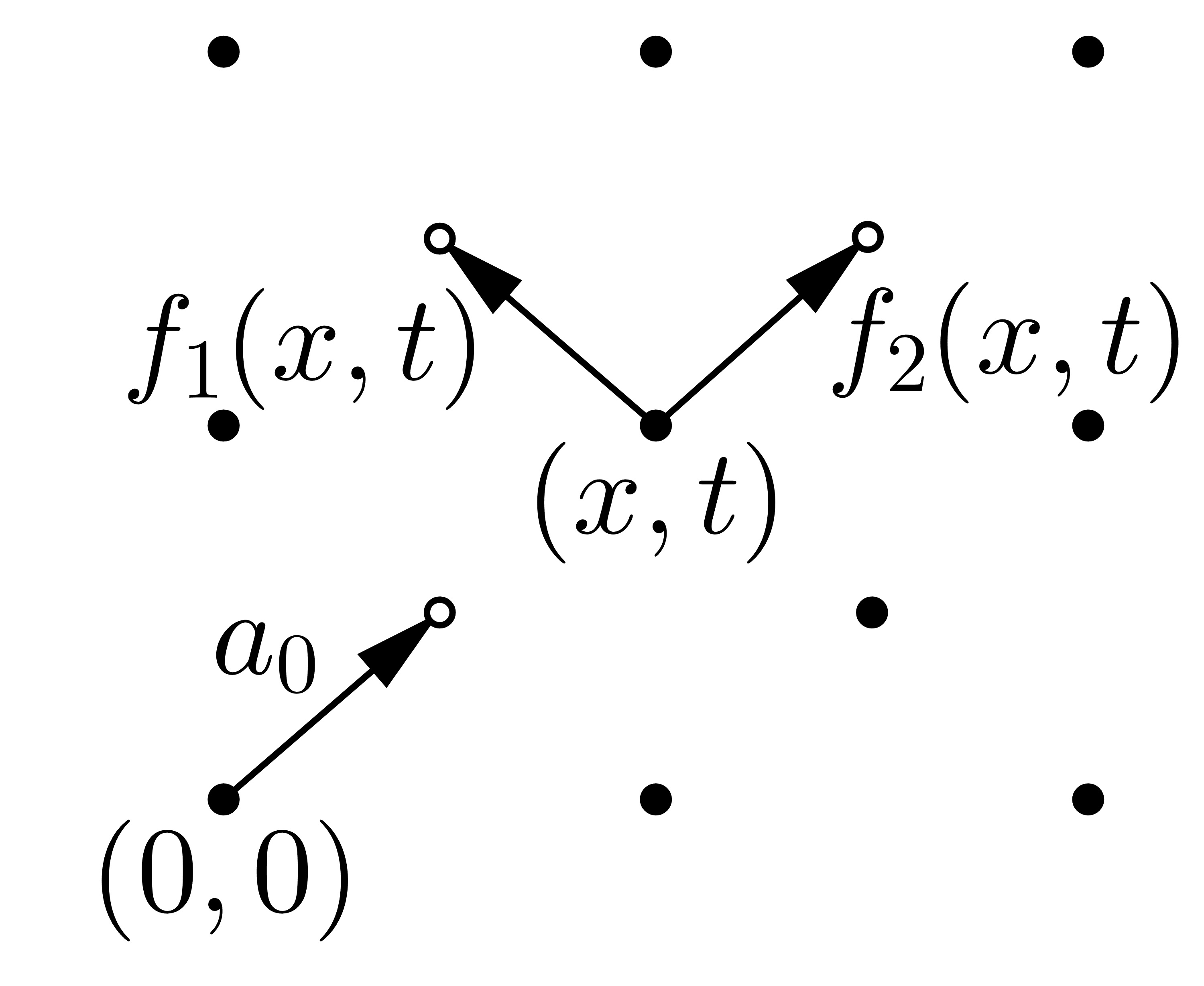}
\includegraphics[height=2.09cm]{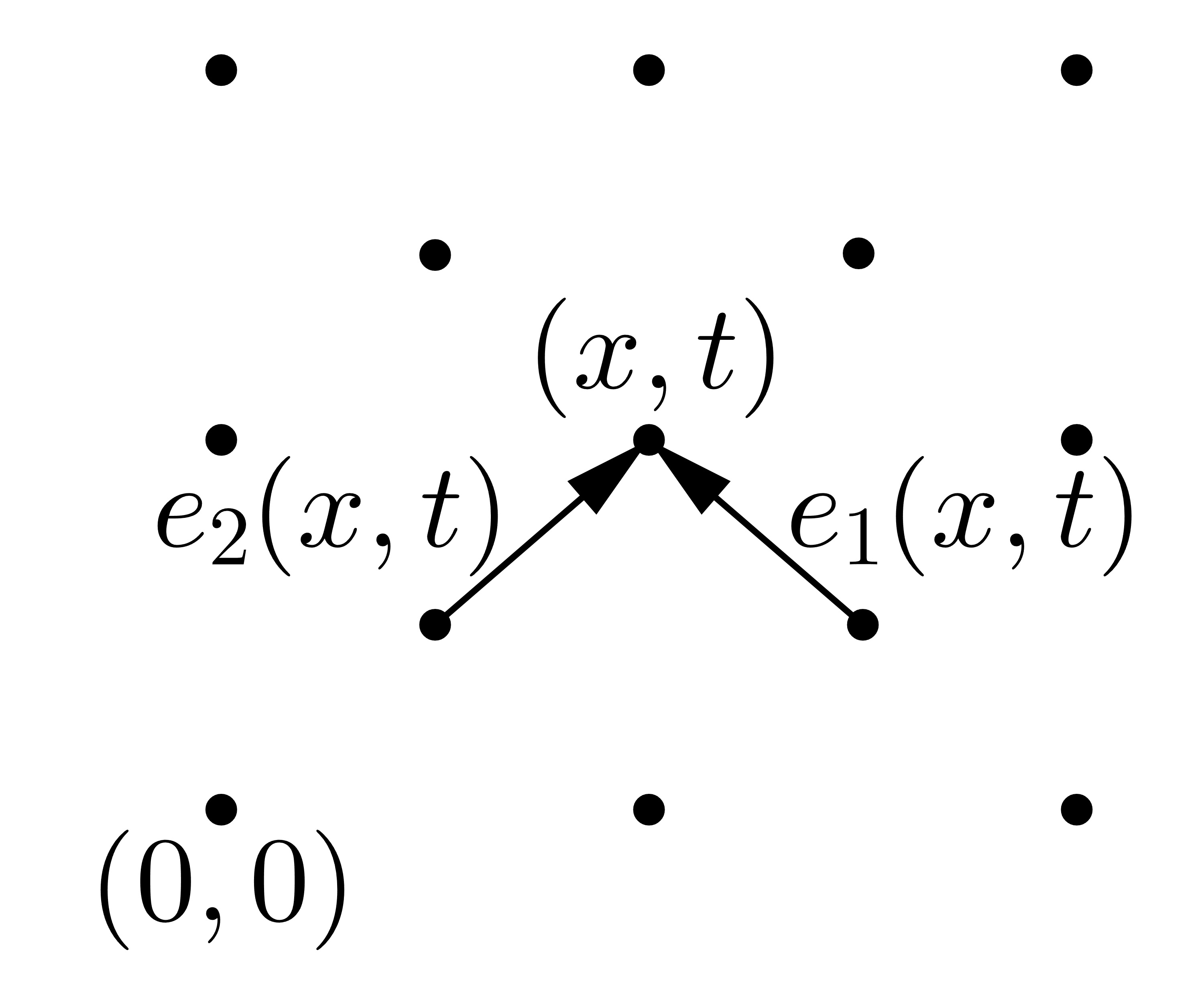}
\includegraphics[height=2.1cm]{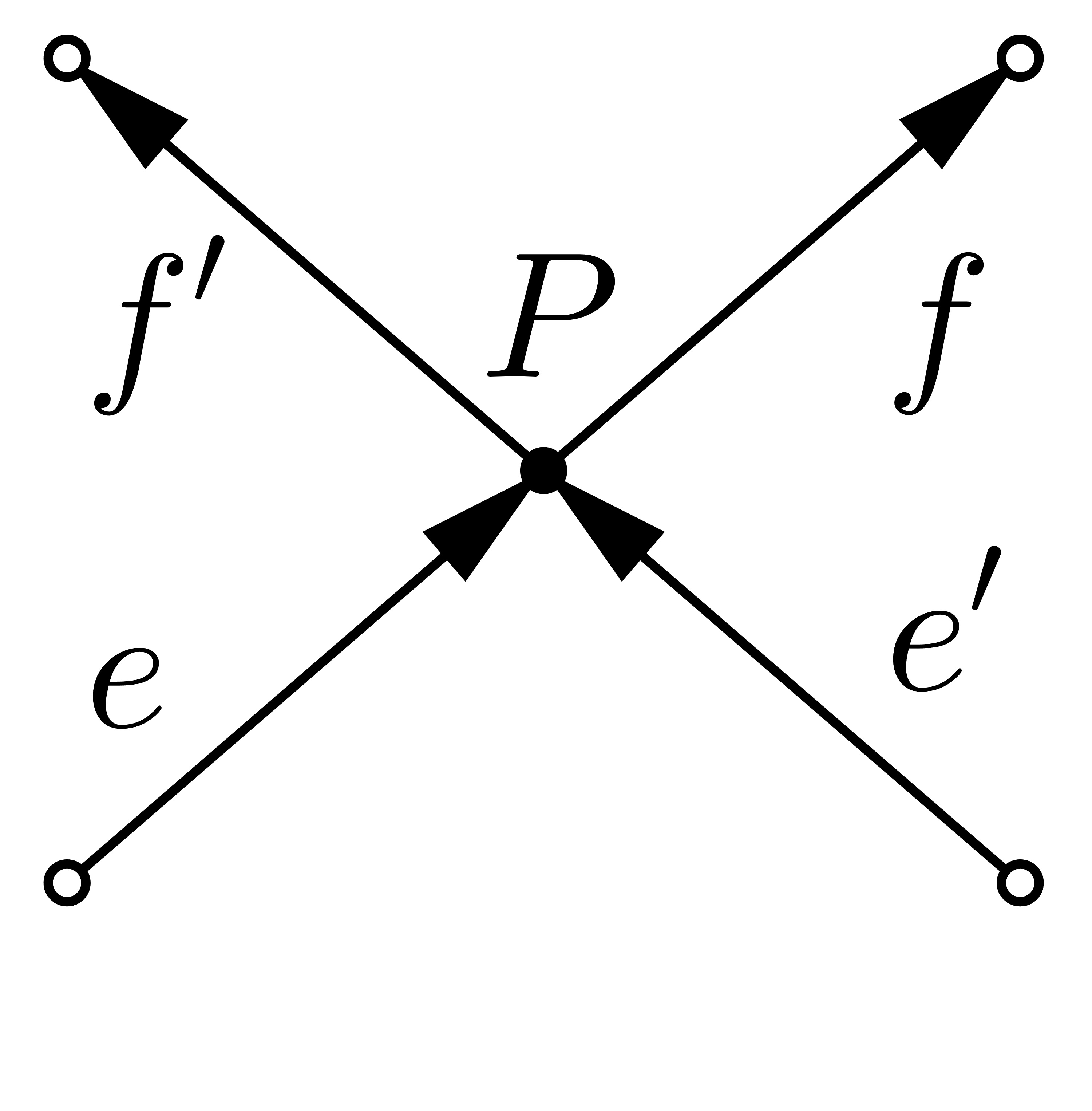}
\caption{Lattices of sizes $1$ and $2$ (left); see Example~\ref{ex-1x1}.
Notation for edges
(right).
}
  \label{fig-1x1}
\end{figure}

\subsection{Combinatorial definition}
\label{ssec-combi-def}

Now we realize the plan from the end of~\S\ref{ssec-outline}, but switch
the role of the lattice and its dual. \blueit{Compared to \cite[Definition~6]{SU-22},
we changed the notation $\delta$ to $\delta\varepsilon$ for symmetry; $\delta\searrow 0$ anyway.
}


\begin{definition} \label{def-anti-combi}
Fix $T\in\mathbb{Z}$ and $\varepsilon,m,\delta>0$ called \emph{lattice size, lattice step, particle mass}, and \emph{small imaginary mass} respectively. Assume $T>0$ and $\delta\varepsilon<1$. The \emph{lattice} is the quotient set
$$
\faktor{
\{\,(x,t)\in[0,T\varepsilon]^2:
2x/\varepsilon,2t/\varepsilon,(x+t)/\varepsilon\in\mathbb{Z}\,\}
}
{
\forall x,t:(x,0)\sim (x,T\varepsilon)\,\&\,(0,t)\sim (T\varepsilon,t).
}
$$
(This is a finite subset of the torus obtained from the square $[0,T\varepsilon]^2$ by an identification of the opposite sides; see Figure~\ref{fig-1x1} to the left.)
A lattice point $(x,t)$ is \emph{even} (respectively, \emph{odd}), if $2x/\varepsilon$ is even (respectively, odd).
An \emph{edge} is a vector starting from a lattice point $(x,t)$ and ending at the lattice point $(x+\varepsilon/2,t+\varepsilon/2)$ or $(x-\varepsilon/2,t+\varepsilon/2)$.

A \emph{generalized checker path} (or just a \emph{path})
is a finite sequence of distinct edges such that the endpoint of each edge is the starting point of the next one. A \emph{cycle} is defined analogously, only the sequence has the unique repetition: the first and the last edges coincide, and there is at least one edge in between. (In particular, a 
path such that the endpoint of the last edge is the starting point of the first one is \emph{not} yet a cycle; coincidence of the first and the last \emph{edges} is required. The first and the last edges of a 
\emph{path} coincide only if the path has a single edge. Thus in our setup, a path is \emph{never} a cycle.) \emph{Changing the starting edge} of a cycle means removal of the first edge from the sequence, then a cyclic permutation, and then adding the last edge of the resulting sequence at the beginning. A \emph{loop} is a cycle up to changing of the starting edge.

A \emph{node} of a path or loop $s$ is {an ordered} pair of consecutive edges in $s$ {(the order of the edges in the pair is the same as in~$s$)}. 
A \emph{turn} is a node such that the two edges are orthogonal. A node or turn is \emph{even} (respectively, \emph{odd}), if the {endpoint of the first edge in the pair}
is even (respectively, odd). Denote by $\mathrm{eventurns}(s)$, $\mathrm{oddturns}(s)$, $\mathrm{evennodes}(s)$, $\mathrm{oddnodes}(s)$ the number of even and odd turns and nodes in $s$. The \emph{arrow} (or \emph{weight)} of $s$ is
\begin{equation}\label{eq-def-anti3}
{A}(s):={A}(s,m,\varepsilon,\delta)
:=\pm\frac{(-im\varepsilon)^{\mathrm{oddturns}(s)}
(-\delta\varepsilon)^{\mathrm{eventurns}(s)}}
{(1+m^2\varepsilon^2)^{\mathrm{oddnodes}(s)/2}
(1-\delta^2\varepsilon^2)^{\mathrm{evennodes}(s)/2}},
\end{equation}
where the overall minus sign is taken when $s$ is a loop \blueit{(``fermionic minus sign'' \cite[Appendix~C]{Zee-10})}.

A set of checker paths or loops is \emph{edge-disjoint}, if
no two of them have a common edge. An edge-disjoint set of loops is a \emph{loop configuration}. A \emph{loop configuration with a source $a$ and a sink $f$} is an edge-disjoint set of any number of loops and exactly one 
path starting with the edge $a$ and ending with the edge $f$.
The \emph{arrow} $A(S):=A(S,m,\varepsilon,\delta)$ of a loop configuration $S$ (possibly with a source and a sink) is the product of arrows of all loops and paths in the configuration. An empty product is set to be $1$.

The \emph{arrow from an edge $a$ to an edge $f$} (or \emph{finite-lattice propagator}) is
\begin{equation}\label{eq-def-finite-lattice-propagator}
{A}(a\to f):={A}(a\to f;m,\varepsilon,\delta,T)
:=\frac{\sum\limits_{\substack{\text{loop configurations $S$}\\
\text{with the source $a$ and the sink $f$}}}A(S,m,\varepsilon,\delta)}
{\sum\limits_{\substack{\text{loop configurations $S$}}}A(S,m,\varepsilon,\delta)}.
\end{equation}

Now take a point $(x,t)\in \varepsilon\mathbb{Z}^2$ and set $ x':=x\mod T\varepsilon$, $t':=t\mod T\varepsilon$.
Denote by $a_0,f_1=f_1(x,t),f_2=f_2(x,t)$ the edges starting at $(0,0)$, $(x',t')$, $(x',t')$ and ending at $(\varepsilon/2,\varepsilon/2)$, $(x'-\varepsilon/2,t'+\varepsilon/2)$, $(x'+\varepsilon/2,t'+\varepsilon/2)$ respectively;
see Figure~\ref{fig-1x1} to the middle-right.
The \emph{arrow of the point $(x,t)$} (or \emph{infinite-lattice propagator}) is the pair of complex numbers
\begin{equation}\label{eq-def-infinite-lattice-propagator}
\widetilde{A}_k^{\bluevar{\circ}}(x,t,m,\varepsilon):=2i^{2-k}\,
\lim_{\delta\searrow 0}\lim_{\substack{T\to\infty}}
{A}(a_0\to f_k(x,t);m,\varepsilon,\delta,T)
\qquad\text{for $k=1,2$.}
\end{equation}
\end{definition}

\begin{table}[htb]
  \centering
\begin{tabular}{|l|c|c|c|c|c|c|c|c|c|}
  \hline
  $S$ & $\varnothing$ & $\{aba\}$ & $\{cdc\}$ & $\{aca\}$ & $\{bdb\}$ & $\{abdca\}$ & $\{acdba\}$ & $\{aba,cdc\}$ & $\{aca,bdb\}$ \\
  $A(S)$ & $1$ & $\frac{-im\varepsilon^2\delta}{n}$ & $\frac{-im\varepsilon^2\delta}{n}$ & $-\frac{1}{n}$ & $-\frac{1}{n}$ & $\frac{m^2\varepsilon^2}{n^2}$ & $\frac{-\delta^2\varepsilon^2}{n^2}$ & $\frac{-m^2\varepsilon^4\delta^2}{n^2}$ & $\frac{1}{n^2}$ \\
  \hline
\end{tabular}
  \caption{All loop configurations on the lattice of size $1$
  and their arrows; see Example~\ref{ex-1x1}.
  }\label{tab-1x1}
\end{table}

\begin{example}[Lattice $1\times 1$; see Figure~\ref{fig-1x1} to the left] \label{ex-1x1}
The lattice of size $1$ lies on the square $[0,\varepsilon]^2$ with the opposite sides identified.
The lattice has $2$ points: the midpoint and the identified vertices of the square.
It has $4$ edges $a,b,c,d$ shown in Figure~\ref{fig-1x1} to the left. Note that the 
paths $abdc$, $acdb$, $bacd$ are distinct although they contain the same edges. Their arrows are $\tfrac{-m^2\varepsilon^2}{\sqrt{1-\delta^2\varepsilon^2}(1+m^2\varepsilon^2)}$,
$\tfrac{-\delta\varepsilon}{\sqrt{1-\delta^2\varepsilon^2}(1+m^2\varepsilon^2)}$,
$\tfrac{\delta^2\varepsilon^2}{(1-\delta^2\varepsilon^2)\sqrt{1+m^2\varepsilon^2}}$
respectively. Those paths are not the same as the cycles $acdba$, $bacdb$. The two cycles determine the same loop with the arrow $\tfrac{-\delta^2\varepsilon^2}{(1-\delta^2\varepsilon^2)(1+m^2\varepsilon^2)}$.
All the $9$ loop configurations and their arrows
are listed in Table~\ref{tab-1x1},
where $n:={\sqrt{1-\delta^2\varepsilon^2}\sqrt{1+m^2\varepsilon^2}}$.
We obtain
\begin{align*}
  {A}(a\to b
  ) &= \frac{-im\varepsilon\sqrt{1-\delta^2\varepsilon^2}
  -\delta\varepsilon\sqrt{1+m^2\varepsilon^2}}
  {2(\bluevar{n}-1-im\varepsilon^2\delta)},
  &
  {A}(a\to c
  ) &= \frac{\sqrt{1-\delta^2\varepsilon^2}-\sqrt{1+m^2\varepsilon^2}}
  {2(\bluevar{n}-1-im\varepsilon^2\delta)},
  \\
  {A}(a\to a
  ) &= \frac{1}{2},
  &
  {A}(a\to d
  ) &=
\frac{-im\varepsilon-\delta\varepsilon}
{2(\bluevar{n}-1-im\varepsilon^2\delta)}.
\end{align*}
\end{example}


\begin{theorem}[Equivalence of definitions
]\label{p-real-imaginary}
Both finite- and infinite-lattice propagators are well-defined, that is,
the denominator of~\eqref{eq-def-finite-lattice-propagator} is nonzero,
limit~\eqref{eq-def-infinite-lattice-propagator} exists and equals
$\widetilde{A}_k(x,t,m,\varepsilon)$.
\end{theorem}



We conclude this section with a few identities for the finite-lattice propagator.
The first one is an analogue of the equality $a(\varepsilon,\varepsilon,m,\varepsilon)=1$
up to a factor of $2$ coming from
~\eqref{eq-def-infinite-lattice-propagator}.

\begin{proposition}[Initial value] \label{p-initial} For each 
edge $a$ we have
$
{A}(a\to a
)={1}/{2}.
$
\end{proposition}


\begin{proposition}[Skew-symmetry] \label{p-symmetry}
For each pair of edges $a\ne f$ we have
$$
{A}(a\to f
)=
\begin{cases}
{A}(f\to a
), &\text{for $a\perp f$};\\
-{A}(f\to a
), &\text{for $a\parallel f$}.
\end{cases}
$$
\end{proposition}

Now we state the crucial property analogous to \cite[Proposition~5]{SU-22}
and \cite[Definition~3.1]{Smirnov-10}. 
\blueit{The path consisting of edges $e$ and $f$ is denoted by $ef$, and its arrow is denoted by~$A(ef)$.}

\begin{proposition}[Dirac equation/\bluevar{massive spin-}holomorphicity] \label{p-dirac-finite}
Let $f$ be an edge starting at a lattice point $P$.
Denote by $e$ and $e'$ the two edges ending
at $P$ such that $e\parallel f$ and $e'\perp f$
(see Figure~\ref{fig-1x1} to the right).
Then for each edge $a$ we have
\begin{align*}
{A}(a\to f
  )
  &= {A}(a\to e)A(ef)+{A}(a\to e')A(e'f)+\delta_{af}
  \\&=
\begin{cases}
  \frac{1}{\sqrt{1-\delta^2\varepsilon^2}}
  ({A}(a\to e
  )
  -\delta\varepsilon {A}(a\to e'
  ))+\delta_{af}, &\text{for $P$ even;}\\
  \frac{1}{\sqrt{1+m^2\varepsilon^2}}
  ({A}(a\to e
  )
  -im\varepsilon {A}(a\to e'
  ))+\delta_{af}, &\text{for $P$ odd.}
\end{cases}
\end{align*}
\end{proposition}

Let us state a simple corollary of the previous three identities.

\begin{proposition}[Adjoint Dirac equation] \label{p-dirac-adjoint}
Under the assumptions of Proposition~\ref{p-dirac-finite},
\begin{equation*}
{A}(f\to a)= A(ef)\left({A}(e\to a)-\delta_{ea}\right)
-A(e'f)\left({A}(e'\to a)-\delta_{e'a}\right).
\end{equation*}
\end{proposition}

The following proposition is a simple generalization of Dirac equation.

\begin{proposition}[Huygens' principle] \label{cor-dirac}
For each $n\le 2T$ and each pair of edges $a, f$ we have
$$
A(a\to f)=\sum_{e\dots f\text{ of length }n}A(a\to e)A(e\dots f)+
\sum_{a\dots f\text{ of length }<n}A(a\dots f),
$$
where the first sum is over all the 
paths $e\dots f$ of length exactly $n$ ending with $f$ and the second sum is over
all the 
paths $a\dots f$
of length less than $n$ starting with $a$ and ending with $f$.
\end{proposition}

Note that the \emph{finite}-lattice propagator does \emph{not} exhibit 
charge conservation
(see Example~\ref{ex-2x2}).



\addcontentsline{toc}{myshrink}{}

\section{Generalizations to several particles}

\label{sec-variations}


In this section we upgrade the model to describe motion of \emph{several} non-interacting electrons and positrons,
then introduce interaction and establish perturbation expansion.

\subsection{Identical particles in Feynman checkers}
\label{sec-identical1}

As a warm up, we upgrade Feynman's original model (see Definition~\ref{def-mass})
to two identical electrons. This upgrade takes into account
\emph{chirality} of electrons, which can be either \emph{right} or \emph{left}
\cite[\S4]{SU-22}, but does not yet incorporate creation and
annihilation of electron-positron pairs.


\begin{definition}\label{def-identical}
Under the notation of Definition~\ref{def-mass}, take $m=\varepsilon=1$. Fix integer points $A=(0,0)$, $A'=(x_0,0)$, $F=(x,t)$, $F'=(x',t)$ and their diagonal neighbors $B=(1,1)$, $B'=(x_0+1,1)$,
$E=(x-1,t-1)$, $E'=(x'-1,t-1)$, where $x_0\ne 0$ and $x'\ge x$.
Denote
$$
{a}(AB,A'B'\to EF,E'F'):=
\sum_{\substack{s=AB\dots EF\\s'=A'B'\dots E'F'}} {a}(s){a}(s')-
\sum_{\substack{s=AB\dots E'F'\\s'=A'B'\dots EF}} {a}(s){a}(s'),
$$
where the first sum is over all pairs consisting of a checker path $s$ starting with the move $AB$ and ending with the move $EF$, and a path $s'$ starting with the move $A'B'$ and ending with the move $E'F'$, whereas in the second sum the final moves are interchanged.

The length square $P(AB,A'B'\to EF,E'F'):=\left|{a}(AB,A'B'\to EF,E'F')\right|^2$ is called the \emph{probability 
to find right electrons at $F$ and $F'$, if they are emitted from $A$ and~$A'$}.
(In particular, $P(AB,A'B'\to EF,EF)=0$, i.e., two right electrons cannot be found at the same point; this is called \emph{exclusion principle}.)

Define $P(AB,A'B'\to EF,E'F')$ similarly for $E=(x\pm 1,t-1)$, $E'=(x'\pm 1,t-1)$. Here we require $x'\ge x$, if both signs in $\pm$ are the same, and allow arbitrary $x'$ and $x$, otherwise.
(The latter requirement is introduced not to count twice
the contribution of physically indistinguishable final states
$(EF,E'F')$ and $(E'F',EF)$.)
\end{definition}

\begin{remark}
  \blueit{This is equivalent to the six-vertex model with the weights in Figure~\ref{fig-6v} to the top.}
\end{remark}


\begin{proposition}[Locality]
\label{p-independence-1}
For $x_0\ge 2t$, $x'>x$, $E=(x-1,t-1)$, and $E'=(x'-1,t-1)$ we have
$P(AB,A'B'\to EF,E'F')=|a_2(x,t,1,1)|^2|a_2(x'-x_0,t,1,1)|^2$.
\end{proposition}

This means that two sufficiently distant electrons move independently.

\begin{proposition}[Probability conservation] \label{p-transfer-matrix}
For each $t>0$ and $x_0\ne 0$
we have the identity $\sum_{E,E',F,F'} P(AB,A'B'\to EF,E'F')=1$,
where the sum is over all quadruples $F=(x,t)$, $F'=(x',t)$,
$E=(x\pm 1,t-1)$, $E'=(x'\pm 1,t-1)$, such that $x'\ge x$, if the latter two signs in $\pm$ are the same.
\end{proposition}


\subsection{Identical particles in Feynman anticheckers}
\label{sec-identical}

Now we generalize the new model (see Definition~\ref{def-anti-combi}) to several non-interacting electrons and positrons which
can be created and annihilated during motion.

\begin{definition} \label{def-multipoint}
A \emph{loop configuration $S$ with sources $a_1,\dots,a_n$ and sinks $f_1,\dots,f_n$} is an edge-disjoint set of any number of loops and exactly $n$ 
paths starting with the edges $a_1,\dots,a_n$ and ending with the edges $f_{\sigma(1)},\dots,f_{\sigma(n)}$ respectively, where $\sigma$ is an arbitrary permutation of $\{1,\dots,n\}$. The \emph{arrow} $A(S,m,\varepsilon,\delta)$ \emph{of 
$S$} 
is the permutation sign $\mathrm{sgn}(\sigma)$ times the product of arrows of all loops and paths in the configuration. Define the \emph{arrow}
${A}(a_1,\dots,a_n\to f_1,\dots,f_n
)$ \emph{from $a_1,\dots,a_n$ to $f_1,\dots,f_n$}
analogously to $A(a\to f)$, only the sum in the numerator
of~\eqref{eq-def-finite-lattice-propagator} is over all loop
configurations $S$ with the sources $a_1,\dots,a_n$
and the sinks $f_1,\dots,f_n$.
(This expression vanishes, if some two sources or some two sinks coincide.)

Given edges $a,e,f$, define
${A}(a\to f \text{ pass } e
)$ analogously to
${A}(a\to f
)$,
only the sum in the numerator of~\eqref{eq-def-finite-lattice-propagator} is now over loop configurations with the source $a$ and the sink $f$ containing the edge $e$ (the sum in the denominator remains the same).
\end{definition}

\begin{physicalinterpretation*} Assume that the edges $a_1,\dots,a_k,f_1,\dots,f_l$ start on a horizontal line $t=t_1$ and the remaining sources and sinks start on a horizontal line $t=t_2$, where $t_2>t_1$. Then the model describes a transition from a state with $k$ electrons and $l$ positrons at the time $t_1$ to $n-l$ electrons and $n-k$ positrons at the time $t_2$. Beware that analogously to \cite[\S9.2]{SU-22} one cannot speak of any transition probabilities.
But one can express, say, the expected charge through the arrows.
\blueit{The arrows can be also viewed as $2n$-point functions of the 
theory.}
\end{physicalinterpretation*}

\begin{proposition}[Determinant formula] \label{p-det}
For each 
edges $a_1,\dots,a_n,f_1,\dots,f_n$ we have
$$
{A}(a_1,\dots,a_n\to f_1,\dots,f_n
)
=\sum_{\sigma} \mathrm{sgn}(\sigma)
A(a_1\to f_{\sigma(1)}
)\dots
A(a_n\to f_{\sigma(n)}
)
=\det\left(A(a_k\to f_l
)\right)_{k,l=1}^n,
$$
where the sum is over all permutations $\sigma$ of $\{1,\dots,n\}$.
\end{proposition}


\begin{proposition}[Pass-or-loop formula] \label{cor-pass}
For each 
edges $a,e,f$ we have
$$
{A}(a\to f \text{ pass } e
)=
{A}(a\to f
)
{A}(e\to e
)+
{A}(a\to e
){A}(e\to f
)
=\frac{1}{2}
{A}(a\to f
)
+
{A}(a\to e
){A}(e\to f
).
$$
\end{proposition}

\subsection{Fermi theory and Feynman diagrams}
\label{sec-fermi}

Now we couple two copies of the model in a way resembling Fermi'\blueit{s} theory,
which describes one type of weak interaction between electrons
and muons, slightly heavier \blueit{analogs} of electrons.

\begin{definition} \label{def-fermi} Fix $g, m_\mathrm{e}, m_\mu>0$ called \emph{coupling constant, electron}, and \emph{muon mass} respectively. Denote by $\mathrm{commonedges}(S_\mathrm{e},S_\mu)$ the number of common edges in two loop configurations $S_\mathrm{e},S_\mu$ (possibly with sources and sinks). The \emph{arrow from edges $a_\mathrm{e}$ and $a_\mu$ to edges $f_\mathrm{e}$ and $f_\mu$} is
\begin{multline}\label{eq-def-fermi}
{A}(a_\mathrm{e},a_\mu\to f_\mathrm{e},f_\mu
)
:=\\:=\frac{
\sum\limits_{\substack{\text{loop configurations $S_\mathrm{e}$}\\
\text{with the source $a_\mathrm{e}$}\\
\text{and the sink $f_\mathrm{e}$}}}\quad
\sum\limits_{\substack{\text{loop configurations $S_\mu$}\\
\text{with the source $a_\mu$}\\
\text{and the sink $f_\mu$}}}
A(S_\mathrm{e},m_\mathrm{e},\varepsilon,\delta)A(S_\mu,m_\mu,\varepsilon,\delta)
(1+g)^{\mathrm{commonedges}(S_\mathrm{e},S_\mu)}}
{\sum\limits_{\substack{\text{loop configurations $S_\mathrm{e}$}}}\quad
\sum\limits_{\substack{\text{loop configurations $S_\mu$}}}
A(S_\mathrm{e},m_\mathrm{e},\varepsilon,\delta)A(S_\mu,m_\mu,\varepsilon,\delta)
(1+g)^{\mathrm{commonedges}(S_\mathrm{e},S_\mu)}}.
\end{multline}
\end{definition}

Informally, the powers of $(1+g)$ are explained as follows:
An interaction may or may not occur on each common edge of $S_\mathrm{e}$
and $S_\mu$. Each occurance gives a factor of $g$.


\begin{proposition}[Perturbation expansion] \label{p-perturbation}
For $g$ sufficiently small in terms of
$m_\mathrm{e},m_\mu,\varepsilon,\delta,T$ and for  \blueit{any} edges $a_\mathrm{e},a_\mu,f_\mathrm{e},f_\mu$,
the arrow from $a_\mathrm{e}$ and $a_\mu$ to $f_\mathrm{e}$ and $f_\mu$
is well-defined, that is, the denominator of~\eqref{eq-def-fermi}
is nonzero.  \blueit{As $g\to 0$,} we have
\begin{multline*}
{A}(a_\mathrm{e},a_\mu\to f_\mathrm{e},f_\mu
)
={A}(a_\mathrm{e}\overset{\mathrm{e}}{\to} f_\mathrm{e})
{A}(a_\mu\overset{\mu}{\to} f_\mu)
+g\sum_{e}
\left(
{A}(a_\mathrm{e}\overset{\mathrm{e}}{\to} e)
{A}(e\overset{\mathrm{e}}{\to} f_\mathrm{e})
{A}(a_\mu\overset{\mu}{\to}e)
{A}(e\overset{\mu}{\to} f_\mu)
\right.+\\+\left.
{A}(a_\mathrm{e}\overset{\mathrm{e}}{\to} f_\mathrm{e})
{A}(e\overset{\mathrm{e}}{\to} e)
{A}(a_\mu\overset{\mu}{\to} e)
{A}(e\overset{\mu}{\to} f_\mu)
+
{A}(a_\mathrm{e}\overset{\mathrm{e}}{\to} e)
{A}(e\overset{\mathrm{e}}{\to} f_\mathrm{e})
{A}(a_\mu\overset{\mu}{\to} f_\mu)
{A}(e\overset{\mu}{\to} e)
\right)
+o(g),
\end{multline*}
where the sum is over all edges $e$ and we denote
$A(a\overset{\mathrm{e}}{\to} f):= A(a\to f;m_\mathrm{e},\varepsilon,\delta,T)$ and
$A(a\overset{\mu}{\to} f):= A(a\to f;m_\mu,\varepsilon,\delta,T)$.
\end{proposition}



\begin{figure}[htbp]
  \centering
  \small
\begin{tabular}{|c|c|}
\hline
${A}(a_\mathrm{e}\overset{\mathrm{e}}{\to} f_\mathrm{e})
{A}(a_\mu\overset{\mu}{\to} f_\mu)$
&
$g\,{A}(a_\mathrm{e}\overset{\mathrm{e}}{\to} e)
{A}(e\overset{\mathrm{e}}{\to} f_\mathrm{e})
{A}(a_\mu\overset{\mu}{\to}e)
{A}(e\overset{\mu}{\to} f_\mu)$
\\[3pt]
\includegraphics[height=2.2cm]{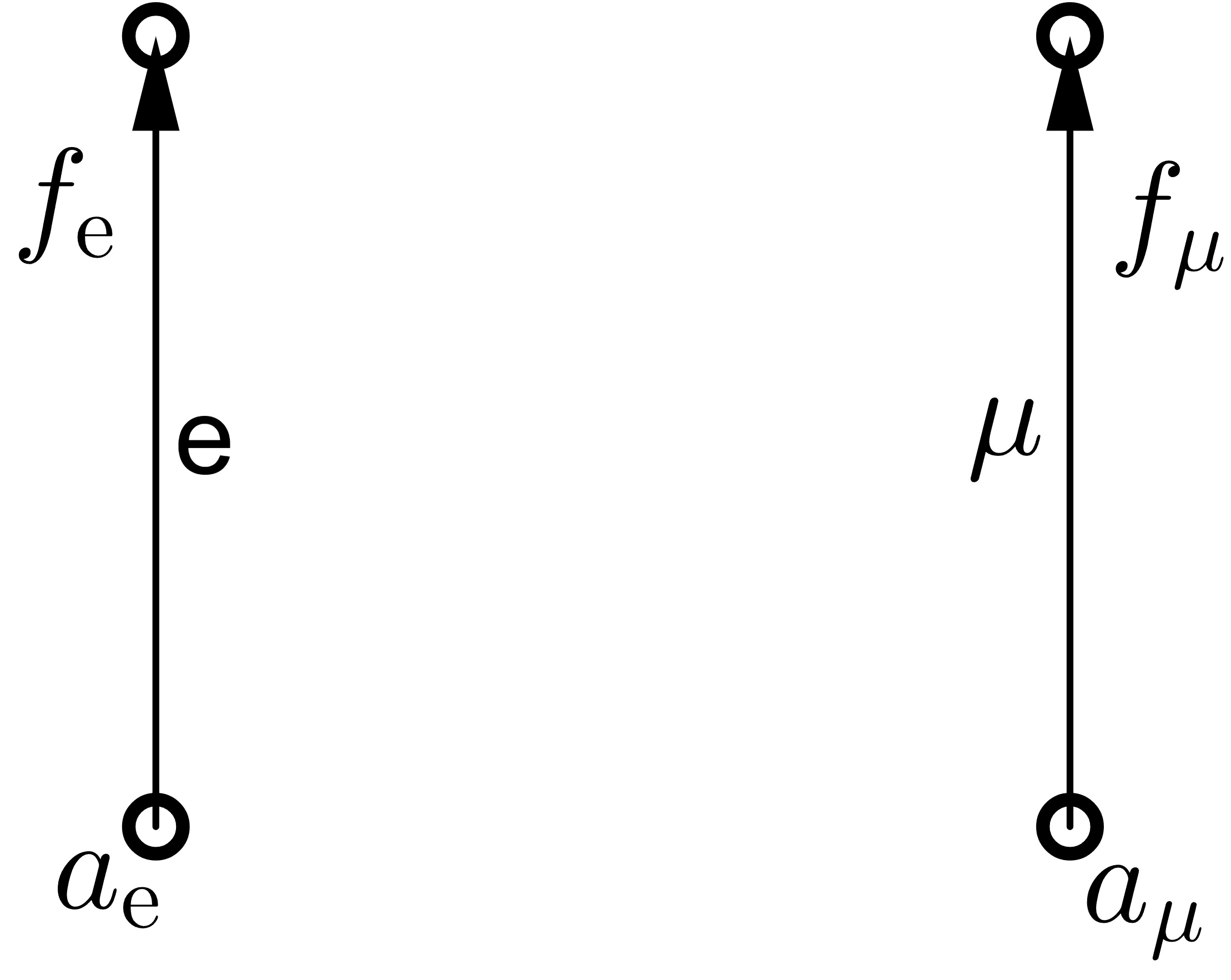} &
\includegraphics[height=2.2cm]{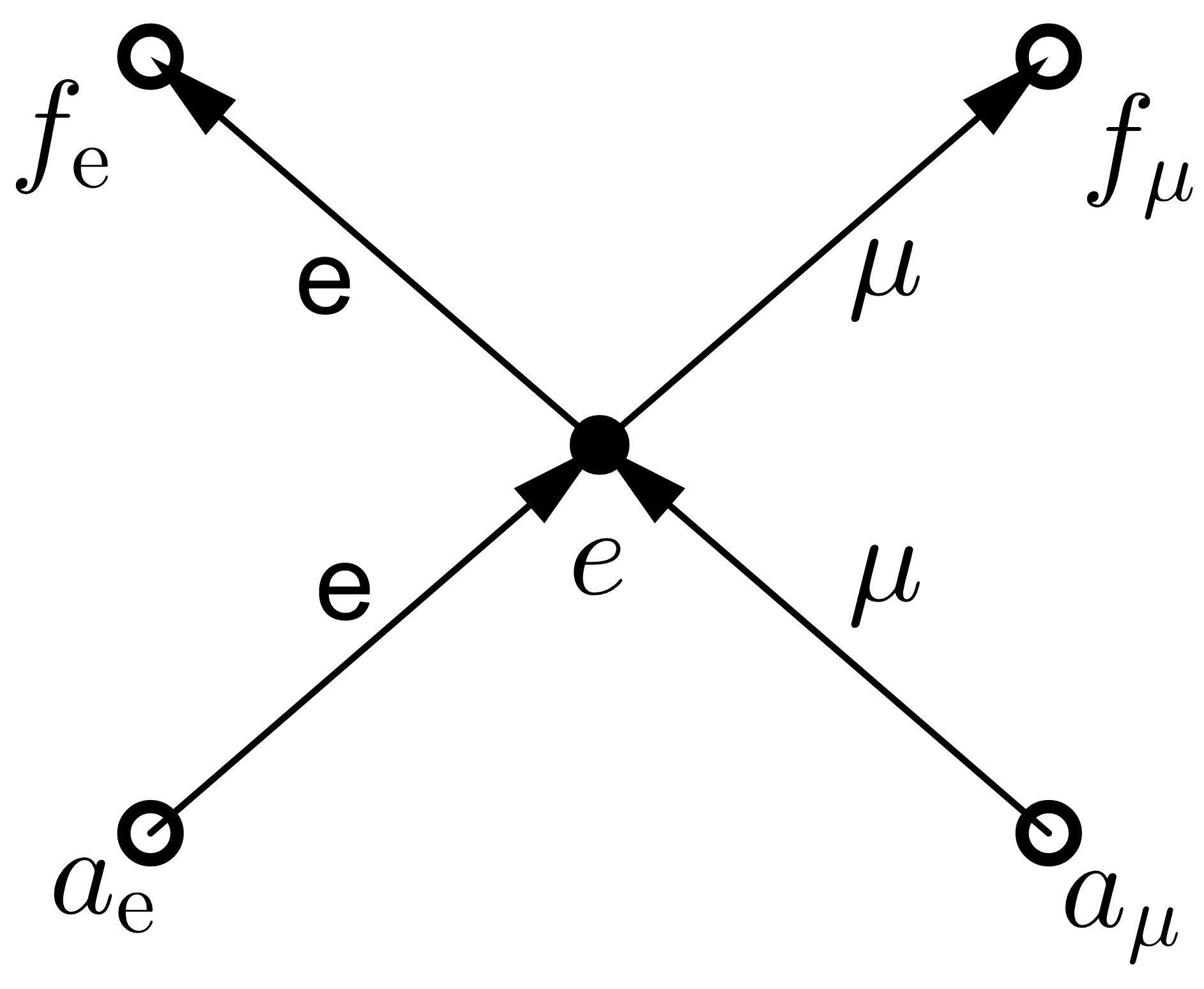}
\\
\includegraphics[height=2.2cm]{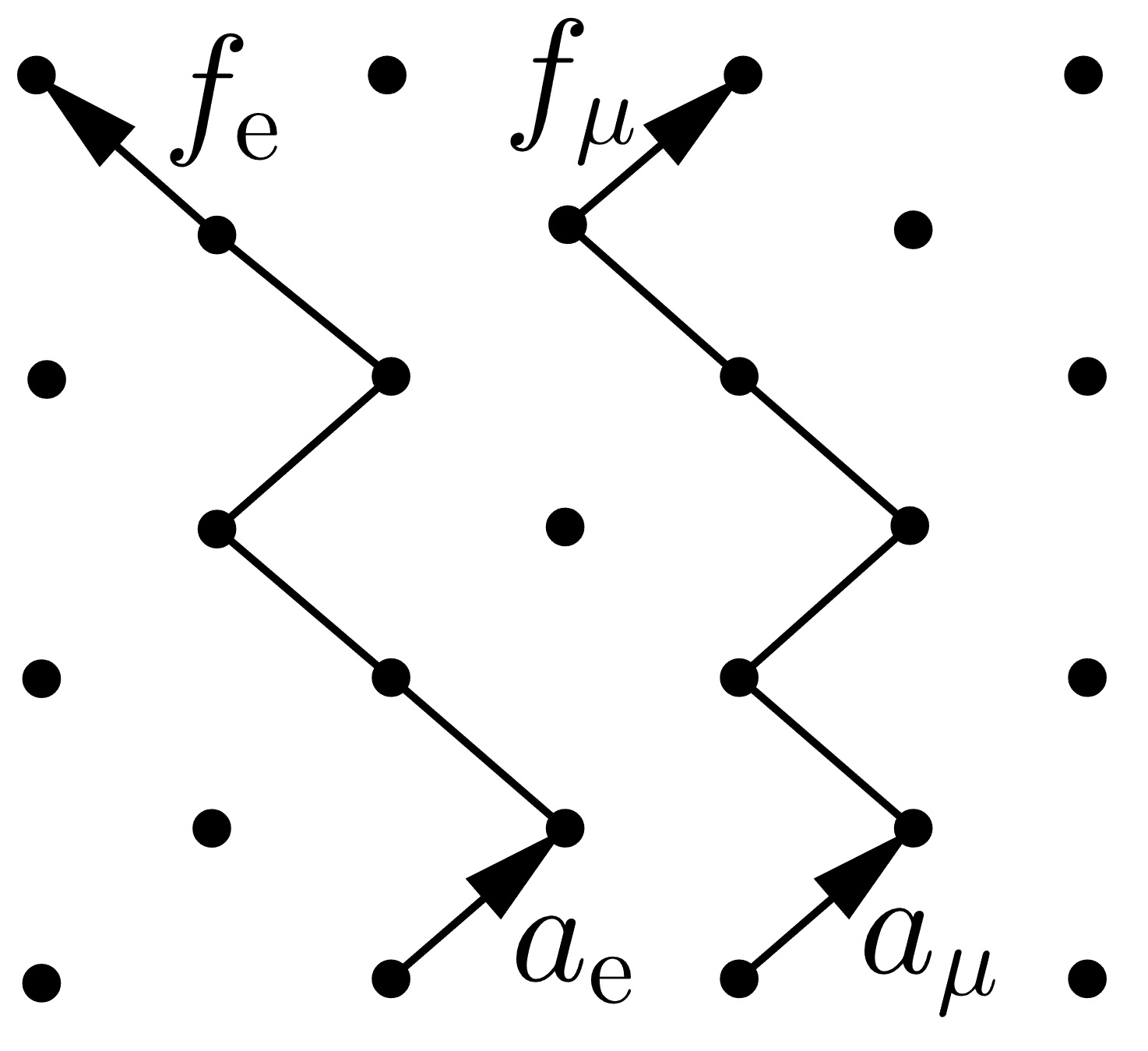} &
\includegraphics[height=2.2cm]{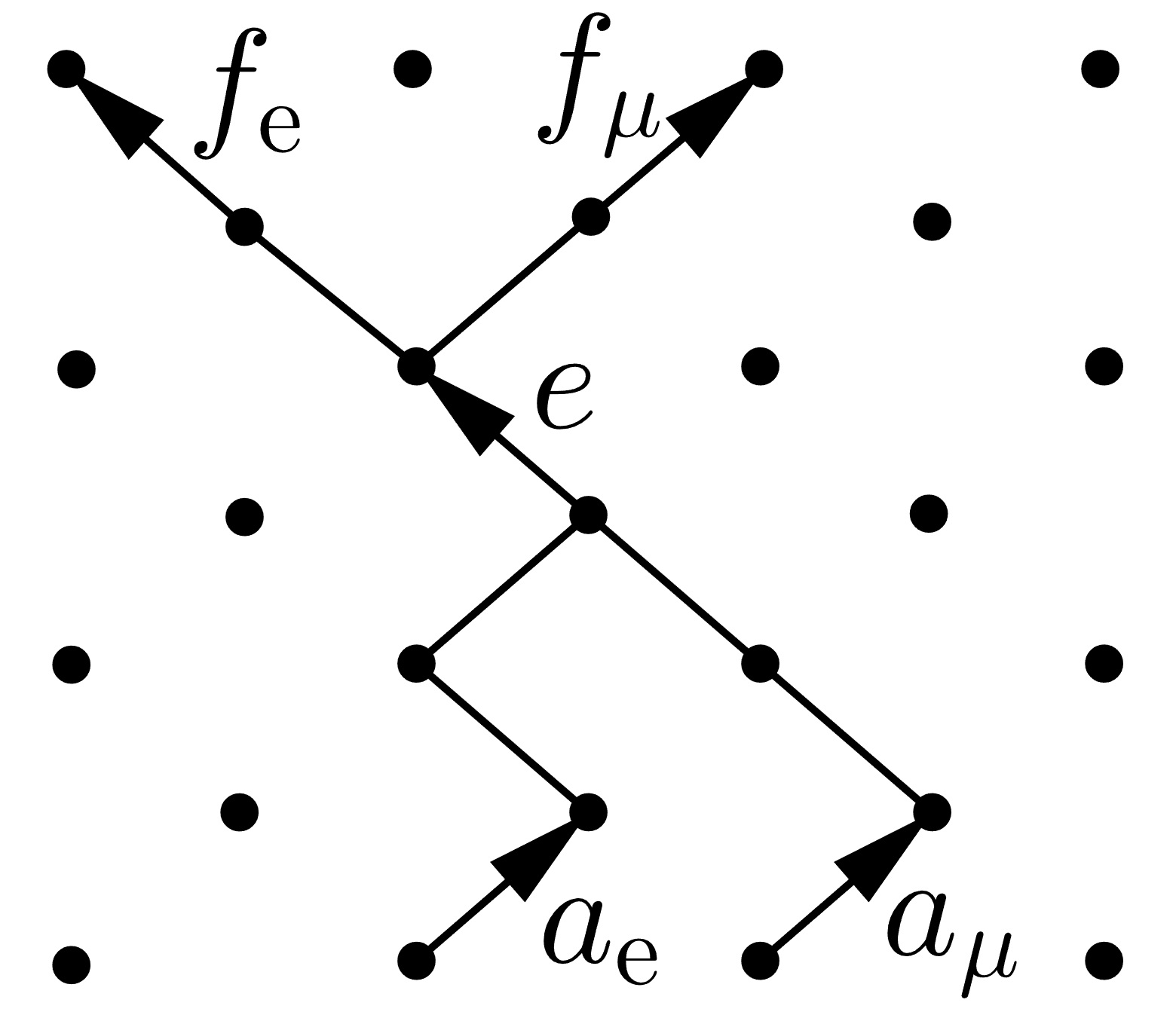}
\\
\hline
$g\,{A}(a_\mathrm{e}\overset{\mathrm{e}}{\to} f_\mathrm{e})
{A}(e\overset{\mathrm{e}}{\to} e)
{A}(a_\mu\overset{\mu}{\to} e)
{A}(e\overset{\mu}{\to} f_\mu)$
&
$g\,{A}(a_\mathrm{e}\overset{\mathrm{e}}{\to} e)
{A}(e\overset{\mathrm{e}}{\to} f_\mathrm{e})
{A}(a_\mu\overset{\mu}{\to} f_\mu)
{A}(e\overset{\mu}{\to} e)$\\[3pt]
\includegraphics[height=2.2cm]{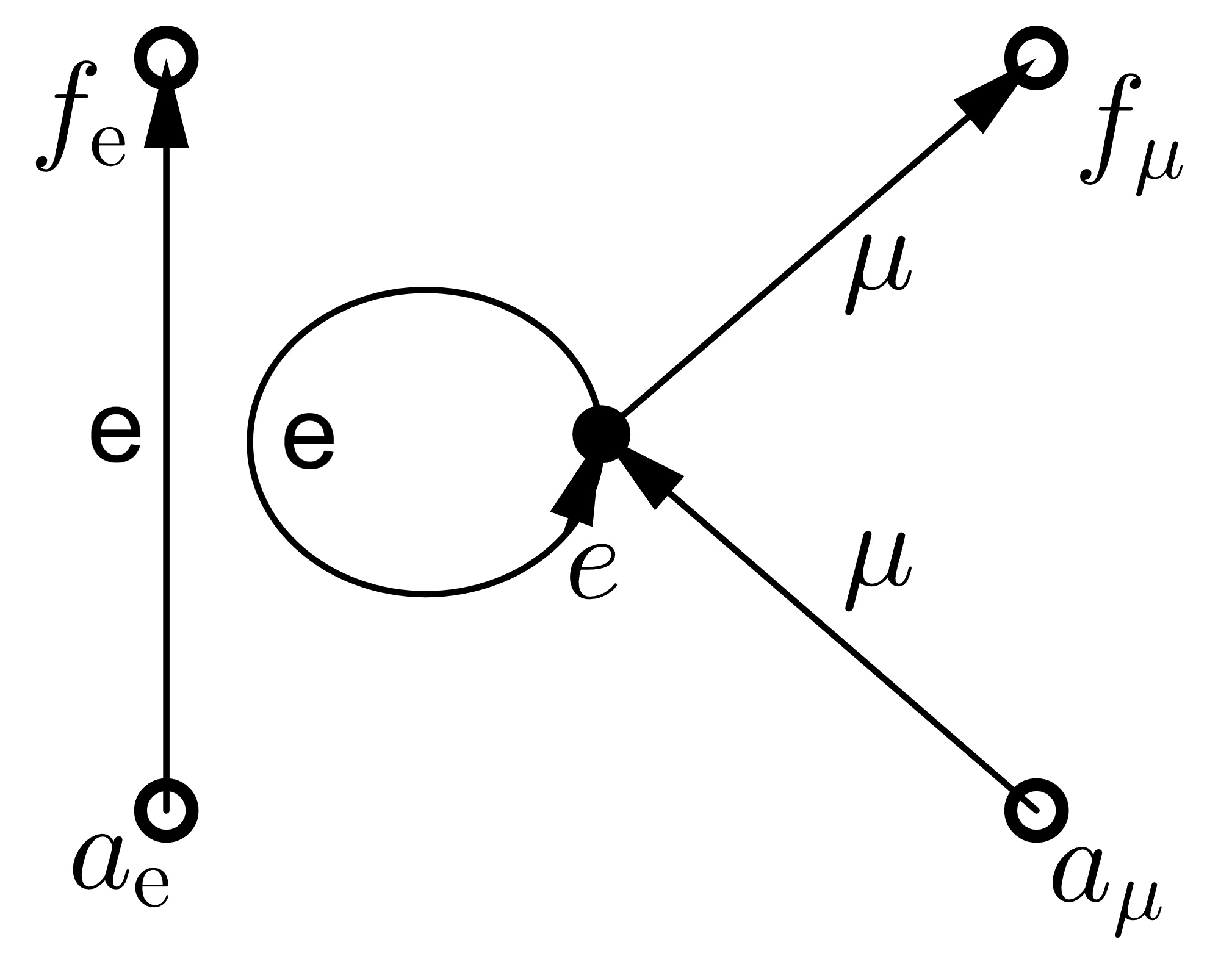} &
\includegraphics[height=2.2cm]{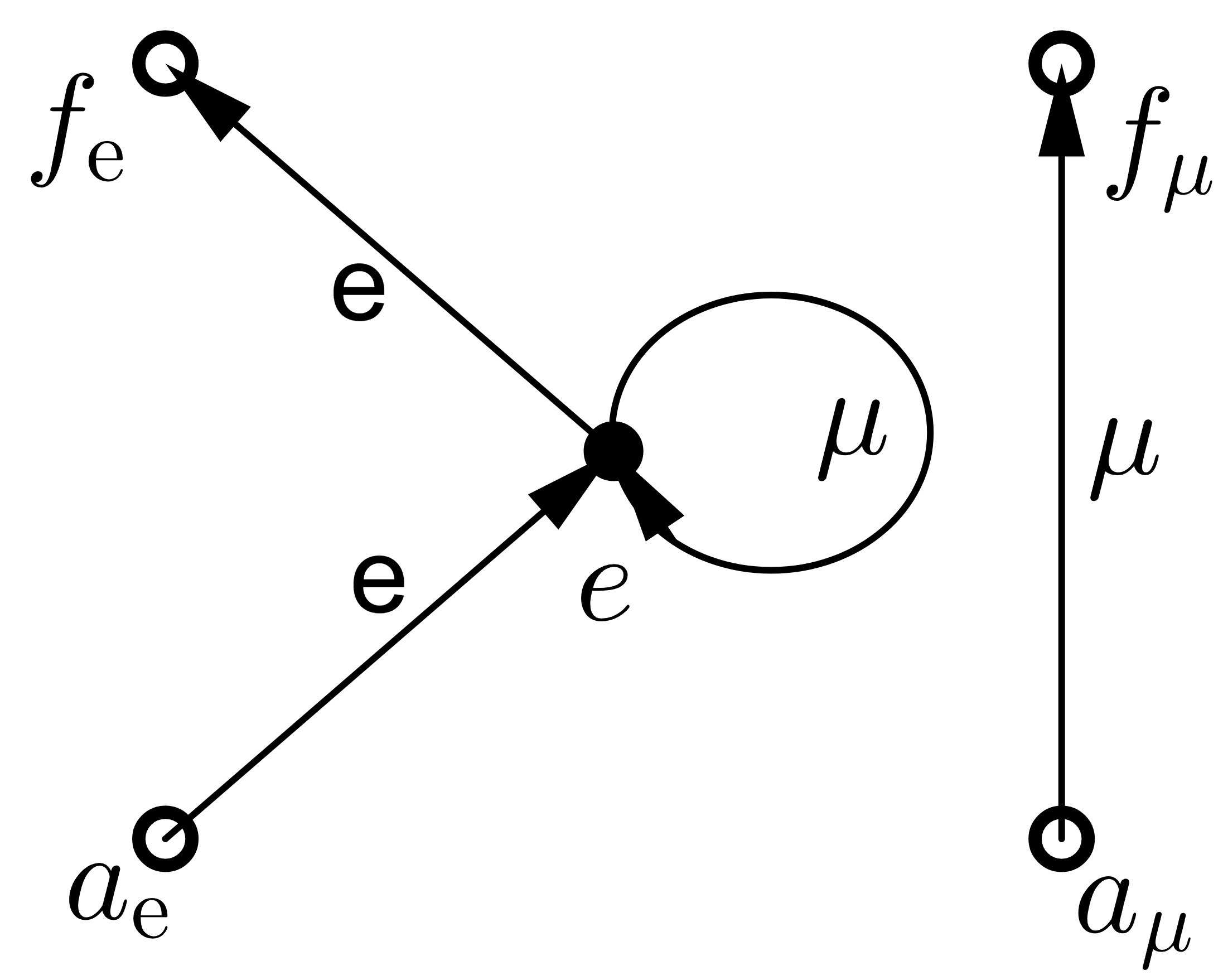} \\
\includegraphics[height=2.2cm]{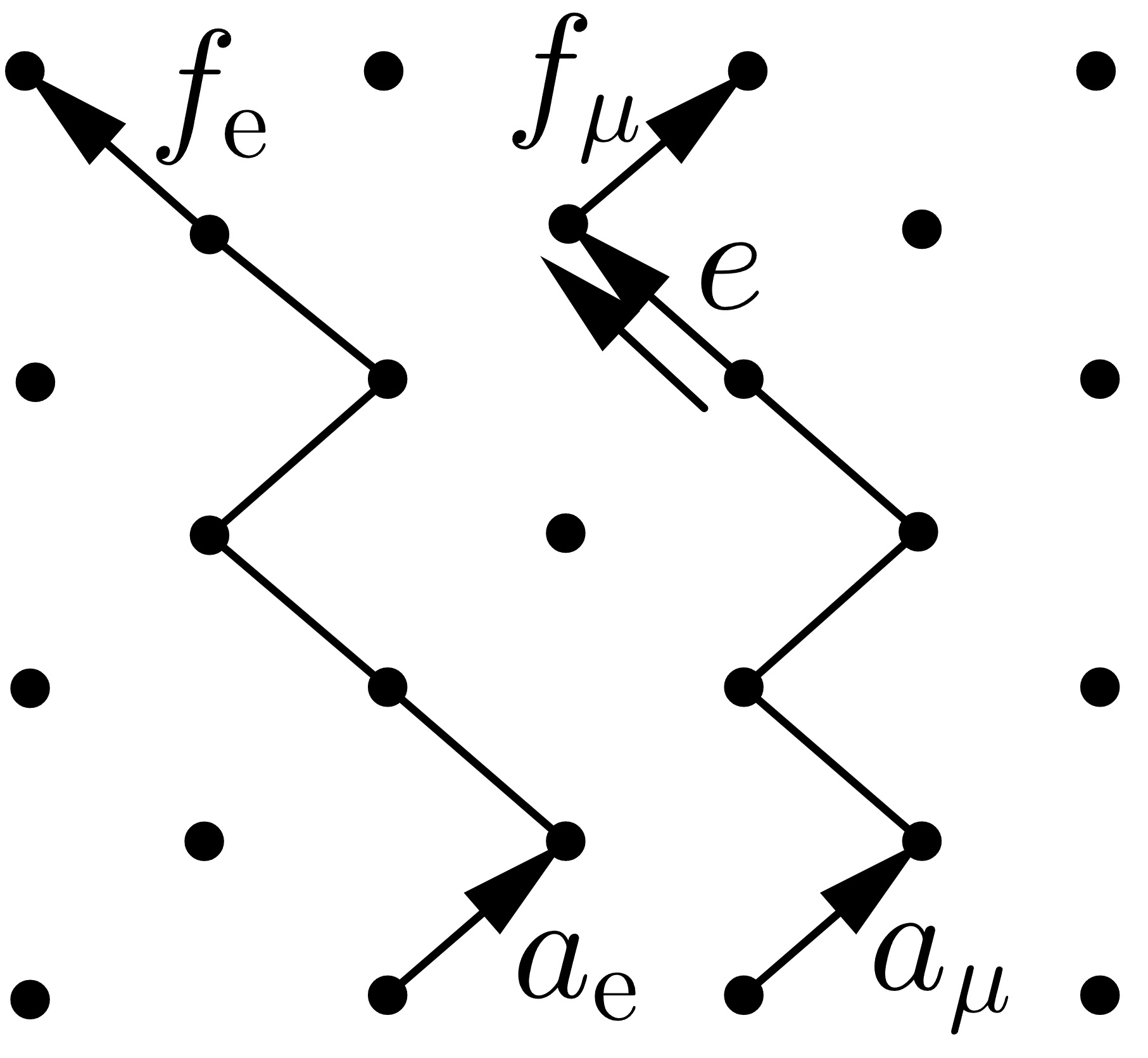} &
\includegraphics[height=2.2cm]{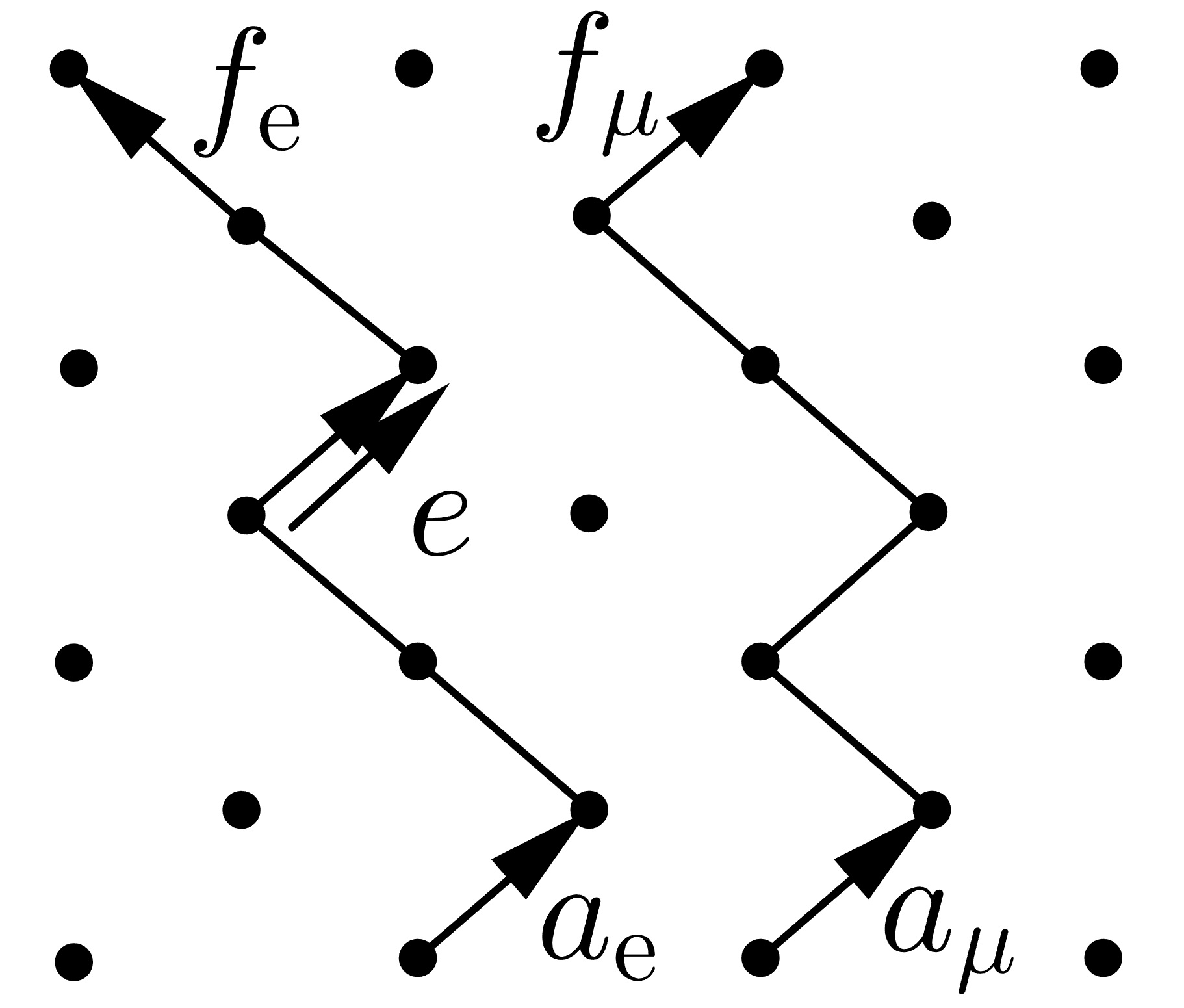}
\\
\hline
\end{tabular}
\caption{Terms in perturbation expansion, their Feynman diagrams,
and 
collections of loop configurations
contributing to the terms (from top to bottom in each cell); see Proposition~\ref{p-perturbation}.}
  \label{fig-feynman}
\end{figure}

The perturbation expansion can be extended to any order in $g$.
The terms are depicted as so-called \emph{Feynman diagrams} as follows (see Figure~\ref{fig-feynman}). For each edge
in the left side, draw a white vertex. For each edge that is a
summation variable in the right side, draw a black vertex. For each factor
of the form $A(a\overset{\mathrm{e}}{\to} f)$
or $A(a\overset{\mu}{\to} f)$, draw an arrow from the vertex drawn for $a$
to the vertex drawn for $f$ (a loop, if $a=f$) labeled by letter ``$\mathrm{e}$''
or ``$\mu$'' respectively.

We conjecture that those Feynman diagrams have usual properties:
each black vertex is the starting point of exactly two arrows
labeled by ``$\mathrm{e}$'' and ``$\mu$'',
is the endpoint of exactly two arrows
also labeled by ``$\mathrm{e}$'' and ``$\mu$'',
and is joined with a white vertex by a sequence of arrows.

Let us give a few comments for specialists. 
As $\varepsilon\to 0$, the contribution of Feynman diagrams
involving loops blows up because
$A(e\to e)=1/2$ by Proposition~\ref{p-initial} whereas
the other arrows are of order
$\varepsilon$ by Theorem~\ref{th-continuum-limit}.
This suggests that the model has no \emph{naive} continuum limit.
As usual, the \emph{true} continuum limit requires \emph{renormalization}, that is, choosing a lattice-dependent coupling $g(\varepsilon)$ in a wise way. Fermi model in $1$ space and $1$ time dimension is known to be \emph{renormalizable} \cite[\S III.3, top of p.~180]{Zee-10}; thus one expects that the true continuum limit exists. Proving the existence mathematically is as hard as for any other model with interaction.

\subsection{Open problems}

The new model is only a starting point of
the missing Minkowskian lattice quantum field theory.
Here we pick up a few informal open problems among
a variety of research directions.

We start with the ones relying on
Definition~\ref{def-anti-alg} only. As a warm-up, we suggest the following.

\begin{pr} (Cf.~\cite[Theorem~1]{Novikov-20},\cite{Kuyanov-Slizkov-22}) 
Does expected charge~\eqref{eq-q} vanish somewhere?
\end{pr}

The most shouting problem is to find a large-time asymptotic formula\tobereplaced{,}{ for $|x|>|t|/\sqrt{1+m^2\varepsilon^2}$
and} especially for $|x|>|t|$.

\begin{pr} (Cf.~Theorem~\tobereplaced{\ref{th-anti-ergenium}}{\ref{th-anti-Airy}} \blueit{and }\cite{Drmota-etal})
Prove that for each $m,\varepsilon,\Delta>0$ and each $(x,t)\in\varepsilon\mathbb{Z}^2$ satisfying
$1/\sqrt{1+m^2\varepsilon^2}<|x/t|<1-\Delta$
we have
\begin{align*}
\tilde{A}_1\left(x,t,m,\varepsilon\right)
&=
\frac{i^{\frac{|x|-|t|+\varepsilon}{\varepsilon}}\varepsilon
\,(\bluevar{-m\mathcal{L}}(x/t,m,\varepsilon))^{1/2}}
{\pi((1+m^2\varepsilon^2)x^2/t^2-1)^{1/4}}
K_{ 1/3}\left(\bluevar{-t\mathcal{L}}(x/t,m,\varepsilon)\right)
\left(1+O_{m,\varepsilon,\Delta}\left(\frac{1}{|t|}\right)\right),\\
\tilde{A}_2\left(x,t,m,\varepsilon \right)
&= \sqrt{\frac{t + x}{t - x}}\cdot
\frac{i^{\frac{|x|-|t|}{\varepsilon}}\varepsilon
\,(\bluevar{-m\mathcal{L}}(x/t,m,\varepsilon))^{1/2}}
{\pi((1+m^2\varepsilon^2)x^2/t^2-1)^{1/4}}
K_{ 1/3}\left(\bluevar{-t\mathcal{L}}(x/t,m,\varepsilon)\right)
\left(1+O_{m,\varepsilon,\Delta}\left(\frac{1}{|t|}\right)\right),\\
\intertext{where}
\bluevar{\mathcal{L}(v,m,\varepsilon)}&:=
\bluevar{\frac{1}{\varepsilon}\,\mathrm{arcosh}\,
\frac{m\varepsilon } {\sqrt{\left(1+m^2\varepsilon^2\right)\left(1-v^2\right)}}
-\frac{|v|}{\varepsilon}\,\mathrm{arcosh}\,
\frac{m\varepsilon |v|}{\sqrt{1-v^2}}.}
\end{align*}
\end{pr}

The limit of small lattice step also deserves attention.
Corollary~\ref{cor-anti-uniform} 
assumes $|x|\ne |t|$, hence misses the main contribution to the 
charge. Now we ask for the weak limit detecting the peak.

\begin{pr} \label{p-weak} (Cf.~Corollary~\ref{cor-anti-uniform})
Find the distributional limits $\lim_{\varepsilon\searrow 0}  \widetilde{A}\left(x_\varepsilon,t_\varepsilon,{m},{\varepsilon}\right)/{4\varepsilon}$ and $\lim_{\varepsilon\searrow 0} {Q}\left(x_\varepsilon,t_\varepsilon,{m},{\varepsilon}\right)/{8\varepsilon^2}$
on the whole $\mathbb{R}^2$. Is the former limit equal to propagator~\eqref{eq-feynman-propagator-fourier} including the generalized function supported on the lines $t=\pm x$?
\end{pr}

The infinite-lattice propagator seems to be unique to
satisfy the variety of properties from \S\S\ref{ssec-asymptotic}--\ref{ssec-identities}.
But there still could be different \emph{finite}-lattice
propagators with the same limit.

\begin{pr} (Cf.~Definition~\ref{def-anti-combi}, Example~\ref{ex-2x2})
Find a combinatorial construction of a finite-lattice propagator
having the following properties:
\begin{description}
  \item[consistency:] it has the same limit~\eqref{eq-def-infinite-lattice-propagator};
  \item[charge conservation:] it satisfies an analogue of
  Proposition~\ref{p-mass2} before passing to the limit;
  \item[other boundary conditions:] it does not require time-periodic boundary conditions.
\end{description}
\end{pr}

The new model describes a free massive spin-$1/2$ quantum field
but can be easily adopted to other spins via known relations between
propagators for different spins.
For instance, \emph{spin-$0$ and spin-$1$ massive infinite-lattice propagators}
are defined to be $\frac{\sqrt{1+m^2\varepsilon^2}}{2m\varepsilon}\widetilde{A}_1\left(x,t,{m},{\varepsilon}\right)$
and $\frac{\sqrt{1+m^2\varepsilon^2}}{2m\varepsilon}\left(
\begin{smallmatrix}
  \widetilde{A}_1\left(x,t,{m},{\varepsilon}\right) & 0 \\
  0 & \widetilde{A}_1\left(x,t,{m},{\varepsilon}\right)
\end{smallmatrix}
\right)$ respectively. Consistency with continuum theory
is 
automatic by Corollary~\ref{cor-anti-uniform} and
Proposition~\ref{p-Klein-Gordon-mass}. However, it is natural to
modify the combinatorial definition.

\begin{pr} (Cf.~\S\ref{ssec-proofs-fourier})
Find a combinatorial construction of
$\widetilde{A}_1\left(x,t,{m},{\varepsilon}\right)$
starting from the Klein--Gordon equation (Proposition~\ref{p-Klein-Gordon-mass})
instead of the Dirac one, to make the construction symmetric with respect to
time reversal $t\mapsto -t$.
\end{pr}

\begin{pr} (Cf.~Example~\ref{p-massless})
Find a 
combinatorial construction of \emph{massless}
spin-$0$, spin-$1/2$, and spin-$1$ infinite-lattice propagators
(obtained from the massive ones in the limit $m\to 0$).
\end{pr}

\begin{pr}
Modify the combinatorial definition of the model with several identical particles (Definition~\ref{def-multipoint})
for spin $0$ so that the determinant
in Proposition~\ref{p-det} is replaced by the permanent.
\end{pr}

The next challenge is to introduce interaction
and to prove that the continuum limit of the resulting model has natural physical properties
(at least, satisfies Wightman axioms recalled in Appendix~\ref{sec-axioms}).
Particular goals could be quantum electrodynamics and Fermi model (see \S\ref{sec-fermi}).


\addcontentsline{toc}{myshrink}{}

\section{Proofs}\label{sec-proofs}



Let us present a chart showing
the dependence of the above results and further subsections:
{
$$
\xymatrix{
\boxed{\text{\ref{ssec-proofs-fourier-series}.
(Theorem~\ref{th-well-defined},
Proposition~\ref{th-equivalence})}}
\ar[d]\ar[r]\ar[dr]
&
\boxed{\text{\ref{ssec-proofs-continuum}. (Theorem~\ref{th-continuum-limit}, Corollary~\ref{cor-anti-uniform})}}
\\
\boxed{\text{\ref{ssec-proofs-identities}. (Propositions~\ref{p-basis}--\ref{p-triple})}}
&
\boxed{\text{\ref{ssec-proofs-fourier}.
(Theorem~\ref{p-real-imaginary},
Propositions~\ref{p-initial}--\ref{cor-dirac})}}
\ar[dl]
\\
\boxed{\text{\ref{ssec-proofs-variations}. (Propositions~\ref{p-independence-1}--\ref{p-perturbation})}}
&
\boxed{\text{Appendix~\ref{ssec-proofs-alternative}. (Propositions~\ref{p-initial}--\ref{cor-dirac})}}
}
$$
}

Section~\ref{ssec-proofs-variations} relies only on
Theorem~\ref{p-real-imaginary},
Proposition~\ref{p-initial}, and~Lemma~\ref{l-loop-expansion}
among the results proved in \S\S\ref{ssec-proofs-fourier-series},\ref{ssec-proofs-fourier}.
Appendix~\ref{ssec-proofs-alternative} contains alternative proofs and is
independent from \S\ref{sec-proofs}.

Throughout this section we use notation~\eqref{eq-omega}.



\subsection{Fourier integral 
(Theorem~\ref{th-well-defined},
Proposition~\ref{th-equivalence})}
\label{ssec-proofs-fourier-series}

In this section we compute the functions in Definition~\ref{def-anti-alg}
by Fourier method (see Proposition~\ref{l-double-fourier-finite}). Then we
obtain Proposition~\ref{th-equivalence} by contour integration
(this step has been already performed in \cite{SU-22}). Finally,
we discuss direct consequences (Corollaries~\ref{cor-scaling}--\ref{cor-halve-interval} and Theorem~\ref{th-well-defined}).
Although the method is analogous to the computation of the continuum
propagator, it is the new idea of putting imaginary mass to the dual lattice what
makes it successful (see Remark~\ref{rem-delta-zero}).

\begin{proposition}[Full space-time Fourier transform] \label{l-double-fourier-finite}
There exists a unique pair of functions satisfying axioms 1--3 in Definition~\ref{def-anti-alg}.
Under notation $x^*:=x+\frac{\varepsilon}{2}$, $t^*:=t+\frac{\varepsilon}{2}$, it is given by
\begin{equation*}
  \begin{aligned}
  A_1(x,t
  )&=
  \begin{cases}
  \frac{\varepsilon^2}{4\pi^2}
  \int_{-\pi/\varepsilon}^{\pi/\varepsilon}
  \int_{-\pi/\varepsilon}^{\pi/\varepsilon}
  \frac{m\varepsilon-i\delta\varepsilon e^{ip\varepsilon}}
  {\sqrt{1-\delta^2\varepsilon^2}\sqrt{1+m^2\varepsilon^2}\cos(\omega\varepsilon)
  -\cos(p\varepsilon)-im\varepsilon^2\delta}
  e^{i p x-i\omega t}\,d\omega dp,
  &\text{for $2x/\varepsilon$ even},\\
 \frac{\varepsilon^2}{4\pi^2}
 \int_{-\pi/\varepsilon}^{\pi/\varepsilon}
 \int_{-\pi/\varepsilon}^{\pi/\varepsilon}
 \frac{m\varepsilon\sqrt{1-\delta^2\varepsilon^2} e^{i\omega\varepsilon}-i\delta\varepsilon\sqrt{1+m^2\varepsilon^2}}
 {\sqrt{1-\delta^2\varepsilon^2}\sqrt{1+m^2\varepsilon^2}\cos(\omega\varepsilon)
  -\cos(p\varepsilon)-im\varepsilon^2\delta}e^{i p x^*-i\omega t^*}\,d\omega dp,
  &\text{for $2x/\varepsilon$ odd};
  \end{cases}
  \\
  A_2(x,t
  )&=
  \begin{cases}
  \frac{
  \varepsilon^2}{4\pi^2}
  \int_{-\pi/\varepsilon}^{\pi/\varepsilon}
  \int_{-\pi/\varepsilon}^{\pi/\varepsilon}
 \frac{\sqrt{1-\delta^2\varepsilon^2}\sqrt{1+m^2\varepsilon^2}
 e^{-i\omega\varepsilon}-e^{ip\varepsilon}-im\varepsilon^2\delta}
  {\sqrt{1-\delta^2\varepsilon^2}\sqrt{1+m^2\varepsilon^2}
  \cos(\omega\varepsilon)
  -\cos(p\varepsilon)-im\varepsilon^2\delta}
  e^{i p x-i\omega t}
  \,d\omega dp,
  &\text{for $2x/\varepsilon$ even},
  \\
 \frac{\varepsilon^2}{4\pi^2}
  \int_{-\pi/\varepsilon}^{\pi/\varepsilon}
  \int_{-\pi/\varepsilon}^{\pi/\varepsilon}
  \frac{\sqrt{1+m^2\varepsilon^2}e^{-ip\varepsilon}
  -\sqrt{1-\delta^2\varepsilon^2}e^{i\omega\varepsilon}}
  {\sqrt{1-\delta^2\varepsilon^2}\sqrt{1+m^2\varepsilon^2}\cos(\omega\varepsilon)
  -\cos(p\varepsilon)-im\varepsilon^2\delta}
  e^{i p x^*-i\omega t^*}\,d\omega dp,
  &\text{for $2x/\varepsilon$ odd}.
  \end{cases}
  \end{aligned}
  \end{equation*}
\end{proposition}

\begin{proof} Substituting axiom~2 into axiom~1 for each $(x,t)\in\varepsilon\mathbb{Z}^2$, we get
\begin{equation}\label{eq-substituted}
\begin{aligned}
  A_1(x,t)&=
\frac{
  A_1(x+\varepsilon,t-\varepsilon)
  +im\varepsilon^2\delta A_1(x,t-\varepsilon)
  -i\delta\varepsilon A_2(x+\varepsilon,t-\varepsilon)
  +m\varepsilon A_2(x,t-\varepsilon)
}{\sqrt{1+m^2\varepsilon^2}\sqrt{1-\delta^2\varepsilon^2}},\\
  A_2(x,t)&=
\frac{
  i\delta\varepsilon A_1(x-\varepsilon,t-\varepsilon)
  -m\varepsilon A_1(x,t-\varepsilon)
  +A_2(x-\varepsilon,t-\varepsilon)
  +im\varepsilon^2\delta A_2(x,t-\varepsilon)
}{\sqrt{1+m^2\varepsilon^2}\sqrt{1-\delta^2\varepsilon^2}}+2\delta_{x0}\delta_{t0}.
\end{aligned}
\end{equation}
It suffices to solve system~\eqref{eq-substituted} on $\varepsilon\mathbb{Z}^2$ (that is, for $2x/\varepsilon$ even) under the restriction given by
axiom~3; then the values for $2x/\varepsilon$ odd are
uniquely determined (and computed) by axiom~2.

We use Fourier method. To a function $A_k(x,t)$ satisfying
axiom~3 assign the Fourier series
$$
\widehat{A}_k(p,\omega):\overset{L^2}{=}
\sum_{(x,t)\in\varepsilon\mathbb{Z}^2}
A_k(x,t)e^{-ipx+i\omega t}\in L^2([-\pi/\varepsilon,\pi/\varepsilon]^2).
$$
Here the summands are understood as functions on $[-\pi/\varepsilon,\pi/\varepsilon]^2$
and mean-square convergence of the series is assumed.
By Plancherel theorem, this assignment is a bijection between the space
of functions satisfying axiom~3 and the space $L^2([-\pi/\varepsilon,\pi/\varepsilon]^2)$
of square-integrable functions $[-\pi/\varepsilon,\pi/\varepsilon]^2\to\mathbb{C}$
up to change on a set of measure zero.

Under this bijection, the shifts $x\mapsto x\pm\varepsilon$
and  $t\mapsto t-\varepsilon$ are taken to multiplication
by $e^{\pm ip\varepsilon}$ and $e^{i\omega\varepsilon}$ respectively,
and $\delta_{x0}\delta_{t0}$ is taken to $1$.
Thus~\eqref{eq-substituted} is transformed to
the following equality almost everywhere
\begin{equation}\label{eq-transformed-equation}
\begin{pmatrix}
            \widehat A_1(p,\omega) \\
            \widehat A_2(p,\omega)
\end{pmatrix}
=
\frac{e^{i\omega\varepsilon}}
{\sqrt{1+m^2\varepsilon^2}\sqrt{1-\delta^2\varepsilon^2}}
\begin{pmatrix}
  e^{ip\varepsilon}+im\varepsilon^2\delta & -i\delta\varepsilon e^{ip\varepsilon}+m\varepsilon \\
  i\delta\varepsilon e^{-ip\varepsilon}-m\varepsilon & e^{-ip\varepsilon}+im\varepsilon^2\delta
\end{pmatrix}
\begin{pmatrix}
  \widehat A_1(p,\omega) \\
  \widehat A_2(p,\omega)
\end{pmatrix}
+
\begin{pmatrix}
  0 \\
  2
\end{pmatrix}.
\end{equation}
The resulting $2\times 2$ linear system has the unique solution
(this is checked in \cite[\S1]{SU-3})
\begin{align*}
\widehat A_1(p,\omega)
&=\frac{m\varepsilon - i\delta\varepsilon e^{ip\varepsilon}}
{\sqrt{1-\delta^2\varepsilon^2}\sqrt{1+m^2\varepsilon^2}\cos(\omega\varepsilon)
- \cos(p\varepsilon) - im\varepsilon^2\delta},\\
\widehat A_2(p,\omega)
&=\frac{
\sqrt{1-\delta^2\varepsilon^2}\sqrt{1+m^2\varepsilon^2}e^{-i\omega\varepsilon}
  - e^{ip\varepsilon}-im\varepsilon^2\delta}
{\sqrt{1-\delta^2\varepsilon^2}\sqrt{1+m^2\varepsilon^2}\cos(\omega\varepsilon)
- \cos(p\varepsilon) - im\varepsilon^2\delta}.
\end{align*}
It belongs to $L^2([-\pi/\varepsilon,\pi/\varepsilon]^2)$ because $m,\varepsilon,\delta>0$
and the denominator vanishes nowhere.
Now the formula for the Fourier coefficients gives the
desired expressions in the proposition.
\end{proof}

\begin{remark}\label{rem-delta-zero} This argument shows that for $\delta=0$ axioms~1--3 are inconsistent
even if $m$ has imaginary part because
$\widehat A_k(p,\omega)$ blows up at $(\pi/2\varepsilon,\pi/2\varepsilon)$.
Thus Step~2 in \S\ref{ssec-outline} is necessary.
\end{remark}

Passing to the limit $\delta\searrow 0$
in Proposition~\ref{l-double-fourier-finite}, using
$\frac{\varepsilon}{2\pi}\int_{-\pi/\varepsilon}^{\pi/\varepsilon}e^{ipx}\,dp=\delta_{x0}$ for $x\in\varepsilon\mathbb{Z}$,
\blueit{and absorbing the factor $m\varepsilon^2$ into $\delta$ in the denominator,}
we get the following result.

\begin{proposition}[Full space-time Fourier transform] \label{cor-double-fourier-anti}
For each 
$(x,t)\in \varepsilon\mathbb{Z}^2$ we have
\begin{equation*}
  \begin{aligned}
  \tilde A_1(x,t
  )&=
  \lim_{\delta\searrow 0}
  \frac{m\varepsilon^3}{4\pi^2}
  \int_{-\pi/\varepsilon}^{\pi/\varepsilon}
  \int_{-\pi/\varepsilon}^{\pi/\varepsilon}
  \frac{ e^{i p x-i\omega t}\,d\omega dp} {\sqrt{1+m^2\varepsilon^2}\cos(\omega\varepsilon)
  -\cos(p\varepsilon)-i\delta},\\
  \hspace{-0.5cm}\tilde A_2(x,t
  )&=
  \lim_{\delta\searrow 0}
  \frac{-i\varepsilon^2}{4\pi^2}
  \int_{-\pi/\varepsilon}^{\pi/\varepsilon}
  \int_{-\pi/\varepsilon}^{\pi/\varepsilon}
  \frac{\sqrt{1+m^2\varepsilon^2}\sin(\omega\varepsilon)+\sin(p\varepsilon)} {\sqrt{1+m^2\varepsilon^2}\cos(\omega\varepsilon)
  -\cos(p\varepsilon)-i\delta}
  e^{i p x-i\omega t}
  \,d\omega dp+\delta_{x0}\delta_{t0}.
  \end{aligned}
  \end{equation*}
\end{proposition}

\begin{proof}[Proof of Proposition~\ref{th-equivalence}]
This follows from Proposition~\ref{cor-double-fourier-anti} and~\cite[Proposition~17]{SU-22}, which states that the right-hand sides of Proposition~\ref{cor-double-fourier-anti} and Proposition~\ref{th-equivalence} are equal
(and in particular, the limits in Proposition~\ref{cor-double-fourier-anti} exist).
\end{proof}

Performing changes of variables
$(p,\omega)\mapsto (p\varepsilon,\omega\varepsilon)$,
$(\pi/\varepsilon-p,\pi/\varepsilon-\omega)$, $(\pm p,-\omega)$ in
the integrals from Proposition~\ref{cor-double-fourier-anti}, one gets
the following three immediate corollaries.

\begin{corollary}[Scaling symmetry] \label{cor-scaling}
For each 
$k\in\{1,2\}$
and
$(x,t)\in\varepsilon\mathbb{Z}^2$ 
we have $\widetilde{A}_k(x,t,m,\varepsilon)
=\widetilde{A}_k(x/\varepsilon,t/\varepsilon,m\varepsilon,1)$.
\end{corollary}


\begin{corollary}[Alternation of real and imaginary values]
\label{cor-alternation}
Let 
$k\in\{1,2\}$
and
$(x,t)\in\varepsilon\mathbb{Z}^2$.
If $(x+t)/\varepsilon+k$ is even (respectively, odd), then
$\widetilde{A}_k(x,t)$ is real (respectively, purely imaginary).
\end{corollary}


\begin{corollary}[Skew symmetry] \label{cor-symmetry}
For each 
$(x,t)\in\varepsilon\mathbb{Z}^2$, where $(x,t)\ne (0,0)$,
we have $\widetilde{A}_1(x,t
)=\widetilde{A}_1(x,-t
)=\widetilde{A}_1(-x,-t
)$ and
$\widetilde{A}_2(x,t
)=-\widetilde{A}_2(-x,-t
)$.
\end{corollary}

\begin{proof}[Proof of Theorem~\ref{th-well-defined}]
The existence and uniqueness of $A_k(x,t,m,\varepsilon,\delta)$
is 
Proposition~\ref{l-double-fourier-finite}.
The existence of limit~\eqref{eq-def-infinite-lattice-propagator} and the required equalities follow from Proposition~\ref{th-equivalence},
Corollary~\ref{cor-symmetry}, and \cite[Proposition~12]{SU-22},
which states that the integrals from
Proposition~\ref{th-equivalence} equal $a_1(x,t+\varepsilon,m,\varepsilon)$ and
$a_2(x+\varepsilon,t+\varepsilon,m, \varepsilon)$
for $t\ge 0$ and appropriate parity of $(x+t)/\varepsilon$.
The rest is Corollary~\ref{cor-alternation}.
\end{proof}


\begin{corollary}[Symmetry] \label{cor-skew-symmetry}
For each $(x,t)\in\varepsilon\mathbb{Z}^2$
we have $(t-x)\widetilde{A}_2(x,t
)=(t+x)\widetilde{A}_2(-x,t
)$.
\end{corollary}

\begin{proof}
Assume that $x\ne 0$.
Changing the sign of the
variable $p$ in Proposition~\ref{cor-double-fourier-anti} we get
$$
\widetilde{A}_2(-x,t)
=
\lim_{\delta\searrow 0}
  \frac{-i\varepsilon^2}{4\pi^2}
  \int_{-\pi/\varepsilon}^{\pi/\varepsilon}
  \int_{-\pi/\varepsilon}^{\pi/\varepsilon}
  \frac{\sqrt{1+m^2\varepsilon^2}\sin(\omega\varepsilon)-\sin(p\varepsilon)} {\sqrt{1+m^2\varepsilon^2}\cos(\omega\varepsilon)
  -\cos(p\varepsilon)-i\delta}
  e^{i p x-i\omega t}
  \,d\omega dp.
$$
Adding the expression for $\widetilde{A}_2(x,t)$ from Proposition~\ref{cor-double-fourier-anti}
and integrating by parts twice we get
\begin{multline*}
\widetilde{A}_2(x,t)+\widetilde{A}_2(-x,t)
=
\lim_{\delta\searrow 0}
  \frac{-i\varepsilon^2}{2\pi^2}
  \int_{-\pi/\varepsilon}^{\pi/\varepsilon}
  \int_{-\pi/\varepsilon}^{\pi/\varepsilon}
  \frac{\sqrt{1+m^2\varepsilon^2}\sin(\omega\varepsilon)} {\sqrt{1+m^2\varepsilon^2}\cos(\omega\varepsilon)
  -\cos(p\varepsilon)-i\delta}e^{i p x-i\omega t}
  \,d\omega dp
=\\=
\lim_{\delta\searrow 0}
  \frac{i\varepsilon^2}{2\pi^2}
  \int_{-\pi/\varepsilon}^{\pi/\varepsilon}
  \int_{-\pi/\varepsilon}^{\pi/\varepsilon}
  \frac{\sqrt{1+m^2\varepsilon^2}\sin(\omega\varepsilon)\sin(p\varepsilon)\varepsilon }
  {(\sqrt{1+m^2\varepsilon^2}\cos(\omega\varepsilon)
  -\cos(p\varepsilon)-i\delta)^2}\cdot\frac{e^{i p x-i\omega t}}{ix}
  \,d\omega dp
=\\=
\lim_{\delta\searrow 0}
  \frac{-i\varepsilon^2}{2\pi^2}
  \int_{-\pi/\varepsilon}^{\pi/\varepsilon}
  \int_{-\pi/\varepsilon}^{\pi/\varepsilon}
  \frac{\sin(p\varepsilon) (it) e^{i p x-i\omega t}
  \,d\omega dp}
  {(\sqrt{1+m^2\varepsilon^2}\cos(\omega\varepsilon)
  -\cos(p\varepsilon)-i\delta)(ix)}
=\frac{t}{x}
\left(\widetilde{A}_2(x,t)-\widetilde{A}_2(-x,t)\right).
\end{multline*}
Here in the second equality, we integrate the exponential
and differentiate the remaining factor with respect to $p$.
In the third equality, we differentiate the exponential
and integrate the remaining factor with respect to $\omega$.
The resulting  identity is equivalent to the required one.
\end{proof}

For the proof of Theorem~\ref{th-continuum-limit},
we halve the integration interval $[-\pi/\varepsilon,\pi/\varepsilon]$
in Proposition~\ref{th-equivalence}.

\begin{corollary} \label{cor-halve-interval}
For each 
$(x,t)\in\varepsilon\mathbb{Z}^2$, where $t\ge 0$, we have
\begin{align*}
  \widetilde{A}_1(x,t
  )&=
  \begin{cases}
  \mathrm{Re}\,\left(\frac{im\varepsilon^2}{\pi}
  \int_{-\pi/2\varepsilon}^{\pi/2\varepsilon}
  \frac{e^{i p x-i\omega_pt}\,dp}
  {\sqrt{m^2\varepsilon^2+\sin^2(p\varepsilon)}}\right),
  &\text{for }(x+t)/\varepsilon\text{ odd};\\
  \mathrlap{i\mathrm{Im}\,\left(\frac{im\varepsilon^2}{\pi}
  \int_{-\pi/2\varepsilon}^{\pi/2\varepsilon}
  \frac{e^{i p x-i\omega_pt}\,dp}
  {\sqrt{m^2\varepsilon^2+\sin^2(p\varepsilon)}}\right),}
  \hphantom{i\mathrm{Im}\,\left(\frac{\varepsilon}{\pi}\int_{-\pi/2\varepsilon}^{\pi/2\varepsilon}
  \left(1+
  \frac{\sin (p\varepsilon)} {\sqrt{m^2\varepsilon^2+\sin^2(p\varepsilon)}}\right) e^{ipx-i\omega_pt}\,dp
\right),}
  &\text{for }(x+t)/\varepsilon\text{ even};
  \end{cases}\\
  \widetilde{A}_2(x,t
  )&=
  \begin{cases}
  \mathrm{Re}\,\left(\frac{\varepsilon}{\pi}\int_{-\pi/2\varepsilon}^{\pi/2\varepsilon}
  \left(1+
  \frac{\sin (p\varepsilon)} {\sqrt{m^2\varepsilon^2+\sin^2(p\varepsilon)}}\right) e^{ipx-i\omega_pt}\,dp
\right),
  &\text{for }(x+t)/\varepsilon\text{ even};\\
  i\mathrm{Im}\,\left(\frac{\varepsilon}{\pi}\int_{-\pi/2\varepsilon}^{\pi/2\varepsilon}
  \left(1+
  \frac{\sin (p\varepsilon)} {\sqrt{m^2\varepsilon^2+\sin^2(p\varepsilon)}}\right) e^{ipx-i\omega_pt}\,dp
\right),
  &\text{for }(x+t)/\varepsilon\text{ odd}.
  \end{cases}
\end{align*}
\end{corollary}

\begin{proof}
This follows from Proposition~\ref{th-equivalence} by the change of variable $p\mapsto \pi/\varepsilon-p$.
For instance,
\begin{multline*}
 \widetilde{A}_1(x,t
  )
  =\frac{im\varepsilon^2}{2\pi}
  \int_{-\pi/2\varepsilon}^{\pi/2\varepsilon}
  \frac{e^{i p x-i\omega_pt}\,dp}
  {\sqrt{m^2\varepsilon^2+\sin^2(p\varepsilon)}}
  +
  \frac{im\varepsilon^2}{2\pi}
  \int_{\pi/2\varepsilon}^{3\pi/2\varepsilon}
  \frac{e^{i p x-i\omega_pt}\,dp}
  {\sqrt{m^2\varepsilon^2+\sin^2(p\varepsilon)}}
  =\\=\frac{im\varepsilon^2}{2\pi}
  \int_{-\pi/2\varepsilon}^{\pi/2\varepsilon}
  \frac{e^{i p x-i\omega_pt}\,dp}
  {\sqrt{m^2\varepsilon^2+\sin^2(p\varepsilon)}}
  +(-1)^{(x+t)/\varepsilon}
  \frac{im\varepsilon^2}{2\pi}
  \int_{-\pi/2\varepsilon}^{\pi/2\varepsilon}
  \frac{e^{-i p x+i\omega_pt}\,dp}
  {\sqrt{m^2\varepsilon^2+\sin^2(p\varepsilon)}},
  \quad\text{as required.}\\[-1.3cm]
\end{multline*}
\end{proof}

\subsection{Asymptotic formulae 
(Theorem~\ref{th-continuum-limit}, Corollary~\ref{cor-anti-uniform})}
\label{ssec-proofs-continuum}

In this subsection we prove Theorem~\ref{th-continuum-limit} and Corollary~\ref{cor-anti-uniform}.

First let us outline the plan of the argument.
We perform the Fourier transform and estimate the difference of the resulting oscillatory integrals for the discrete and continuum models, using tails exchange and non-stationary-phase method.  The proof of~\eqref{eq1-th-continuum-limit} consists of $3$ steps:
\begin{description}
  \item[Step 1:] we replace the integration interval in the Fourier integral for the continuum model by the one from the discrete model (cutting off large momenta);
  \item[Step 2:] we replace the phase in the discrete model by the one from the continuum model;
  \item[Step 3:] we estimate the difference of the resulting integrals for the continuum and discrete models
      for small and intermediate momenta.
\end{description}
In the proof of~\eqref{eq2-th-continuum-limit},
we first subtract the massless propagator
from the massive one to make the Fourier integral convergent.
The \blueit{required} integral for the continuum model is as follows.

\begin{lemma}\label{l-feynman-fourier} 
Under notation~\eqref{eq-feynman-propagator}, 
for each $m>0$, $t\ge 0$, and $x\ne \pm t$ we have
\begin{align*}
G^F_{11}(x,t,m)&=
\frac{im}{4\pi}
  \int_{-\infty}^{+\infty}
  \frac{e^{i p x-i\sqrt{m^2+p^2}t}\,dp}
  {\sqrt{m^2+p^2}},
\end{align*}
where the integral is understood as a conditionally convergent improper Riemann integral.
\end{lemma}

\begin{proof}
For $t>|x|$, use
the change of variables $q=tp-x\sqrt{m^2+p^2}$
and \cite[8.421.11, 8.405.1]{Gradstein-Ryzhik-63}:
$$
\frac{im}{4\pi}
  \int_{-\infty}^{+\infty}
  \frac{e^{i p x-i\sqrt{m^2+p^2}t}\,dp}
  {\sqrt{m^2+p^2}}
=\frac{im}{4\pi}
  \int_{-\infty}^{+\infty}
  \frac{e^{-i\sqrt{m^2(t^2-x^2)+q^2}}\,dq}
  {\sqrt{m^2(t^2-x^2)+q^2}}
=G^F_{11}(x,t,m).
$$
For $0\le t<|x|$, use
the change of variables $q=xp-t\sqrt{m^2+p^2}$
and \cite[8.432.5]{Gradstein-Ryzhik-63}.
\end{proof}


\begin{proof}[Proof of formula~\eqref{eq1-th-continuum-limit} in Theorem~\ref{th-continuum-limit}]
In the case when $|x|>|t|$ and $(x+t)/\varepsilon$ is odd,
formula~\eqref{eq1-th-continuum-limit} follows from Theorem~\ref{th-well-defined} and Definition~\ref{def-mass};
thus we exclude this case in what follows. We may assume that $x,t\ge 0$ because
~\eqref{eq1-th-continuum-limit} 
is invariant under the transformations $(x,t)\mapsto (\pm x,\pm t)$
by Corollaries~\ref{cor-symmetry}--\ref{cor-skew-symmetry}.
Assume that $(x+t)/\varepsilon$ is even; otherwise the proof
is the same up to an obvious modification of the very first inequality below.
Use notation~\eqref{eq-feynman-propagator}.
Formula~\eqref{eq1-th-continuum-limit} will follow from
\begin{multline*}
\left|
\widetilde{A}_1(x,t,m,\varepsilon)
-4\bluevar{i}\varepsilon\, \mathrm{Im}G^F_{11}(x,t,m)
\right|
\le\left|
  \frac{im\varepsilon^2}{\pi}
  \int\limits_{-\pi/2\varepsilon}^{\pi/2\varepsilon}
  \frac{e^{i p x-i\omega_pt}\,dp}
  {\sqrt{m^2\varepsilon^2+\sin^2(p\varepsilon)}}
- \frac{im\varepsilon}{\pi}
  \int\limits_{-\infty}^{+\infty}
  \frac{e^{i p x-i\sqrt{m^2+p^2}t}\,dp}
  {\sqrt{m^2+p^2}}
  \right|\\
\le
  \frac{m\varepsilon}{\pi}
  \left|
  \int_{|p|\ge\pi/2\varepsilon}
  \frac{e^{i p x-i\sqrt{m^2+p^2}t}\,dp}
  {\sqrt{m^2+p^2}}
  \right|
+ \frac{m\varepsilon}{\pi}
  \left|
  \int_{|p|\le \pi/2\varepsilon}
  \frac{\varepsilon\left(e^{i p x-i\omega_pt}
  -e^{i p x-i\sqrt{m^2+p^2}t}\right)\,dp}
  {\sqrt{m^2\varepsilon^2+\sin^2(p\varepsilon)}}
  \right| \\
+ \frac{m\varepsilon}{\pi}
  \left|
  \int_{|p|\le\pi/2\varepsilon}
  \left(
  \frac{\varepsilon}{\sqrt{m^2\varepsilon^2+\sin^2(p\varepsilon)}}
  -\frac{1}{\sqrt{m^2+p^2}}
  \right)
  e^{i p x-i\sqrt{m^2+p^2}t}\,dp
  \right|
=\\=m\varepsilon\, O\left(\frac{\varepsilon}{|x-t|}+\frac{m\varepsilon(x+t)}{s}\right)
+ m\varepsilon\, O\left(m^2t\varepsilon\right)
+ m\varepsilon\, O\left(\frac{\varepsilon}{|x-t|}+\frac{m\varepsilon(x+t)}{s}\right)
= m\varepsilon\, O\left(\varepsilon \Delta\right).
\end{multline*}
Here the first inequality follows from Corollary~\ref{cor-halve-interval}, Lemma~\ref{l-feynman-fourier}, and the inequality
$|\mathrm{Im}\,z-\mathrm{Im}\,w|\le |z-w|$. The second inequality
is straightforward. The obtained $3$ integrals are estimated below in Steps~1--3 respectively.
The last bound follows from $2m(x+t)/s\le \Delta$.
Below we restrict the integrals to $p\ge 0$;
the argument for 
$p<0$ is analogous.
The estimates are slightly different for $\varepsilon<s/m(x+t)$
and $\varepsilon>s/m(x+t)$.
Denote 
\begin{equation}\label{eq-eps-plus}
\begin{aligned}
\varepsilon_+&:=\max\{\varepsilon,s/m(x+t)\};\\
\varepsilon_-&:=\min\{\varepsilon,s/m(x+t)\}.
\end{aligned}
\end{equation}
We use the following known result.

\begin{lemma}[First derivative test] \label{l-first-derivative-test}
\textup{\cite[Lemma~5.1.2]{Huxley-96}}
Let $\alpha,\beta\in\mathbb{R}$ and $\alpha<\beta$.
Assume that $f\in C^1[\alpha,\beta]$ has monotone nonvanishing derivative;
then for each $g\in C^0[\alpha,\beta]$  we have
\begin{equation}\label{eq-l-first-derivative-test}
\left|\int_{\alpha}^{\beta}g(p)e^{if(p)}\,dp\right|
\le \frac{2\max_{[\alpha,\beta]}|g|+2V_\alpha^\beta (g)}
{\min_{[\alpha,\beta]}|f'|}.
\end{equation}
\end{lemma}

\textbf{Step 1.} Apply Lemma~\ref{l-first-derivative-test}
for $\alpha=\pi/2\varepsilon_-$, $\beta\to+\infty$, $f(p):=p x-\sqrt{m^2+p^2}\,t$,  and $g(p):=1/\sqrt{m^2+p^2}$.
The derivative $f'(p)$ is monotone because $f''(p)=-m^2t/(m^2+p^2)^{3/2}\le 0$.
Since $g(p)\searrow 0$ as $p\to +\infty$, it follows that
the numerator in the right side of~\eqref{eq-l-first-derivative-test}
tends to $4g(\pi/2\varepsilon_-)\le 8\varepsilon_-/\pi =O(\varepsilon)$ as $\beta\to+\infty$.
To bound the denominator from below (and in particular to check the assumption $f'(p)\ne 0$), we need a lemma.

\begin{lemma} \label{l-omegaprime} If
$p\ge \pi m(x+t)/2s$ then $|x-tp/\sqrt{m^2+p^2}|\ge|x-t|/4$.
\end{lemma}

\begin{proof}
We may assume that $x\ne0$. Since
$p\ge \pi m(x+t)/2s\ge \pi mx/2\sqrt{|t^2-x^2|}$
it follows that $m/p\le 1$ and ${m^2}/{p^2}=\eta{(t^2-x^2)}/{x^2}$ for some $\eta\in [-4/\pi^2;4/\pi^2]\subset[-1/2;1/2]$.
Thus
$$
\left|x-t\frac{p}{\sqrt{m^2+p^2}}\right|=
\frac{\left|x^2\left(1+m^2/p^2\right)-t^2\right|}
{x\left(1+m^2/p^2\right)+t\sqrt{1+m^2/p^2}}
=\frac{(1-\eta)|t^2-x^2|}{x\left(1+m^2/p^2\right)+t\sqrt{1+m^2/p^2}}
\ge 
\frac{|x-t|}{4}.
\vspace{-0.6cm}
$$
%
\end{proof}

Since $\alpha=\pi/2\varepsilon_-\ge \pi m(x+t)/2s$,
by Lemmas~\ref{l-first-derivative-test}--\ref{l-omegaprime}
we get
$\int_{\pi/2\varepsilon_-}^{+\infty}
  \frac{e^{i p x-i\sqrt{m^2+p^2}t}\,dp}{\sqrt{m^2+p^2}}
  =O\left(\frac{\varepsilon}{|x-t|}\right)
$.

For $\varepsilon\le s/m(x+t)$, which is equivalent to $\varepsilon_-=\varepsilon$, this completes step~1.

For $\varepsilon>s/m(x+t)$ we need an additional bound
$$
  \left|\int_{\pi/2\varepsilon}^{\pi/2\varepsilon_-}
  \frac{e^{i p x-i\sqrt{m^2+p^2}t}\,dp}
  {\sqrt{m^2+p^2}}\right|
  \le
  \int_{\pi/2\varepsilon}^{\pi/2\varepsilon_-}
  \frac{dp}
  {\pi/2\varepsilon}
  \le
  \frac{\pi/2\varepsilon_-}{\pi/2\varepsilon}
  =\frac{m\varepsilon(x+t)}{s}.
$$

\textbf{Step 2.} Using that $1-e^{iz}=O(|z|)$ for $z\in\mathbb{R}$
and $\sin z> z/2$ for $0<z<\pi/2$
we get
\begin{multline*}
  \int_{0}^{\pi/2\varepsilon}
  \frac{\varepsilon\left(e^{i p x-i\omega_pt}
  -e^{i p x-i\sqrt{m^2+p^2}t}\right)\,dp}
  {\sqrt{m^2\varepsilon^2+\sin^2(p\varepsilon)}}
= \int_{0}^{\pi/2\varepsilon}
  \frac{
  \left(1-e^{i\omega_pt-i\sqrt{m^2+p^2}t}\right)
  \varepsilon e^{i p x-i\omega_pt}\,dp}
  {\sqrt{m^2\varepsilon^2+\sin^2(p\varepsilon)}}
=\\= O\left(\int_{0}^{\pi/2\varepsilon}
  \frac{\left|\omega_p-\sqrt{m^2+p^2}\right|\varepsilon t \,dp}
  {\sqrt{m^2\varepsilon^2+\sin^2(p\varepsilon)}}\right)
= O\left(\int_{0}^{\pi/2\varepsilon}\frac{m^2\varepsilon^2t\sqrt{m^2+p^2}\,dp}
  {\sqrt{m^2+p^2}}\right)
= O\left(\frac{\pi}{2\varepsilon}\cdot m^2\varepsilon^2t\right)
= O(m^2 t\varepsilon).
\end{multline*}
Here the third estimate is proved in the following lemma.

\begin{lemma} \label{l-g-difference} For $|p|\le \pi/2\varepsilon$ we have
$$\omega_p=\sqrt{p^2+m^2}+O\left(m^2\varepsilon^2\sqrt{p^2+m^2}\right)
\qquad\text{
and
}\qquad
\frac{\partial \omega_p}{\partial p}
= \frac{\sin p\varepsilon}{\sqrt{\sin^2p\varepsilon+m^2\varepsilon^2}} =\frac{p}{\sqrt{p^2+m^2}}+O(m^2\varepsilon^2).
$$
\end{lemma}

\begin{proof}
First we estimate the derivative. By the Lagrange theorem,
there is $\varepsilon'\in[0,\varepsilon]$ such that
\begin{multline*}
\frac{\sin p\varepsilon}{\sqrt{\sin^2p\varepsilon+m^2\varepsilon^2}}
-\frac{p}{\sqrt{p^2+m^2}}
=
\left.
\frac{\partial }{\partial \varepsilon}\left(\frac{\sin p\varepsilon}{\sqrt{\sin^2p\varepsilon+m^2\varepsilon^2}}\right)
\right|_{\varepsilon=\varepsilon'}\cdot\varepsilon
=\\=
\left.\frac{m^2\varepsilon(p\varepsilon\cos p\varepsilon-\sin p\varepsilon)}{(\sin^2p\varepsilon+m^2\varepsilon^2)^{3/2}}
\right|_{\varepsilon=\varepsilon'}\cdot\varepsilon
=
O(m^2\varepsilon^2)
\end{multline*}
because $\sin z-z\cos z=O(z^3)$ and $\sin z> z/2$ for $0<z<\pi/2$.

Now we estimate $\omega_p$. By the Lagrange theorem,
there is $m'\in[0,m]$ such that
$$
\omega_0=\left.
\frac{\partial \omega_0}{\partial m}\right|_{m=m'}\cdot m
=\left.
\frac{1}{1+m^2\varepsilon^2}
\right|_{m=m'}\cdot m=m+O(m^3\varepsilon^2).
$$
Then by the estimate for the derivative $\frac{\partial \omega_p}{\partial p}$, for some $p'\in[0,p]$ we have
$$
\omega_p-\sqrt{p^2+m^2}=\omega_0-m+\left.
\left(\frac{\partial \omega_p}{\partial p}-\frac{p}{\sqrt{p^2+m^2}}\right)\right|_{p=p'}\cdot p
=O(m^2\varepsilon^2\sqrt{p^2+m^2}).
\vspace{-1.1cm}
$$
\end{proof}


\textbf{Step 3.}
We have
\begin{multline*}
\hspace{-1.2cm}\int\limits_{0}^{2\pi/\varepsilon_+}
  \left(
  \frac{\varepsilon}{\sqrt{m^2\varepsilon^2+\sin^2(p\varepsilon)}}
  -\frac{1}{\sqrt{m^2+p^2}}
  \right)
  e^{i p x-i\sqrt{m^2+p^2}t}\,dp
=O\left(\int\limits_{0}^{2\pi/\varepsilon_+}p\varepsilon^2\,dp\right)
=O\left(\frac{\varepsilon^2}{\varepsilon_+^2}\right)
=O\left(\frac{\varepsilon m(x+t)}{s}\right).
\end{multline*}
Here the first estimate is proved in the following lemma.

\begin{lemma} \label{l-difference} For $0\le p\le \pi/2\varepsilon$ we have
$\frac{\varepsilon}{\sqrt{\sin^2p\varepsilon+m^2\varepsilon^2}}-\frac{1}{\sqrt{p^2+m^2}}=O(p\varepsilon^2).
$
\end{lemma}

\begin{proof}
By the Lagrange theorem, for some $\varepsilon'\in (0;\varepsilon)$ we have
\begin{multline*}
\hspace{-0.5cm}  \frac{\varepsilon}{\sqrt{\sin^2p\varepsilon+m^2\varepsilon^2}} -\frac{1}{\sqrt{p^2+m^2}}
  = \left.\frac{\partial}{\partial\varepsilon} \frac{\varepsilon}{\sqrt{\sin^2p\varepsilon+m^2\varepsilon^2}}
  \right|_{\varepsilon=\varepsilon'} \cdot \varepsilon
  =
  \left.\frac{\sin p\varepsilon
  (\sin p\varepsilon - p\varepsilon\cos p\varepsilon)}
  {(\sin^2 p\varepsilon+m^2\varepsilon^2)^{3/2}}
  \right|_{\varepsilon=\varepsilon'} \cdot \varepsilon
  = O(p\varepsilon^2)
\end{multline*}
because $\sin z-z\cos z=O(z^3)$ and $z/2<\sin z<z$ for $0<z<\pi/2$.
\end{proof}

For $\varepsilon<s/m(x+t)$ we need an additional bound:
\begin{equation*}
  \int_{\pi/2\varepsilon_+}^{\pi/2\varepsilon}
  \left(
  \frac{\varepsilon}{\sqrt{m^2\varepsilon^2+\sin^2(p\varepsilon)}}
  -\frac{1}{\sqrt{m^2+p^2}}
  \right)
  e^{i p x-i\sqrt{m^2+p^2}t}\,dp
=O\left( \frac{\varepsilon}{|x-t|}\right)
\end{equation*}
obtained by applying Lemma~\ref{l-first-derivative-test}
for $g(p):={\varepsilon}/{\sqrt{m^2\varepsilon^2+\sin^2(p\varepsilon)}}-1/\sqrt{m^2+p^2}$
and $f(p)=p x-\sqrt{m^2+p^2}\,t$. The lower bound for the denominator of~\eqref{eq-l-first-derivative-test} is obtained by Lemma~\ref{l-omegaprime}. The numerator
is at most $4g(\pi/2\varepsilon)=O(\varepsilon+2\varepsilon/\pi)=O(\varepsilon)$ by the following lemma.
\end{proof}

\begin{lemma}
The function $g(p):={\varepsilon}/{\sqrt{m^2\varepsilon^2+\sin^2(p\varepsilon)}}-1/\sqrt{m^2+p^2}$
increases on $[0,\pi/2\varepsilon]$.
\end{lemma}

\begin{proof}
It suffices to prove that
$$
\frac{\partial g}{\partial p}
= -\frac{\varepsilon^2\sin(p\varepsilon)\cos(p\varepsilon)}{(m^2\varepsilon^2+\sin^2(p\varepsilon))^{3/2}}
  +\frac{p}{(m^2+p^2)^{3/2}}\ge 0.
$$
Since this expression clearly tends to $0$ as $\varepsilon\to 0$,
it suffices to prove that
$$
\frac{\partial^2 g}{\partial p\partial\varepsilon}
=
\frac{\varepsilon\left(2\sin^2 p \varepsilon
\left(p\varepsilon(1-\cos p\varepsilon)^2
+2(p\varepsilon-\sin p \varepsilon)\cos p\varepsilon\right)
+m^2 \varepsilon^2
(\sin 2 p \varepsilon-2p \varepsilon \cos 2 p\varepsilon)\right)}
{2(m^2\varepsilon^2+\sin^2(p\varepsilon))^{5/2}}
\ge 0.
$$
The equality and the inequality follow from \cite[\S3]{SU-3}
and $\sin z\le z\le \tan z$ for $z\in[0,\pi/2]$.
\end{proof}

This completes the proof of~\eqref{eq1-th-continuum-limit}.
For the proof of~\eqref{eq2-th-continuum-limit} we need two lemmas
establishing the Fourier integral for the continuum model.

\begin{lemma}\label{l-feynman-fourier2} 
Under notation~\eqref{eq-feynman-propagator} and \eqref{eq-massless}, for each
$m, t\ge 0$ and $x\ne \pm t$ we have
\begin{align*}
G^F_{12}(x,t,m)-G^F_{12}(x,t,0)&=
  \frac{1}{4\pi}
\int_{-\infty}^{+\infty}
  \left(
  \left(1+\frac{p}
  {\sqrt{m^2+p^2}}\right)
  e^{-i\sqrt{m^2+p^2}t}-
  \left(1+\mathrm{sgn}(p)
 \right)
  e^{-i|p|t}
  \right)
  e^{i p x}\,dp.
\end{align*}
where the integral is understood as conditionally convergent improper Riemann integral.
\end{lemma}

\begin{proof}
This is the limiting case $n\searrow 0$ of the following formula:
\begin{multline*}
G^F_{12}(x,t,m)-G^F_{12}(x,t,n)
=\left(
\frac{\partial}{\partial t}-\frac{\partial}{\partial x}
\right)
\left(\frac{G^F_{11}(x,t,m)}{m}-\frac{G^F_{11}(x,t,n)}{n}\right)
=\\
\hspace{-0.5cm}=
\int_{-\infty}^{+\infty}
\left(
\frac{\partial}{\partial t}-\frac{\partial}{\partial x}
\right)
  \left(\frac{e^{i p x-i\sqrt{m^2+p^2}t}}
  {\sqrt{m^2+p^2}}-
\frac{e^{i p x-i\sqrt{n^2+p^2}t}}
  {\sqrt{n^2+p^2}}
  \right)\frac{i \,dp}{4\pi}
=
\int_{-\infty}^{+\infty}
  \left(
  e^{ipx-i\sqrt{m^2+p^2}t}-
  e^{ipx-i\sqrt{n^2+p^2}t}
  \right)
  \,\frac{dp}{4\pi}
+\\+
\int_{-\infty}^{+\infty}
  \left(
  \frac{p}
  {\sqrt{m^2+p^2}}
  e^{-i\sqrt{m^2+p^2}t}-
  \frac{p}
  {\sqrt{n^2+p^2}}
  e^{-i\sqrt{n^2+p^2}t}
  \right)
  \frac{e^{i p x}\, dp}{4\pi}
\end{multline*}
Here we first applied \cite[8.473.4,5]{Gradstein-Ryzhik-63},
then Lemma~\ref{l-feynman-fourier}. We can change the order of the differentiation and the integration
(and pass to the limit $n\searrow 0$ under the integral) by \cite[Proposition~6 in \S2.3 in Ch.~7]{Zorich-97}
because the latter two integrals converge uniformly
on compact subsets of $\mathbb{R}^2\setminus\{|x|=|t|\}$ by the following lemma.
\end{proof}

\begin{lemma} \label{l-tail} For each $m,n,x,t\ge 0$, $\alpha>0$, and $x\ne t$ we have
\begin{align*}
\int_{\alpha}^{+\infty}
  \left(
  \frac{p}
  {\sqrt{m^2+p^2}}
  e^{-i\sqrt{m^2+p^2}t}-
  \frac{p}
  {\sqrt{n^2+p^2}}
  e^{-i\sqrt{n^2+p^2}t}
  \right)
  e^{i p x}\,dp
  &=O\left(\frac{(m^2+n^2)t}{\alpha|x-t|}\right);\\
 \int_{\alpha}^{+\infty}
  \left(
  e^{ipx-i\sqrt{m^2+p^2}t}-
  e^{ipx-i\sqrt{n^2+p^2}t}
  \right)
  \,dp&=O\left(\frac{(m^2+n^2)t}{\alpha|x-t|}\right).
\end{align*}

\end{lemma}

\begin{proof} Assume $m\ge n$ without loss of generality.
Let us prove the first formula; the second one is proved analogously.
Rewrite the integral as a sum of two ones:
\begin{multline*}
  \hspace{-0.5cm}\int_{\alpha}^{+\infty}
  \left(\frac{p}
  {\sqrt{m^2+p^2}}-\frac{p}
  {\sqrt{n^2+p^2}}\right)
  e^{ipx-i\sqrt{n^2+p^2}t}
  \,dp+
  \int_{\alpha}^{+\infty}
  \frac{p\left(
  e^{ipt-i\sqrt{m^2+p^2}t}-
  e^{ipt-i\sqrt{n^2+p^2}t}
  \right)
  e^{i p (x-t)}\,dp}
  {\sqrt{m^2+p^2}}.
\end{multline*}
%
%

The first integral is estimated immediately as
$\int_{\alpha}^{+\infty}(m^2-n^2)\,dp/p^2=O\left((m^2+n^2)/\alpha\right)$.

To estimate the second integral, apply Lemma~\ref{l-first-derivative-test}
for $\beta\to+\infty$, $f(p):=p(x-t)$, and
\begin{align*}
g(p)&:=
\frac{p}{\sqrt{m^2+p^2}}
\left(e^{ipt-i\sqrt{m^2+p^2}t}-e^{ipt-i\sqrt{n^2+p^2}t}\right).
\end{align*}
(Clearly, the lemma remains true for a complex-valued function $g$,
if 
the right side of~\eqref{eq-l-first-derivative-test} is multiplied by $2$.)
The right side of~\eqref{eq-l-first-derivative-test}
is not greater than
\begin{multline*}
\frac{4\int_{\alpha}^{+\infty}|g'(p)|\,dp}{|x-t|}
=\\=O\left(\int_{\alpha}^{+\infty}
\left(
\frac{m^2}{(m^2+p^2)^{3/2}}
\left(\sqrt{m^2+p^2}-\sqrt{n^2+p^2}\right)
+2-\frac{p}{\sqrt{n^2+p^2}}-\frac{p}{\sqrt{m^2+p^2}}\right)
\frac{t\,dp}{|x-t|}\right)
=\\=
O\left(\int_{\alpha}^{+\infty}
\left(1-\frac{p}{\sqrt{m^2+p^2}}\right)\frac{t\,dp}{|x-t|}\right)
=O\left(\int_{\alpha}^{+\infty}
\frac{m^2t\,dp}{|x-t| p^2}\right)
=O\left( \frac{(m^2+n^2)t}{|x-t|\alpha}\right),
\end{multline*}
where we use the Leibniz differentiation rule and the bounds
$e^{iz}-e^{iw}=O(|z-w|)$ for $z,w\in\mathbb{R}$,
$\frac{p}{\sqrt{m^2+p^2}}=O(1)$, and
$1-\frac{p}{\sqrt{m^2+p^2}}=O(\frac{m^2}{p^2})$.
\end{proof}

\begin{proof}[Proof of formula~\eqref{eq2-th-continuum-limit} in Theorem~\ref{th-continuum-limit}]
This is a modification of the proof of formula~\eqref{eq1-th-continuum-limit} above.
In particular, we use conventions from the first paragraph of that proof
except that now we assume that $(x+t)/\varepsilon$ is odd.
Use notation~\eqref{eq-feynman-propagator}, \eqref{eq-massless-lattice}, \eqref{eq-massless},
\eqref{eq-eps-plus}. Formula~\eqref{eq2-th-continuum-limit} follows from the estimates
obtained from Example~\ref{p-massless}, Corollary~\ref{cor-halve-interval},
Lemma~\ref{l-feynman-fourier2} and Steps~1--3 below:
\begin{multline*}
\hspace{-0.5cm}\left|
\widetilde{A}_2(x,t,m,\varepsilon)
-4\bluevar{i}\varepsilon \, \mathrm{Im}G^F_{12}(x,t,m)
\right|
= \left|
\widetilde{A}_2(x,t,m,\varepsilon)-\widetilde{A}_2(x,t,0,\varepsilon)
-4\bluevar{i}\varepsilon\, \mathrm{Im}\left(G^F_{12}(x,t,m)-G^F_{12}(x,t,0)\right)
\right|
\\
\le\left|
\frac{\varepsilon}{\pi}
  \int_{-\pi/2\varepsilon}^{\pi/2\varepsilon}
  \left(\left(1+\frac{\sin p\varepsilon }
  {\sqrt{m^2\varepsilon^2+\sin^2 p\varepsilon}}\right)
  e^{-i\omega_pt}
  -
  \left(1+\mathrm{sgn}(p)
 \right)
  e^{-i|p|t}
  \right)
  e^{i p x}\,dp
-\right.\\
-\left. \frac{\varepsilon}{\pi}
\int_{-\infty}^{+\infty}
  \left(
  \left(1+\frac{p}
  {\sqrt{m^2+p^2}}\right)
  e^{-i\sqrt{m^2+p^2}t}-
  \left(1+\mathrm{sgn}(p)
 \right)
  e^{-i|p|t}
  \right)
  e^{i p x}\,dp
  \right|\\
\le
  \frac{\varepsilon}{\pi}
  \left|\int_{|p|\ge \pi/2\varepsilon}
  \left(
  \left(1+\frac{p}
  {\sqrt{m^2+p^2}}\right)
  e^{-i\sqrt{m^2+p^2}t}-
  \left(1+\mathrm{sgn}(p)
 \right)
  e^{-i|p|t}
  \right)
  e^{i p x}\,dp
  \right|
+\\+
\frac{\varepsilon}{\pi}
  \left|
  \int_{|p|\le \pi/2\varepsilon}
  \left(
  1+\frac{\sin p\varepsilon}{\sqrt{m^2\varepsilon^2+\sin^2 p\varepsilon}}
  \right)
  \left(e^{-i\omega_pt}-e^{-i\sqrt{m^2+p^2}t}\right)
  e^{i p x}\,dp
  \right|
+\\+
\frac{\varepsilon}{\pi}
  \left|
  \int_{|p|\le \pi/2\varepsilon}
  \left(
  \frac{\sin p\varepsilon}{\sqrt{m^2\varepsilon^2+\sin^2p\varepsilon}}
  -\frac{p}{\sqrt{m^2+p^2}}
  \right)
  e^{i p x-i\sqrt{m^2+p^2}t}\,dp
  \right|
=\\=\varepsilon\, O\left(
\frac{m^2\varepsilon t}{|x-t|}\right)
+\varepsilon\, O\left(\frac{m^3\varepsilon (x+t)^2}{s}
+\frac{m^2\varepsilon t}{|x-t|}\right)
+ \varepsilon\, O\left(m^2\varepsilon\right)
= \frac{m\varepsilon(x+t)}{s}\,O\left(\varepsilon\Delta\right).
\end{multline*}
Below we restrict the integrals to $p\ge 0$;
the argument for 
$p<0$ is analogous.

\textbf{Step 1.} The integral over $p\ge 2\pi/\varepsilon$
is estimated in Lemma~\ref{l-tail} for $n=0$ and $\alpha=2\pi/\varepsilon$.



\textbf{Step 2.} By Lemma~\ref{l-g-difference} we have
\begin{multline*}
\hspace{-0.5cm}\int_{0}^{\pi/2\varepsilon_+}
  \left(
  1+\frac{\sin p\varepsilon}{\sqrt{m^2\varepsilon^2+\sin^2 p\varepsilon}}
  \right)
  \left(e^{-i\omega_pt}-e^{-i\sqrt{m^2+p^2}t}\right)
  e^{i p x}\,dp
=O\left(
\int_{0}^{\pi/2\varepsilon_+}
  \left|1-e^{it(\omega_p-\sqrt{m^2+p^2})}\right|
  \,dp\right)
=\\=O\left(
\int_{0}^{\pi/2\varepsilon_+}
m^2\varepsilon^2(p+m)t\,dp\right)
=O\left(m^2\varepsilon^2t
\left(\frac{1}{\varepsilon_+^2}+\frac{m}{\varepsilon_+}\right)\right)
=O\left(\frac{m^2\varepsilon t}{\varepsilon_+}+m^3\varepsilon t\right)
=O\left(\frac{m^3\varepsilon (x+t)^2}{s}\right).
\end{multline*}

For $\varepsilon<s/m(x+t)$ in addition apply Lemma~\ref{l-first-derivative-test}
for $\alpha:=2\pi/\varepsilon_+$, $\beta:=2\pi/\varepsilon$,
$f(p):=px-\sqrt{m^2+p^2}\,t$,
$g(p):=\left(
  1+\frac{\sin p\varepsilon}{\sqrt{m^2\varepsilon^2+\sin^2 p\varepsilon}}
  \right)
  \left(e^{i\sqrt{m^2+p^2}t-i\omega_pt}-1\right)$.
The maximum in~\eqref{eq-l-first-derivative-test} is
estimated analogously to the previous paragraph using
the inequality $\varepsilon<s/m(x+t)\le 1/m$:
\begin{multline*}
\max_{[\pi/2\varepsilon_+;\pi/2\varepsilon]}|g|
=O\left(\max_{p\in [\pi/2\varepsilon_+;\pi/2\varepsilon]}
m^2\varepsilon^2(p+m)t\right)
=O\left( m^2\varepsilon^2\left(\frac{2\pi}{\varepsilon}+m\right)t\right)
= O(m^2\varepsilon t).
\end{multline*}
The variation in~\eqref{eq-l-first-derivative-test} is
estimated using the Leibniz rule and Lemma~\ref{l-g-difference}:
\begin{multline*}
V_{\pi/2\varepsilon_+}^{\pi/2\varepsilon}(g)
=\int_{\pi/2\varepsilon_+}^{\pi/2\varepsilon}|g'(p)|\,dp
=\\=
O\left(\int_{\pi/2\varepsilon_+}^{\pi/2\varepsilon}
\left(\frac{m^2\varepsilon^3\cos p\varepsilon}{(m^2\varepsilon^2+\sin^2 p\varepsilon)^{3/2}}
\left|\omega_p-\sqrt{m^2+p^2}\right|+
\left|\frac{\partial\omega_p}{\partial p}
-\frac{\partial \sqrt{m^2+p^2}}{\partial p}\right|
\right)
t\,dp\right)
=\\=
O\left(\frac{\pi}{2\varepsilon}\cdot m^2\varepsilon^2 t\right)
=O\left(m^2\varepsilon t\right).
\end{multline*}
The denominator of~\eqref{eq-l-first-derivative-test}
is estimated using Lemma~\ref{l-omegaprime}.
Thus $\int_{\alpha}^{\beta}g(p)e^{if(p)}\,dp
=O\left(
{m^2\varepsilon t}/{|x-t|}\right)$.

\textbf{Step 3.} By Lemma~\ref{l-g-difference} we have
\begin{multline*}
\int_{0}^{\pi/2\varepsilon}
  \left(
  \frac{\sin p\varepsilon}{\sqrt{m^2\varepsilon^2+\sin^2p\varepsilon}}
  -\frac{p}{\sqrt{m^2+p^2}}
  \right)
  e^{i p x-i\sqrt{m^2+p^2}t}\,dp
=O\left(\frac{\pi}{2\varepsilon}\cdot m^2\varepsilon^2\right)
=O\left(m^2\varepsilon\right).\\[-1.6cm]
\end{multline*}
\end{proof}

\begin{proof}[Proof of Corollary~\ref{cor-anti-uniform}]
For $m>0$ this follows from Theorem~\ref{th-continuum-limit} because $\Delta$ is uniformly bounded and the limiting functions are uniformly continuous on each compact subset of $\mathbb{R}^2\setminus\{|x|=|t|\}$.
For $m=0$ this follows from Example~\ref{p-massless}.
\end{proof}

\subsection{Identities (Propositions~\ref{p-basis}--\ref{p-triple})}

\label{ssec-proofs-identities}

We first prove the results of \S\ref{ssec-identities}; then those of \S\ref{ssec-analytic} (except \bluevar{for} Proposition~\ref{th-equivalence} proved above).

\begin{proof}[Proof of Proposition~\ref{p-mass}]
This is obtained by substituting axiom~2 into axiom~1 in Definition~\ref{def-anti-alg}
and passing to the limit $\delta\searrow 0$ (cf.~\eqref{eq-substituted}).
\end{proof}

Proposition~\ref{p-Klein-Gordon-mass} is deduced from the
previous 
one similarily to
\cite[Proof of Proposition~7]{SU-22}.

\begin{proof}[Proof of Proposition~\ref{p-Klein-Gordon-mass}]
Substituting $t=t+\varepsilon$ in Proposition~\ref{p-mass},
Eq.~\eqref{eq-Dirac-source1}, we get
$$
\tilde A_1(x,t+\varepsilon
) = \frac{1}{\sqrt{1+m^2\varepsilon^2}}
(\tilde A_1(x+\varepsilon,t
)
+ m \varepsilon\, \tilde A_2(x,t
)).
$$
Changing the signs of both $x$ and $t$ and applying Corollary~\ref{cor-symmetry}
we get for $(x,t)\ne (0,0)$
$$
\tilde A_1(x,t-\varepsilon
) = \frac{1}{\sqrt{1+m^2\varepsilon^2}}
(\tilde A_1(x-\varepsilon,t
)
- m \varepsilon\, \tilde A_2(x,t
)).
$$
Adding the resulting two identities we get the required identity for $k=1$.

The one for $k=2$ is proved analogously but we start with~\eqref{eq-Dirac-source2}.
The analogues of the above two identities hold for
$(x,t)\ne (0,-\varepsilon)$ and $(x,t)\ne (-\varepsilon,0)$ respectively.
\end{proof}

\begin{proof}[Proof of Proposition~\ref{p-symmetry-mass}]
This has been proved in Corollaries~\ref{cor-symmetry} and~\ref{cor-skew-symmetry}.
%
\end{proof}

\begin{proof}[Proof of Proposition~\ref{p-mass2}]
By Proposition~\ref{th-equivalence} and the Plancherel theorem we get
\begin{multline*}
\sum\limits_{x\in\varepsilon\mathbb{Z}}
\left(|\widetilde{A}_1\left(x,t
\right)|^2+
|\widetilde{A}_2\left(x,t
\right)|^2
\right)
=\\=\frac{\varepsilon}{2\pi}
  \int_{-\pi/\varepsilon}^{\pi/\varepsilon}
  \left(\left|
  \frac{m\varepsilon e^{-i\omega_pt}}
  {\sqrt{m^2\varepsilon^2+\sin^2(p\varepsilon)}}
  \right|^2
  +
  \left|
  \left(1+
  \frac{\sin (p\varepsilon)}
  {\sqrt{m^2\varepsilon^2+\sin^2(p\varepsilon)}}\right)
  e^{-i\omega_pt}
  \right|^2\right)\,dp
 =\\ =\frac{\varepsilon}{2\pi}
 \int_{-\pi/\varepsilon}^{\pi/\varepsilon}
  \left(2+\frac{2\sin(p\varepsilon)}
  {\sqrt{m^2\varepsilon^2+\sin^2(p\varepsilon)}}\right)
  \,dp
=2,
\end{multline*}
because the second summand in the latter integral is an odd function in $p$.
\end{proof}

\begin{proof}[Proof of Proposition~\ref{p-Huygens2}]
This follows from Proposition~\ref{th-equivalence} and the convolution theorem,
because
\begin{align*}
  2\frac{im\varepsilon e^{-i\omega_pt}}
  {\sqrt{m^2\varepsilon^2+\sin^2(p\varepsilon)}}
  &=  \left(1+
  \frac{\sin (p\varepsilon)}
  {\sqrt{m^2\varepsilon^2+\sin^2(p\varepsilon)}}\right)
  e^{-i\omega_pt'}\cdot
  \frac{im\varepsilon e^{-i\omega_p(t-t')}}
  {\sqrt{m^2\varepsilon^2+\sin^2(p\varepsilon)}}
  \\&+
  \frac{im\varepsilon e^{-i\omega_pt'}}
  {\sqrt{m^2\varepsilon^2+\sin^2(p\varepsilon)}}
  \cdot \left(1-
  \frac{\sin (p\varepsilon)}
  {\sqrt{m^2\varepsilon^2+\sin^2(p\varepsilon)}}\right)
  e^{-i\omega_p(t-t')},\\
\hspace{-0.5cm}  2\left(1+
  \frac{\sin (p\varepsilon)}
  {\sqrt{m^2\varepsilon^2+\sin^2(p\varepsilon)}}\right)
  e^{-i\omega_pt}
  &=\left(1+
  \frac{\sin (p\varepsilon)}
  {\sqrt{m^2\varepsilon^2+\sin^2(p\varepsilon)}}\right)
  e^{-i\omega_pt'}\cdot
  \left(1+
  \frac{\sin (p\varepsilon)}
  {\sqrt{m^2\varepsilon^2+\sin^2(p\varepsilon)}}\right)
  e^{-i\omega_p(t-t')}
  \\&- \frac{im\varepsilon e^{-i\omega_pt'}}
  {\sqrt{m^2\varepsilon^2+\sin^2(p\varepsilon)}}
  \cdot \frac{im\varepsilon e^{-i\omega_p(t-t')}}
  {\sqrt{m^2\varepsilon^2+\sin^2(p\varepsilon)}}.
  \\[-1.3cm]
\end{align*}
\end{proof}

For the next proposition we need a lemma, which follows from
Definition~\ref{def-mass} and Theorem~\ref{th-well-defined}.

\begin{lemma}[Initial value] \label{l-initial}
For $k+x/\varepsilon$ even we have
$\widetilde{A}_k(x,0)=\delta_{k2}\delta_{x0}$.
\end{lemma}

\begin{proof}[Proof of Proposition~\ref{p-Huygens}]
The proof is by induction on $t$.
The base $t=t'$ is Lemma~\ref{l-initial}.
The inductive step follows from
\begin{multline*}
\widetilde{A}_1(x,t+\varepsilon)
=\frac{1}{\sqrt{1+m^2\varepsilon^2}}
(\tilde A_1(x+\varepsilon,t
)
+ m \varepsilon\, \tilde A_2(x,t
))
=\\=\sum\limits_{\substack{x'\in\varepsilon\mathbb{Z}:\\
(x+\varepsilon+x'+t+t')/\varepsilon \text{ odd}}}
\widetilde{A}_2(x',t')\widetilde{A}_1(x-x'+\varepsilon,t-t')
+\sum\limits_{\substack{x'\in\varepsilon\mathbb{Z}:\\
(x+\varepsilon+x'+t+t')/\varepsilon \text{ even}}}
\widetilde{A}_1(x',t')\widetilde{A}_2(x'-x-\varepsilon,t-t')
+\\+ m \varepsilon
\sum\limits_{\substack{x'\in\varepsilon\mathbb{Z}:\\
(x+\varepsilon+x'+t+t')/\varepsilon \text{ odd}}}
\widetilde{A}_2(x',t')\widetilde{A}_2(x-x',t-t')
 - m \varepsilon \sum\limits_{\substack{x'\in\varepsilon\mathbb{Z}:\\
(x+\varepsilon+x'+t+t')/\varepsilon \text{ even}}}
 \widetilde{A}_1(x',t')\widetilde{A}_1(x'-x,t-t')
=\\=
\sum\limits_{\substack{x'\in\varepsilon\mathbb{Z}:\\
(x+x'+t+\varepsilon+t')/\varepsilon \text{ odd}}}
\widetilde{A}_2(x',t')\widetilde{A}_1(x-x',t+\varepsilon-t')
+\sum\limits_{\substack{x'\in\varepsilon\mathbb{Z}:\\
(x+x'+t+\varepsilon+t')/\varepsilon \text{ even}}}
\widetilde{A}_1(x',t')\widetilde{A}_2(x'-x,t+\varepsilon-t')
\end{multline*}
and an analogous computation for $\widetilde{A}_2(x,t+\varepsilon)$.
Here the first and the last equality follow from Proposition~\ref{p-mass},
and the middle equality is the inductive hypothesis.
\end{proof}

\begin{proof}[Proof of Proposition~\ref{l-mean}]
Assume that $t\ge 0$; otherwise
perform the transformation $(x,t)\mapsto(-x,-t)$,
which preserves~\eqref{eq-mean1}--\eqref{eq-mean2}
by Proposition~\ref{p-symmetry-mass}.

Identity~\eqref{eq-mean2} is then obtained from
Propositions~\ref{p-symmetry-mass} and~\ref{th-equivalence}
as follows:
\begin{multline*}
2m\varepsilon x\widetilde{A}_2(x,t)=
m\varepsilon(t+x)\left(\widetilde{A}_2(x,t)-\widetilde{A}_2(-x,t)\right)
=\frac{(t+x)m\varepsilon^2}{\pi}\int_{-\pi/\varepsilon}^{\pi/\varepsilon}
\frac{\sin (p\varepsilon) e^{ipx-i\omega_pt}\,dp}
{\sqrt{m^2\varepsilon^2+\sin^2(p\varepsilon)}}
=\\= \frac{-i(t+x)m\varepsilon^2}{2\pi}
\int_{-\pi/\varepsilon}^{\pi/\varepsilon}
  \frac{\left(e^{ip(x+\varepsilon)}-e^{ip(x-\varepsilon)}\right)e^{-i\omega_pt}\,dp}
  {\sqrt{m^2\varepsilon^2+\sin^2(p\varepsilon)}}
=(t+x)\left(\widetilde{A}_1(x-\varepsilon,t)
-\widetilde{A}_1(x+\varepsilon,t)\right).
\end{multline*}

To prove~\eqref{eq-mean1}, apply
Proposition~\ref{th-equivalence} and integrate by parts:
\mscomm{!!! ref to Wolfram !!!}
\begin{align}
\notag
2m\varepsilon x\widetilde{A}_1(x,t)&=
\frac{m^2\varepsilon^3}{\pi}
\int_{-\pi/\varepsilon}^{\pi/\varepsilon}
\frac{\left(ixe^{ipx}\right)e^{-i\omega_pt}\,dp}
  {\sqrt{m^2\varepsilon^2+\sin^2(p\varepsilon)}}
=-\frac{m^2\varepsilon^3}{\pi}
\int_{-\pi/\varepsilon}^{\pi/\varepsilon}
e^{ipx}\frac{\partial}{\partial p}\left(\frac{e^{-i\omega_pt}}
  {\sqrt{m^2\varepsilon^2+\sin^2(p\varepsilon)}}\right)\,dp
\\ \label{eq-proof-mean1}
&=
\frac{m^2\varepsilon^3}{\pi}
\int_{-\pi/\varepsilon}^{\pi/\varepsilon}
\frac{\sin (p\varepsilon)e^{ipx-i\omega_pt}}
  {m^2\varepsilon^2+\sin^2(p\varepsilon)}
\left(it+\frac{\varepsilon\cos(p\varepsilon)}
  {\sqrt{m^2\varepsilon^2+\sin^2(p\varepsilon)}}\right)
  \,dp;
\\ \notag
(x-t)\widetilde{A}_2(x,t)&=
\frac{-i\varepsilon}{2\pi}
\int_{-\pi/\varepsilon}^{\pi/\varepsilon}
\left(i(x-t)e^{ip(x-t)}\right)e^{-i(\omega_p-p)t}
\left(1+\frac{\sin(p\varepsilon)}
{\sqrt{m^2\varepsilon^2+\sin^2(p\varepsilon)}}\right)\,dp
\\ \notag
&=
\frac{i\varepsilon}{2\pi}
\int_{-\pi/\varepsilon}^{\pi/\varepsilon}
e^{ip(x-t)}\frac{\partial}{\partial p}
\left(e^{-i(\omega_p-p)t}
\left(1+\frac{\sin(p\varepsilon)}
{\sqrt{m^2\varepsilon^2+\sin^2(p\varepsilon)}}\right)\right)\,dp
\\ \label{eq-proof-mean2}
&=
\frac{im^2\varepsilon^3}{2\pi}
\int_{-\pi/\varepsilon}^{\pi/\varepsilon}
\frac{e^{ipx-i\omega_pt}\,dp}
  {m^2\varepsilon^2+\sin^2(p\varepsilon)}
\left(it+\frac{\varepsilon\cos(p\varepsilon)}
  {\sqrt{m^2\varepsilon^2+\sin^2(p\varepsilon)}}\right).
\end{align}
Substituting $x\pm\varepsilon$ for $x$ in~\eqref{eq-proof-mean2},
subtracting the resulting equalities, and adding~\eqref{eq-proof-mean1},
we get~\eqref{eq-mean1}.
\end{proof}

\begin{proof}[Proof of Proposition~\ref{p-triple}]
The first required identity follows from Proposition~\ref{l-mean}
by substituting $x\pm\varepsilon$ for $x$ in~\eqref{eq-mean2}
and inserting the resulting expressions into~\eqref{eq-mean1}.
The second one is obtained by the same argument with \eqref{eq-mean2} and
\eqref{eq-mean1} interchanged.
\end{proof}

\begin{proof}[Proof of Example~\ref{p-massless}]
This follows directly from Proposition~\ref{th-equivalence}.
\end{proof}

\begin{proof}[Proof of Example~\ref{ex-simplest-values}]
Eq.~\eqref{eq-Gauss}--\eqref{eq-lemniscate} are checked directly.
\mscomm{!!! ref to Wolfram !!!}
Table~\ref{table-a4} is filled inductively using Lemma~\ref{l-initial}
and Propositions~\ref{p-symmetry-mass}, \ref{l-mean}, \ref{p-mass}.
\end{proof}

\begin{proof}[Proof of Proposition~\ref{p-basis}]
The value $2^{|t|/2}\mathrm{Re}\,\widetilde{A}_k(x,t,1,1)$
is an integer by Theorem~\ref{th-well-defined} and Definition~\ref{def-mass}.
It remains to prove that $2^{|t|/2}\mathrm{Im}\,\widetilde{A}_k(x,t,1,1)$
is a rational linear combination of $G$ and $L'$ for $x+t+k$ odd;
otherwise the expression vanishes by Theorem~\ref{th-well-defined}.
By Proposition~\ref{p-symmetry-mass} we may assume that $x,t\ge 0$.

The proof is by induction on $t$. The base $t=0$ is proved by induction on $x$.
The base $(x,t)=(0,0)$ and $(1,0)$ is Example~\ref{ex-simplest-values}.
The step from $x$ to $x+1$ follows from Proposition~\ref{l-mean}.
Thus the assertion holds for $t=0$ and each $x$.
The step from $t$ to $t+1$ follows from Proposition~\ref{p-mass}.
\end{proof}

\begin{proof}[Proof of Proposition~\ref{p-mass5}]
It suffices to prove the proposition for $x,t\ge 0$ and $\varepsilon=1$.
Indeed, for $\varepsilon\ne 1$ perform the transformation
$(x,t,m,\varepsilon)\mapsto (x/\varepsilon,t/\varepsilon,m\varepsilon,1)$ which
preserves the required formulae by Corollary~\ref{cor-scaling}.
For $x<0$ change the sign of $x$. The left sides transform as shown in Proposition~\ref{p-symmetry-mass}.
By the identity $\binom{n}{k}=\frac{n}{n-k}\binom{n-1}{k}$
it follows that the right sides transform in the same way as the left sides.
For $t<0$ change the sign of $t$.
The left sides transform as in Proposition~\ref{p-symmetry-mass}.
By the Euler transformation
${}_2F_1\left(p,q;r;z\right)
=(1-z)^{r-p-q}{}_2F_1\left(r-p,r-q;r;z\right)
$ \cite[9.131.1]{Gradstein-Ryzhik-63}
and the identity $\binom{n}{k}=(-1)^{k}\binom{k-n-1}{k}$,
the right sides transform in the same way.

For $x,t\ge 0$ and $\varepsilon=1$ the proof is by induction on $t$.

Induction base: $t=0$. To compute $\widetilde{A}_k(x,0)$, consider the following $3$ cases:

Case 1: $x+k$ even. The required formula holds by Lemma~\ref{l-initial}
and the identities $\binom{(x+k-2)/2}{x}=\delta_{k2}\delta_{x0}$
and ${}_2F_1(0,1;1;z)=1$ for each $k\in\{1,2\}$, $0\le x\in\mathbb{Z}$, $z\in \mathbb{R}$.

Case 2: $x$ even, $k=1$. Recall that $\varepsilon=1$. Then the required identity follows from
\begin{align*}
\hspace{-0.7cm}
\widetilde{A}_1(x,0)
&=\frac{im}{2\pi}
  \int_{-\pi}^{\pi}
  \frac{e^{i p x}\,dp}
  {\sqrt{m^2+\sin^2p}}
=\frac{im}{\pi}
  \int_{0}^{2\pi}
  \frac{ \sqrt{z}\cos(q x/2)\,dq}
  {\sqrt{1-2z\cos q+z^2}}
=\frac{4im z^{(x+1)/2}}
{xB(1/2,x/2)}\cdot
{}_2F_1\left(\frac{1}{2},\frac{x+1}{2};
\frac{x}{2}+{1};z^2\right)
\\
&=
im (-1)^{x/2}2^{x+1}z^{(x+1)/2}
\binom{(x-1)/2}{x}
{}_2F_1\left(\frac{x+1}{2},\frac{1}{2};
\frac{x}{2}+{1};z^2\right)
\\
&= \frac{im (-1)^{x/2}2^{x+1} z^{(x+1)/2}}
{(1-z)^{x+1}}
\binom{(x-1)/2}{x}
{}_2F_1\left(\frac{x+1}{2},
\frac{x+1}{2};
1+{x};\frac{-4z}{(1-z)^2}\right)
\\
&=
i(-im)^{-x}
\binom{{(x-1)}/{2}}{x}
{}_2F_1\left(\frac{x+1}{2},
\frac{x+1}{2}; 1+x; -\frac{1}{m^2}\right).
\end{align*}
Here the first equality is Proposition~\ref{th-equivalence}.
The second equality is obtained by the change of the integration
variable $q:=2p$, a transformation the denominator
using the notation
$z:=(\sqrt{1+m^2}-m)^2$,
dropping the odd function containing $\sin (qx/2)$,
and halving the integration interval for the remaining even function.
The third equality is obtained by applying
\cite[9.112]{Gradstein-Ryzhik-63}
for $p=1/2$ and $n=x/2$.
The fourth equality is obtained by evaluation of the beta-function
\cite[8.384.1,8.339.1--2]{Gradstein-Ryzhik-63}
and applying \cite[9.131.1]{Gradstein-Ryzhik-63}.
The fifth equality is obtained by applying
\cite[9.134.3]{Gradstein-Ryzhik-63} (with the sign of $z$ changed).
The last equality follows from $4z/(1-z)^2=1/m^2$.

Case 3: $x$ odd, $k=2$. By Case~2, Proposition~\ref{l-mean}, and
\cite[9.137.15]{Gradstein-Ryzhik-63} we get the identity
\begin{multline*}
\widetilde{A}_2(x,0)
=\frac{1}{2m}\widetilde{A}_1(x-1,0)
-\frac{1}{2m}\widetilde{A}_1(x+1,0)
=
\frac{i(-im)^{1-x}}{2m}
\binom{x/{2}-1}{x-1}
{}_2F_1\left(\frac{x}{2},
\frac{x}{2}; x; -\frac{1}{m^2}\right)
-\\-
\frac{i(-im)^{-x-1}}{2m}
\binom{x/{2}}{x+1}
{}_2F_1\left(\frac{x}{2}+1,
\frac{x}{2}+1; x+2; -\frac{1}{m^2}\right)
=
(-im)^{-x}
\binom{x/{2}}{x}
{}_2F_1\left(\frac{x}{2},
1+\frac{x}{2}; 1+x; -\frac{1}{m^2}\right).
\end{multline*}

Induction step. Using Proposition~\ref{p-mass}, the inductive hypothesis,
and \cite[9.137.11]{Gradstein-Ryzhik-63} we get
\begin{multline*}
\tilde A_1(x,t
) = \frac{1}{\sqrt{1+m^2}}
(\tilde A_1(x+1,t-1
)
+ m\, \tilde A_2(x,t-1
))
=\\=
i\left({1+m^2}\right)^{-\frac{t}{2}}
(-im)^{{t-x-2}}\binom{\frac{t+x-1}{2}}{x+1}
{}_2F_1\left(\frac{x-t+3}{2},\frac{x-t+3}{2};x+2;-\frac{1}{m^2}\right)+
\\+
m\left({1+m^2}\right)^{-\frac{t}{2}}
(-im)^{t-x-1}\binom{\frac{t+x-1}{2}}{x}
{}_2F_1\left(\frac{x-t+1}{2},\frac{x-t+3}{2};x+1;-\frac{1}{m^2}\right)
=\\=
i\left({1+m^2}\right)^{-\frac{t}{2}}
(-im)^{{t-x}}\binom{\frac{t+x-1}{2}}{x}
{}_2F_1\left(\frac{x-t+1}{2},\frac{x-t+1}{2};x+1;-\frac{1}{m^2}\right).
\end{multline*}
For $\tilde A_2(x,t)$ the step is analogous; it uses
~\cite[9.137.18]{Gradstein-Ryzhik-63} for $x\ne 0$
and~\cite[9.137.12]{Gradstein-Ryzhik-63} for $x\!=\!0$.
\end{proof}

\subsection{Combinatorial definition 
(Theorem~\ref{p-real-imaginary},
Propositions~\ref{p-initial}--\ref{cor-dirac})}
\label{ssec-proofs-fourier}

In this section we compute full space-time Fourier transform of the finite-lattice propagator
(Proposition~\ref{cor-double-fourier-finite}) \blueit{and}
use it to prove some identities
(Corollary~\ref{cor-symmetry-finite}, Propositions~\ref{p-initial}--\ref{cor-dirac})
and Theorem~\ref{p-real-imaginary}.
We follow the classical approach known from \blueit{the}
Kirchhoff matrix-tree theorem, the Kasteleyn and Kenyon theorems \cite{Levine-11,Kenyon-02}.
Namely, the solution of Dirac equation on the finite lattice \blueit{(Proposition~\ref{p-dirac-finite})}
is expressed through determinants, interpreted combinatorially
via loop expansion, and computed explicitly via Fourier transform.

\begin{notation*}
Let $e_1=e_1(x,t)\perp(1,1)$ and $e_2=e_2(x,t)\parallel (1,1)$
be the two edges ending at a lattice point $(x,t)$;
cf.~Figure~\ref{fig-1x1} to the right. Denote $b_k:=e_k(0,0)$ and
$x^*:=x+\frac{\varepsilon}{2}$, $t^*:=t+\frac{\varepsilon}{2}$.
\end{notation*}

\begin{proposition}[Full space-time Fourier transform] \label{cor-double-fourier-finite}
The denominator of~\eqref{eq-def-finite-lattice-propagator} is nonzero.
For each 
even lattice point $(x,t)$ we have
\begin{align}\notag
  A(b_2\to e_1)&=
  \frac{-i}{2T^2}
\sum_{p,\omega\in
\faktor{(2\pi/T\varepsilon)\mathbb{Z}}{(2\pi/\varepsilon)\mathbb{Z}}}
  \frac{m\varepsilon -i\delta\varepsilon e^{ip\varepsilon}}
  {\sqrt{1-\delta^2\varepsilon^2}\sqrt{1+m^2\varepsilon^2}\cos(\omega\varepsilon)
  -\cos(p\varepsilon)-im\varepsilon^2\delta}e^{i p x-i\omega t},\\ \notag
  A(b_2\to e_2)&=
\frac{-i}{2T^2}
\sum_{p,\omega\in
\faktor{(2\pi/T\varepsilon)\mathbb{Z}}{(2\pi/\varepsilon)\mathbb{Z}}}
  \frac{\sqrt{1-\delta^2\varepsilon^2}\sqrt{1+m^2\varepsilon^2}\sin(\omega\varepsilon)+\sin(p\varepsilon)}
  {\sqrt{1-\delta^2\varepsilon^2}\sqrt{1+m^2\varepsilon^2}\cos(\omega\varepsilon)
  -\cos(p\varepsilon)-im\varepsilon^2\delta}
  e^{i p x-i\omega t}
  +\frac{1}{2}\delta_{x0}\delta_{t0}.\\
\intertext{For each 
odd lattice point $(x,t)$ we have}
\notag
  A(b_2\to e_1)&=
 \frac{-i}{2T^2}
 \sum_{p,\omega\in
 \faktor{(2\pi/T\varepsilon)\mathbb{Z}}{(2\pi/\varepsilon)\mathbb{Z}}}
 \frac{m\varepsilon\sqrt{1-\delta^2\varepsilon^2} e^{i\omega\varepsilon}-i\delta\varepsilon\sqrt{1+m^2\varepsilon^2}}
 {\sqrt{1-\delta^2\varepsilon^2}\sqrt{1+m^2\varepsilon^2}\cos(\omega\varepsilon)
  -\cos(p\varepsilon)-im\varepsilon^2\delta}e^{i p x^*-i\omega t^*},\\
\notag
  A(b_2\to e_2)&=
  \frac{1}{2T^2}
\sum_{p,\omega\in
\faktor{(2\pi/T\varepsilon)\mathbb{Z}}{(2\pi/\varepsilon)\mathbb{Z}}}
  \frac{\sqrt{1+m^2\varepsilon^2}e^{-ip\varepsilon}
  -\sqrt{1-\delta^2\varepsilon^2}e^{i\omega\varepsilon}}
  {\sqrt{1-\delta^2\varepsilon^2}\sqrt{1+m^2\varepsilon^2}\cos(\omega\varepsilon)
  -\cos(p\varepsilon)-im\varepsilon^2\delta}
  e^{i p x^*-i\omega t^*}.
\end{align}
\end{proposition}

The proposition follows from the next $3$ lemmas.
The first one is proved completely analogously to Proposition~\ref{l-double-fourier-finite},
only the Fourier series is replaced by the discrete Fourier transform.

\begin{lemma}[Full space-time Fourier transform] \label{l-discrete-fourier}
There exists a unique pair of functions $A_k(x,t)$ on the lattice of size $T$
satisfying axioms~1--2 from Definition~\ref{def-anti-alg}. It is given
by the expressions from Proposition~\ref{l-double-fourier-finite}, only
the integrals are replaced by the sums over
$\faktor{(2\pi/T\varepsilon)\mathbb{Z}}{(2\pi/\varepsilon)\mathbb{Z}}$,
and the factors $\varepsilon^2/4\pi^2$ are replaced by $1/T^2$.
\end{lemma}

For combinatorial interpretation, we pass from functions
on the lattice to functions on edges.

\begin{lemma}[Equivalence of equations; \bluevar{cf.~Proposition~\ref{p-dirac-finite}}] \label{l-equivalence-equations}
Functions $A_k(x,t)$ on the lattice of size $T$
satisfy axioms~1--2 from Definition~\ref{def-anti-alg}
if and only if the function
$\alpha(e_k(x,t)):=-i^k A_k(x,t)/2$ on the set of edges
satisfies the equation
\begin{equation}\label{eq-l-equivalence-equations}
\alpha(f)= \alpha(e)A(ef)+\alpha(e')A(e'f)+\delta_{b_2f}
\end{equation}
for each edge $f$, where $e\parallel f$ and $e'\perp f$ are the two edges ending at
the starting point of $f$.
\end{lemma}

\begin{proof} Assume $f=e_2(x,t)$ and $(x,t)$ is even; the other cases are analogous.
Then
\begin{align*}
e &=e_2\left(x-\frac{\varepsilon}{2},t-\frac{\varepsilon}{2}\right),
& A(ef)&=\frac{1}{\sqrt{1+m^2\varepsilon^2}},
& \delta_{b_2f}&=\delta_{x0}\delta_{t0},
\\
e'&=e_1\left(x-\frac{\varepsilon}{2},t-\frac{\varepsilon}{2}\right),
& A(e'f)&=\frac{-im\varepsilon}{\sqrt{1+m^2\varepsilon^2}}.
& &
\end{align*}
Substituting $\alpha(e_k(x,t))=-i^k A_k(x,t)/2$,
we get the second equation of axiom~2.
\end{proof}

Now we solve the system of equations~\eqref{eq-l-equivalence-equations}
by Cramer's rule.

\begin{lemma}[Loop expansion]\label{l-loop-expansion}
Define two matrices with the rows and columns
indexed by edges:
$$
A_{fa}:=A(a\to f)\qquad\text{and}\qquad
U_{fe}
:=\begin{cases}
A(ef), &\text{if the endpoint of $e$ is the starting point of $f$},\\
0, &\text{otherwise}.
\end{cases}
$$
Denote by $Z$ be the denominator of~\eqref{eq-def-finite-lattice-propagator}.
Then $Z=\det (I-U)$. If $Z\ne 0$ then $A=(I-U)^{-1}$.
\end{lemma}

\begin{proof}
The first formula follows from
$$
\det (I-U)=\sum_\sigma \mathrm{sgn}(\sigma)\prod_e (I-U)_{\sigma(e)e}
=\sum_\sigma \mathrm{sgn}(\sigma)\prod_{e:\sigma(e)\ne e}(-U_{\sigma(e)e})
=\sum_S A(S)=Z.
$$
Here the products are over all edges $e$, the first two sums are over all permutations $\sigma$ of edges, and the last sum is over all loop configurations $S$.
All the equalities except the third one follow from definitions.

To prove the third equality, take a permutation $\sigma$ of edges and decompose it into disjoint cycles.
Take one of the cycles $e_1e_2\dots e_ke_1$ of length $k>1$.
The contribution of the cycle to the product is nonzero only if
the endpoint of each edge $e_i$ is the starting point of the next one.
In the latter case the contribution is
$$
(-U_{e_2e_1})\dots (-U_{e_1e_k})
= (-1)^k A(e_1e_{2})\dots A(e_ke_{1})
= (-1)^{k-1} A(e_1e_{2}\dots \bluevar{e_k}e_{1}),
$$
where we have taken the minus sign in~\eqref{eq-def-anti3} into account.

Multiply the resulting contributions over all cycles of length greater than $1$.
The cycles form together a loop configuration $S$,
and the product of their arrows is $A(S)$.
Since $(-1)^{k-1}$ is the sign of the cyclic permutation, the product
of such signs equals $\mathrm{sgn}(\sigma)$.
Clearly, the resulting loop configurations $S$ are in
bijection with all permutations giving a nonzero contribution to
the sum. This proves that $\det (I-U)=Z$.

To prove the formula $A=(I-U)^{-1}$, replace the entry $(I-U)_{af}$
of the matrix $I-U$ by $1$, and all the other entries in the row $a$
by $0$. Analogously to the previous argument,
the determinant of the resulting matrix (the cofactor of $I-U$)
equals the numerator of~\eqref{eq-def-finite-lattice-propagator}.
By Cramer's rule we get $A=(I-U)^{-1}$.
\end{proof}

\begin{proof}[Proof of Proposition~\ref{cor-double-fourier-finite}]
This follows from Lemmas~\ref{l-discrete-fourier}--\ref{l-loop-expansion}.
In particular, $Z=\det(I-U)\ne 0$ because $I-U$ is the matrix of system~\eqref{eq-l-equivalence-equations}
having a unique solution by Lemmas~\ref{l-discrete-fourier}--\ref{l-equivalence-equations}.
\end{proof}

\begin{remark} \label{rem-z} \blueit{One can compute $Z$} using Lemma~\ref{l-loop-expansion}, \blueit{the block structure
$U=\left(\begin{smallmatrix}
            0 & U_{\text{even}} \\
            U_{\text{odd}} & 0
          \end{smallmatrix}\right)$ in a suitably ordered basis, Schur's formula, and} the discrete Fourier transform.
\blueit{Namely,} multiplying the determinants of equations~\eqref{eq-transformed-equation}
over all $p,\omega$, one \blueit{gets} 
$$
Z=2^{T^2}\prod_{p,\omega\in
\faktor{(2\pi/T\varepsilon)\mathbb{Z}}{(2\pi/\varepsilon)\mathbb{Z}}}
\left(
\cos(\omega\varepsilon)
  -\frac{\cos(p\varepsilon)+im\varepsilon^2\delta}
  {\sqrt{1-\delta^2\varepsilon^2}\sqrt{1+m^2\varepsilon^2}}
\right).
$$
This remains true for $m=0$ or $\delta=0$, implying that
$Z=0$ for $T$ divisible by $4$ (because of the factor
obtained for $p=\omega=\pi/2\varepsilon$). For $\delta=0$ the latter remains true
even if $m$ has imaginary part, which shows that Step~2${}'$ in \S\ref{ssec-outline} is necessary.
Moreover, by Proposition~\ref{cor-double-fourier-finite},
the limit $\lim_{\delta\searrow 0}A(b_2\to e_1;m,\varepsilon,\delta,T)$,
hence $\lim_{\delta\searrow 0}A(a_0\to f_1;m,\varepsilon,\delta,T)$,
does not exist for $T$ divisible by $4$ and, say, $x=t=0$. Thus
one cannot interchange the order of limits in~\eqref{eq-def-infinite-lattice-propagator}.
\end{remark}

\begin{example}[No charge conservation on the $2\times 2$ lattice] \label{ex-2x2} For $T=2$ we have
\textup{\cite[\S2]{SU-3}}:
$$\sum_{f\text{ starting on }t=0}|{A}(a\to f)|^2
=\frac{(1+\delta^2\varepsilon^2)(1+m^2\varepsilon^2)}{4(m^2\varepsilon^2+\delta^2\varepsilon^2)}
\ne \frac{(1-\delta^2\varepsilon^2)(1+m^2\varepsilon^2)}{4(m^2\varepsilon^2+\delta^2\varepsilon^2)}
=\sum_{f\text{ starting on }t=\varepsilon}|{A}(a\to f)|^2,
$$
where the sums are over all edges $f$ starting 
on the line $t=0$ and $t=\varepsilon$ respectively.
\end{example}

Performing the change of variables
$(p,\omega)\mapsto (\pm p,-\omega)$ in Proposition~\ref{cor-double-fourier-finite},
we get the following.

\begin{corollary}[Skew symmetry] \label{cor-symmetry-finite}
For each $(x,t)\in\varepsilon\mathbb{Z}^2$, where $(x,t)\ne (0,0)$,
we have the identities $A(b_2\to e_1(x,-t))=A(b_2\to e_1(x,t))$ and
$A(b_2\to e_2(-x,-t))=-A(b_2\to e_2(x,t))$.
\end{corollary}

For the proof of the identities from~\S\ref{ssec-combi-def},
we need \blueit{Proposition~\ref{l-invariance} below}, which follows immediately from defining equations~\eqref{eq-def-anti3}--\eqref{eq-def-finite-lattice-propagator}.
\blueit{(Cf.~Example~\ref{ex-1x1}.)}

\begin{definition} The arrow $A(a\to f)$ is \emph{invariant}
under a transformation $\tau$ of the lattice, if $A(\tau(a)\to \tau(f))=A(a\to f)$.
Clearly, $A(a\to f)=A_{fa}(im\varepsilon, \delta\varepsilon, \sqrt{1+m^2\varepsilon^2},\sqrt{1-\delta^2\varepsilon^2})$
for some rational function $A_{fa}$ in $4$ variables, depending on the parameters $a$, $f$, $T$.
A transformation $\tau$ \emph{acts as the replacement $\delta\leftrightarrow im$},
if $A(\tau(a)\to \tau(f))=A_{fa}(\delta\varepsilon, im\varepsilon, \sqrt{1-\delta^2\varepsilon^2}, \sqrt{1+m^2\varepsilon^2})$.
\end{definition}

\begin{proposition}[Invariance \bluevar{and duality principle}] \label{l-invariance}
The arrow $A(a\to f)$ is invariant under the translations by the vectors
$(\varepsilon,0)$ and $(0,\varepsilon)$ and under the reflection with respect
to the line $x=0$. The translation by $(\varepsilon/2,\varepsilon/2)$
acts as the replacement $\delta\leftrightarrow im$.
\end{proposition}

\begin{proof}[Proof of Proposition~\ref{p-initial}]
By \blueit{Proposition}~\ref{l-invariance} we may assume that $a=b_2$ because
the required equation $A(a\to a)=1/2$ is invariant under the replacement $\delta\leftrightarrow im$.
Apply Proposition~\ref{cor-double-fourier-finite}
for $x=t=0$ so that $e_2=b_2$. The change of variables
$(p,\omega)\mapsto (-p,-\omega)$ shows that the sum over
$p,\omega$ in the expression for $A(b_2\to e_2)$
vanishes. The remaining term is $1/2$. 
\end{proof}


\begin{proof}[Proof of Proposition~\ref{p-symmetry}]
By \blueit{Proposition}~\ref{l-invariance} we may assume $a=b_2$.
Assume that $(x,t)$ is even; otherwise the proof is analogous.
Consider the following $2$ cases.

Case 1: $f=e_1(x,t)$. Translate both $a$ and $f$ by $(-x,-t)$
and reflect with respect to the line $x=0$.
By \blueit{Proposition}~\ref{l-invariance} and Corollary~\ref{cor-symmetry-finite}
we get $A(e_1(x,t)\to b_2)=A(b_2\to e_1(x,-t))=A(b_2\to e_1(x,t))$, as required.

Case 2: $f=e_2(x,t)\ne b_2$.
Translating by $(-x,-t)$, applying \blueit{Proposition}~\ref{l-invariance}
and Corollary~\ref{cor-symmetry-finite}, we get
$A(e_2(x,t)\to b_2)=A(b_2\to e_2(-x,-t))=-A(b_2\to e_2(x,t))$, as required.
\end{proof}

\begin{proof}[Proof of Proposition~\ref{p-dirac-finite}]
By Lemma~\ref{l-loop-expansion} we get $(I-U)A=I$, that is, $A_{fa}-\sum_e U_{fe}A_{ea}=\delta_{fa}$, which is equivalent
to the required identity.
\end{proof}

\begin{proof}[Proof of Proposition~\ref{p-dirac-adjoint}]
This follows from Propositions~\ref{p-initial}--\ref{p-dirac-finite}
because $e\parallel f$ and $e'\perp f$.
\end{proof}

\begin{proof}[Proof of Proposition~\ref{cor-dirac}]
By Lemma~\ref{l-loop-expansion} we get
$(I-U^n)A=(I+U+\dots+U^{n-1})(I-U)A=I+U+\dots+U^{n-1}$,
which is equivalent to the required identity for $n\le 2T$.
\end{proof}

\begin{proof}[Proof of Theorem~\ref{p-real-imaginary}]
The denominator of~\eqref{eq-def-finite-lattice-propagator}
does not vanish by Proposition~\ref{cor-double-fourier-finite}.
Limit~\eqref{eq-def-infinite-lattice-propagator}
is computed as follows: \blueit{}
\begin{multline*}
A(a_0\to f_k)
=\sum_{j,l=1}^2 (-1)^l A(b_la_0)A(b_l\to e_j)A(e_jf_k)
=\frac{1}{1-\delta^2\varepsilon^2}\sum_{j,l=1}^2 (\delta\varepsilon)^{2-l}(\bluevar{-}\delta\varepsilon)^{|j-k|}A(b_l\to e_j)
\\
=\bluevar{\frac{1}{1-\delta^2\varepsilon^2}}\sum_{j,l=1}^2 (\delta\varepsilon)^{2-l}(\bluevar{-}\delta\varepsilon)^{|j-k|}
A(b_2\to e_{j'}((-1)^l x,t))
\\
\overset{T\to\infty}{\to}
-\frac{1}{2(1-\delta^2\varepsilon^2)}\sum_{j,l=1}^2 (\delta\varepsilon)^{2-l}(\bluevar{-}\delta\varepsilon)^{|j-k|}i^{j'}
A_{j'}((-1)^l x,t)
\overset{\delta\searrow 0}{\to}
-\frac{1}{2}i^k\widetilde{A}_k(x,t),
\end{multline*}
where $j':=2-|j-l|$.
Here the first two equalities follow from Propositions~\ref{p-dirac-finite}--\ref{p-dirac-adjoint}.
The third one is obtained by a reflection.
The convergence holds by Propositions~\ref{cor-double-fourier-finite},
\ref{l-double-fourier-finite}, and Definition~\ref{def-anti-alg}.
\end{proof}


\subsection{Generalizations to several particles (Propositions~\ref{p-independence-1}--\ref{p-perturbation})}

\label{ssec-proofs-variations}

The results of \S\ref{sec-identical1} are proved easily.

\begin{proof}[Proof of Proposition~\ref{p-independence-1}]
Due to the condition $x_0 \geq 2t$ there are no paths starting at $A$ and ending at $F'$ and no paths starting at $A'$ and ending at $F$.
Therefore
\begin{multline*}
a (AB,A'B'\to EF,E'F')
= \sum\limits_{\substack{s=AB\dots EF\\s'=A'B'\dots E'F'}}
a(s)a(s')
=\\= \sum\limits_{s=AB\dots EF}a(s)\sum\limits_{s=AB\dots E'F'}a(s')=
a_2(x,t,1,1)a_2(x'-x_0,t,1,1).
\end{multline*}
Taking the norm square, we get the required formula.
\end{proof}

\begin{proof}[Proof of Proposition~\ref{p-transfer-matrix}]
The proof is by induction on~$t$. The base $t=1$ is obvious.
The step is obtained from the following identity by summation over
all unordered pairs $E,E'$:
\begin{equation}\label{eq-local-conservation}
\sum_{F,F'}P(AB,A'B'\to EF,E'F')
=\sum_{D,D'}P(AB,A'B'\to DE,D'E'),
\end{equation}
where the sums are over all ordered pairs $(F,F')$ and $(D,D')$
of integer points such that
$\overrightarrow{EF},\overrightarrow{E'F'},\overrightarrow{DE},
\overrightarrow{D'E'}\in\{(1,1),(-1,1)\}$.
To prove~\eqref{eq-local-conservation}, consider the following $2$ cases.

Case 1: $E\ne E'$. Dropping the last moves of the paths $s$ and $s'$
from Definition~\ref{def-identical}, we get
$$a(AB,A'B'\to EF,E'F')=
\sum_{D,D'}a(AB,A'B'\to DE,D'E')\frac{1}{i}a(DEF)\frac{1}{i}a(D'E'F').
$$
Consider the $4\times 4$ matrix with the entries $\frac{1}{i}a(DEF)\frac{1}{i}a(D'E'F')$,
where $(D,D')$ and $(F,F')$ run through all pairs as in~\eqref{eq-local-conservation}.
A direct checking shows that the matrix is unitary
(actually a Kronecker product of two $2\times 2$ unitary matrices),
which implies~\eqref{eq-local-conservation}.

Case 2: $E=E'$. Dropping the last moves of the two paths,
we get for $F\ne F'$
\begin{multline*}
a(AB,A'B'\to EF,EF')=
a(AB,A'B'\to DE,D'E)\left(a(DEF')a(D'EF)-a(DEF)a(D'EF')\right)
=\\=
a(AB,A'B'\to DE,D'E)
\left(\frac{1}{\sqrt{2}}\frac{1}{\sqrt{2}}-
\frac{i}{\sqrt{2}}\frac{i}{\sqrt{2}}\right)
=a(AB,A'B'\to DE,D'E),
\end{multline*}
where the integer points $D$ and $D'$ are now defined by the conditions
$\overrightarrow{DE}=\overrightarrow{EF}$
and $\overrightarrow{D'E}=\overrightarrow{EF'}$.
Since $a(AB,A'B'\to EF,EF')=a(AB,A'B'\to DE,D'E)=0$ for $F=F'$, we get~\eqref{eq-local-conservation}.
\end{proof}

For the results of \S\ref{sec-identical} we need the following lemma
proved analogously to Lemma~\ref{l-loop-expansion}.
%

\begin{lemma}[Loop expansion] \label{l-loop-expansion3}
Let $a_1,\dots,a_n$ be distinct edges.
In the matrix $I-U$, replace the entries
$(I-U)_{a_1f_1},\dots,(I-U)_{a_nf_n}$ by $1$,
and all the other entries in the rows $a_1,\dots,a_n$
by $0$. Then the determinant of the resulting matrix
equals $ZA(a_1,\dots,a_n\to f_1,\dots,f_n)$.
\end{lemma}

\begin{proof}[Proof of Proposition~\ref{p-det}]
By Theorem~\ref{p-real-imaginary} we get $Z\ne 0$. Then
by Lemma~\ref{l-loop-expansion} we get $(I-U)^{-1}=A$
and $\det A=1/\det(I-U)=1/Z$.
Then by the \blueit{Jacobi} relation between complementary minors of two inverse matrices,
the determinant of the matrix from Lemma~\ref{l-loop-expansion3}
equals
$\det \left(A_{f_ja_i}\right)_{i,j=1}^n/\det A=
Z\det \left(A(a_i\to f_j)\right)_{i,j=1}^n$.
It remains to use Lemma~\ref{l-loop-expansion3} and cancel~$Z$.
\end{proof}

\begin{proof}[Proof of Proposition~\ref{cor-pass}]
The proposition follows from
\begin{multline*}
\hspace{-0.7cm}{A}(a\to f \text{ pass } e
)
={A}(a\to f
)
-{A}(a,e\to f,e
)
=
{A}(a\to f
)
-{A}(a\to f
)
{A}(e\to e
)
+{A}(a\to e
)
{A}(e\to f
)
=\\=\frac{1}{2}
{A}(a\to f
)+
{A}(a\to e
){A}(e\to f
)=
{A}(a\to f
){A}(e\to e
)+
{A}(a\to e
){A}(e\to f
).
\end{multline*}
Here the first equality holds because $S\mapsto S\cup\{e\}$ is a bijection between loop configurations $S$ with the source $a$ and the sink $f$  not passing through $e$ and loop configurations with the sources $a,e$ and the sinks $f,e$.
This bijection preserves $A(S
)$ because $A(S\cup\{e\})=A(S)A(e)=A(S)\cdot 1$. The rest follows from Propositions~\ref{p-det} and~\ref{p-initial}.
\end{proof}

For the result of \S\ref{sec-fermi} we need the following lemma.


\begin{lemma} \label{l-half}
For each edge $e$ we have $\sum_{S\ni e}A(S)=\frac{1}{2}\sum_{S}A(S)$,
where the left sum is over loop configurations containing $e$
and the right sum is over all loop configurations.
\end{lemma}

\begin{proof}
This follows from
$
\frac{\sum_{S\ni e}A(S)}{\sum_{S}A(S)}=
1-\frac{\sum_{S\not\ni e}A(S)}{\sum_{S}A(S)}
=1-A(e\to e)=\frac{1}{2}.
$
Here the second equality holds because $S\mapsto S\cup\{e\}$ is a bijection between loop configurations $S$ not passing through $e$ and loop configurations with the source $e$ and the sink $e$.
The third equality is Proposition~\ref{p-initial}.
\end{proof}

\begin{proof}[Proof of Proposition~\ref{p-perturbation}]
For loop configurations $S_\mathrm{e}$ and $S_\mu$
denote $A(S_\mathrm{e}):=A(S_\mathrm{e},m_\mathrm{e},\varepsilon,\delta)$,
$A(S_\mu):=A(S_\mu,m_\mu,\varepsilon,\delta)$,
$Z_\mathrm{e}:=\sum_{S_\mathrm{e}}A(S_\mathrm{e})$, $Z_\mu:=\sum_{S_\mu}A(S_\mu)$.

Up to terms of order $g^2$, the denominator of~\eqref{eq-def-fermi} equals
$$
\hspace{-0.5cm}\sum_{S_\mathrm{e},S_\mu}A(S_\mathrm{e})A(S_\mu)
\left(1+\sum_{e\in S_\mathrm{e},S_\mu}g\right)
=
\sum_{S_\mathrm{e}}A(S_\mathrm{e})\sum_{S_\mu}A(S_\mu)
+g\sum_{e}\sum_{S_\mathrm{e}\ni e}A(S_\mathrm{e})\sum_{S_\mu\ni e}A(S_\mu)
=Z_\mathrm{e}Z_\mu\left(1+\sum_{e}\frac{g}{4}\right),
$$
where the second sum is over all common edges $e$ of $S_\mathrm{e}$ and $S_\mu$,
the fifth and the last sums are over all the edges $e$,
and we applied Lemma~\ref{l-half}.
In particular, the denominator of~\eqref{eq-def-fermi} is nonzero for $g$ sufficiently small in terms of
$m_\mathrm{e},m_\mu,\varepsilon,\delta,T$ because $Z_\mathrm{e}Z_\mu\ne 0$
by Theorem~\ref{p-real-imaginary}.

Up to terms of order $g^2$, the numerator of~\eqref{eq-def-fermi}
equals
\begin{multline*}
Z_\mathrm{e}Z_\mu
\left(A(a_\mathrm{e}\overset{\mathrm{e}}{\to} f_\mathrm{e})
A(a_\mu\overset{\mu}{\to} f_\mu)
+g\sum_{e}A(a_\mathrm{e}\to f_\mathrm{e}
\text{ pass }e)A(a_\mu\to f_\mu \text{ pass }e)
\right)
=\\=Z_\mathrm{e}Z_\mu
A(a_\mathrm{e}\overset{\mathrm{e}}{\to} f_\mathrm{e})
A(a_\mu\overset{\mu}{\to} f_\mu)
+\\
+gZ_\mathrm{e}Z_\mu\sum_{e}
\left(A(a_\mathrm{e}\overset{\mathrm{e}}{\to} e)
A(e\overset{\mathrm{e}}{\to} f_\mathrm{e})
+\frac{1}{2}A(a_\mathrm{e}\overset{\mathrm{e}}{\to} f_\mathrm{e})\right)
\left(A(a_\mu\overset{\mu}{\to} e)A(e\overset{\mu}{\to} f_\mu)
+\frac{1}{2}A(a_\mu\overset{\mu}{\to} f_\mu)\right),
\end{multline*}
where the sums are over all the edges $e$,
and we applied Proposition~\ref{cor-pass}.

Dividing the resulting expressions and applying Proposition~\ref{p-initial}, we get the result.
\end{proof}


\appendix

\addcontentsline{toc}{myshrink}{}

\section{Alternative definitions and proofs}
\label{ssec-proofs-alternative}


Here we give alternative combinatorial proofs
of Propositions~\ref{p-initial}--\ref{cor-dirac}
and alternative definitions of Feynman anticheckers \blueit{in the spirit of the six-vertex model and exclusion process.}
The proofs are elementary and rely only on the assertion that
the finite-lattice propagator is well-defined (see Theorem~\ref{p-real-imaginary}).
The proofs of Propositions~\ref{p-dirac-finite} and \ref{cor-dirac} are especially simple and are presented first.
\blueit{Cf.~\cite[Proof of Lemma~3.4]{KSS-23} and~\cite{Dmitriev-23}.}
The proofs of Propositions~\ref{p-initial}--\ref{p-symmetry}
require some auxiliary definitions and assertions. Proposition~\ref{p-dirac-adjoint} follows from Propositions~\ref{p-initial}--\ref{p-dirac-finite}. 

\begin{proof}[Second proof of Proposition~\ref{p-dirac-finite}]
Case 1: $a\ne f$. Define $A(a\to f\text{ pass }e)$ (respectively, $A(a\to f\text{ bypass }e)$)
analogously to ${A}(a\to f)$,
only the sum in the numerator of~\eqref{eq-def-finite-lattice-propagator}
is now over loop configurations with the source $a$ and the sink $f$
containing the edge $e$ (respectively, not containing $e$).
In this case the proposition follows from the two identities:
\begin{align}\label{eq-proof-dirac-finite-1}
  A(a\to e\text{ bypass }f)A(ef)+A(a\to e'\text{ bypass }f)A(e'f)
  &=A(a\to f),\\
  \label{eq-proof-dirac-finite-2}
  A(a\to e\text{ pass }f)A(ef)+A(a\to e'\text{ pass }f)A(e'f)&= 0.
\end{align}

Let us prove~\eqref{eq-proof-dirac-finite-1}.
For each loop configuration $S$ with the source $a$ and the sink $f$,
remove the last edge of the path from $a$ to $f$. Since $a\ne f$, we get
a loop configuration $S'$ not containing $f$, with the source $a$ and the sink either $e$ or $e'$. We have decreased the number of nodes in $S$ by $1$, hence
either $A(S)=A(S')A(ef)$ or $A(S)=A(S')A(e'f)$ depending on if the sink is
$e$ or $e'$. Summing over all $S$ we get~\eqref{eq-proof-dirac-finite-1} because
the map $S\mapsto S'$ is clearly invertible.

Let us prove~\eqref{eq-proof-dirac-finite-2}. To each loop configuration $S$
with the source $a$ and the sink $e$ containing $f$, assign a loop configuration $S'$
with the source $a$ and the sink $e'$ containing $f$ as follows.
If $f$ in contained in a loop $f\dots e'f$ from $S$ then combine the loop
with the path $a\dots e$ from $S$ into the new path
$a\dots ef\dots e'$. If $f$ in contained in the path $a\dots e'f\dots e$ from $S$,
then decompose the path into a new path $a\dots e'$ and a new loop $f\dots ef$.
All the other loops in $S$ remain unmodified. The resulting loop configuration is $S'$.
The map $S\mapsto S'$ changes the parity of the number of loops and preserves
the nodes except \blueit{for} that the node $(e',f)$ is replaced by $(e,f)$.
Hence $A(S')A(e'f)=-A(S)A(ef)$. Summing over all $S$ we get~\eqref{eq-proof-dirac-finite-2}
because the map $S\mapsto S'$ is invertible.

Case 2: $a=f$. To each loop configuration $S$ with the source $f$
and the sink $f,e,e'$ respectively, assign a loop configuration $S'$
(without sources and sinks) as follows.
If $S$ has the sink $f$, then remove the path $f$ from $S$.
If $S$ has the sink either $e$ or $e'$, then close up
the path in $S$ into a new loop.
The resulting loop configuration is $S'$.
In the former case, $S'$ has the same loops and nodes as $S$,
and in the latter case, we have added one loop and one node.
Thus, if the sink is $f$, or $e$, or $e'$,
then $A(S')=A(S)$, or $-A(S)A(ef)$, or
$-A(S)A(e'f)$ respectively.
Summing over all $S$ and dividing by the denominator
of~\eqref{eq-def-finite-lattice-propagator}, we get
$$
A(f\to f)
-A(f\to e)A(ef)-A(f\to e')A(e'f)=1,
$$
because the map $S\mapsto S'$ is invertible.
\end{proof}

\begin{proof}[Second proof of Proposition~\ref{cor-dirac}]
The proof is by induction on $n$. The base $n=1$ is the trivial assertion $A(f)=1$.
To perform the induction step, take a path $e\dots f$ of length $n-1$.
Let $d$ and $d'$ be the two edges with the endpoint at the starting point
of $e$.
By Proposition~\ref{p-dirac-finite} we get
\begin{multline*}
A(a\to e)A(e\dots f)=
\left(A(a\to d)A(de)+A(a\to d')A(d'e)+\delta_{ae}\right)A(e\dots f)
\\=A(a\to d)A(de\dots f)+A(a\to d')A(d'e\dots f)+\delta_{ae}A(e\dots f).
\end{multline*}
Here $d$ and $d'$ are distinct from all the edges in $e\dots f$
by the assumption $n\le 2T$, and hence can be added to the path.
Summing over all such paths $e\dots f$ we get
$$
\sum_{e\dots f\text{ of length }n-1}A(a\to e)A(e\dots f)=
\sum_{d\dots f\text{ of length }n}A(a\to d)A(d\dots f)
+\sum_{a\dots f\text{ of length }n-1}A(a\dots f).
$$
Adding $-A(a\to f)+\sum_{a\dots f\text{ of length }<n-1}A(a\dots f)$ to both
sides and applying the inductive hypothesis, we get the required identity.
\end{proof}


The next definitions and a lemma are
needed for combinatorial proofs of Propositions~\ref{p-initial}--\ref{p-symmetry}.

\begin{definition} \label{def-flip} (See Figure~\ref{fig-1x1} to the right) Let $(e,f)$ be a pair of edges such that the endpoint of $e$ is the starting point of $f$. The \emph{complementary pair} $(e',f')$ is formed by the other edge $e'\ne e$ with the same endpoint as $e$ and the other edge $f'\ne f$  with the same starting point as~$f$.

Let $S$ be a loop configuration containing both nodes $(e,f)$ and $(e',f')$. The \emph{flip} of $S$ (at the endpoint of $e$) is the loop configuration obtained as follows. If the nodes $(e,f)$ and $(e',f')$ belong to distinct loops $ef\dots e$ and $e'f'\dots e'$ of $S$, then combine them into one new loop $ef'\dots e'f\dots e$. If the nodes $(e,f)$ and $(e',f')$ belong to the same loop $ef\dots e'f'\dots e$, then decompose the latter into two new loops $e'f\dots e'$ and $ef'\dots e$. All the other loops in $S$ remain unmodified. The \emph{flip} of a loop configuration with sources and sinks is defined analogously.
\end{definition}



\begin{lemma} \label{l-complete-loop}
If a loop configuration $S$ (without sources and sinks) contains all the edges, then the number of loops in $S$ has the same parity as one half of the total number of turns in $S$.
\end{lemma}

\begin{proof}
The proof is by induction over the total number of turns.

Base: If $s$ has no turns, then each loop entirely consists of the edges of the same direction. The reflection with respect to the vertical line $x=T\varepsilon/2$ shows that there is equal number of loops consisting of upwards-left and upwards-right edges. Hence the number of loops is even.

Step: Assume that $s$ has a loop with a turn $(e,f)$.
Since $s$ contains all the edges, it has also a loop with the complementary turn $(e',f')$.
Then a flip of $s$
changes the parity of the number of loops and reduces the total number of turns by $2$. By induction, the lemma follows.
\end{proof}

\blueit{Now we give a definition of Feynman anticheckers as a six-vertex model with complex weights (see Figure~\ref{fig-6v}). Instead of assigning a direction to each edge, as in the common six-vertex model, we
speak of choosing a set of edges, and the edges 
are always oriented in the time direction. To each configuration, i.e., set of edges, we also assign an overall sign defined globally in terms of a loop decomposition of the set. Our loop decomposition is very different from the one in~\cite{Duminil-Copin-etal-22}.}

\begin{definition} \label{def-anti-subgraphs}
A set $s$ of edges is a \emph{current}, if for each lattice point the number of edges in $s$ starting at the point equals the number of edges in $s$ ending at the point \blueit{(recall that each edge is oriented in the direction where time $t$ increases).} A set $s$ of edges is a \emph{current with sources $a_1,\dots,a_n$ and sinks $f_1,\dots,f_n$}, if \blueit{edges $a_1,\dots,a_n$ are distinct, egdes $f_1,\dots,f_n$ are distinct,} $s$ contains $a_1,\dots,a_n,f_1,\dots,f_n$, and for each lattice point the number of edges in $s$ starting at the point and distinct from $a_1,\dots,a_n$ equals the number of edges in $s$ ending at the point and distinct from $f_1,\dots,f_n$.

\remove{}

A lattice point is a \emph{singularity} of $s$, if it is the starting point of two edges of $s$, distinct from the sources. Clearly, for each $s$ there exists a unique loop configuration (called the \emph{loop decomposition} of $s$) having the same sources and sinks, consisting of the same edges, and having no turns at the singularities of $s$. 
If the loop decomposition
has exactly $l$ loops and $n$ paths joining $a_1,\dots,a_n$ with $f_{\sigma(1)},\dots,f_{\sigma(n)}$ respectively for some permutation $\sigma$, then denote $\mathrm{sgn}(s):=(-1)^l\mathrm{sgn}(\sigma)$. Here we set $\mathrm{sgn}(\sigma)=+1$ for $n=0$.

A \emph{node} of $s$ is an ordered pair $(e,f)$ of edges of $s$ such that the endpoint of $e$ is the starting point of $f$ and is not a singularity of $s$, the edge $e$ is not a sink, and $f$ is not a source. The numbers $\mathrm{eventurns}(s)$, $\mathrm{oddturns}(s)$, $\mathrm{evennodes}(s)$, $\mathrm{oddnodes}(s)$, and $A(s
)$ are defined literally as for a path or loop (see Definition~\ref{def-anti-combi}), with the overall sign in~\eqref{eq-def-anti3} set to be $\mathrm{sgn}(s)$.
Denote
$$
{A}^{\bluevar{\mathrm{6V}}}(a_1,\dots,a_n\to f_1,\dots,f_n
)
:=\frac{\sum\limits_{\substack{\text{currents $s$}\\
\text{with the sources $a_1,\dots,a_n$ and the sinks $f_1,\dots,f_n$}}}A(s
)}
{\sum\limits_{\substack{\text{currents $s$}}}A(s
)}.
$$

If all $a_1,\dots,a_n,f_1,\dots,f_n$ are distinct, then the \emph{complement} to $s$ is the current $\bar s$ with sources $f_1,\dots,f_n$ and sinks $a_1,\dots,a_n$
formed by $f_1,\dots,f_n, a_1,\dots,a_n$ and exactly those other edges that do not belong to $s$.
\end{definition}

\begin{example}[Empty and complete currents] \label{ex-complete}
We have $A(\emptyset
)=A(\overline{\emptyset}
)=1$,
where $\overline{\emptyset}$ is the current consisting of all the edges. Indeed, the currents $\emptyset$ and  $\overline{\emptyset}$ have no nodes, and
$\mathrm{sgn} (\overline{\emptyset})=\mathrm{sgn}(\emptyset)=+1$ by Lemma~\ref{l-complete-loop}
because the loop decomposition of $\overline{\emptyset}$ has no turns.
\end{example}

\begin{proposition}[Equivalence of definitions] \label{p-equivalence} For \blueit{any}
edges
$a_1,\dots,a_n,f_1,\dots,f_n$ we have
$$
{A}^{\bluevar{\mathrm{6V}}}(a_1,\dots,a_n\to f_1,\dots,f_n
)
=
{A}(a_1,\dots,a_n\to f_1,\dots,f_n
).
$$
\end{proposition}


\begin{proof}[Proof of Proposition~\ref{p-equivalence}]
To each loop configuration $S$ (possibly with sources and sinks), assign the set of all edges contained in the loops and paths of $S$. Clearly, we get a current with the same sources and sinks. A current $s$ with $K$ even and $J$ odd singularities has $2^{K+J}$ preimages, obtained from the loop decomposition $S$ of $s$ by flips at any subset of the set of singularities.

It suffices to prove that $A(s)=\sum_{S'}A(S')$, where the sum is over all $2^{K+J}$ preimages $S'$ of $s$. Take a loop configuration $S'$ obtained from $S$ by flips at $k$ even and $j$ odd singularities. Since each such flip
increases the number of turns by $2$ and
changes either the parity of the number of loops
or the sign of the permutation $\sigma$ from Definition~\ref{def-multipoint},
it follows that $A(S')=(-1)^{k+j}(-im\varepsilon)^{2j}(-\delta\varepsilon)^{2k}A(S)=
(m^2\varepsilon^2)^{j}(-\delta^2\varepsilon^2)^{k}A(S)
$.
Summing over all the subsets of the set of singularities,
we get the required equality
$$
\sum_{S'}A(S')=\sum_{k=0}^{K}\sum_{j=0}^{J}
\binom{K}{k}\binom{J}{j}
(m^2\varepsilon^2)^{j}(-\delta^2\varepsilon^2)^{k}A(S)
= (1+m^2\varepsilon^2)^{J}(1-\delta^2\varepsilon^2)^{K}A(S)
= A(s),
$$
where the factor before $A(S)$ in the latter equality compensates the contribution
of the $2K+2J$ nodes of the loop decomposition $S$ which are not nodes of the current $s$.
\end{proof}

The following proposition demonstrates a symmetry between particles and antiparticles.

\begin{proposition}[Complement formula] \label{p-complement}
For each 
current $s$, possibly with sources $a_1,\dots,a_n$ and sinks $f_1,\dots,f_n$, where all $a_1,\dots,a_n,f_1,\dots,f_n$ are distinct, we have
$
{A}(s
)
=
(-1)^{|\{k:a_k\parallel f_k\}|} {A}(\bar s
).
$
\end{proposition}

\begin{example}
In Example~\ref{ex-1x1}, if the set $s=\{a,c\}$ is viewed as a current without sources and sinks, then it has the complement $\bar{s}=\{b,d\}$, so that ${A}(s
)
=-1/\sqrt{1+m^2\varepsilon^2}\sqrt{1-\delta^2\varepsilon^2}
={A}(\bar s
)$.
If the same set $s=\{a,c\}$ is viewed as a current with
the source $a$ and the sink $c$, then the complement $\bar{s}=\{a,b,c,d\}$ has the source $c$ and the sink $a$, so that
${A}(s
)
=1/\sqrt{1+m^2\varepsilon^2}=
-{A}(\bar s
)$.
\end{example}

\begin{proof}[Proof of Proposition~\ref{p-complement}]
First let us show that $A(s)=A(\bar s)$ up to sign, namely,
$A(s)\mathrm{sgn}(s)=A(\bar s)\mathrm{sgn}(\bar s)$.
To each node $(e,f)$ of $s$, assign the complementary pair $(e',f')$. The latter is a node of $\bar s$. Indeed, since
the starting point of $f$ (equal to the endpoint of $e$) is not a singularity, it follows that either $f'\notin s$ or $f'$ is a source of $s$. Thus $f'\in \bar s$ and it is not a source of $\bar s$. Analogously, $e'\in \bar s$ and it is not a sink of $\bar s$. Since $f$ is not a source, it follows that either $f\notin \bar s$ or $f$ is a source of $\bar s$. This means that the starting point of $f'$ is not a singularity of $\bar s$. Thus $(e,f)\mapsto(e',f')$ is a bijection between the sets of nodes of $s$ and $\bar s$. This bijection preserves the parity of nodes and takes turns to turns. Thus $A(s)\mathrm{sgn}(s)=A(\bar s)\mathrm{sgn}(\bar s)$.

Second let us show that
$\mathrm{sgn}(s)\mathrm{sgn}(\bar s)=(-1)^{n-\mathrm{turns}(s)}$, where $\mathrm{turns}(s)$ is the total number of turns in the current $s$.
Let $S$ and $\overline{S}$ be the loop decompositions of $s$ and $\bar s$. Let $S$ have exactly $l$ loops and $n$ paths $a_1\dots f_{\sigma(1)}$, \dots,
$a_n\dots f_{\sigma(n)}$ for some permutation $\sigma$. Let $\bar S$ have exactly $\bar l$ loops and $n$ paths $f_1\dots a_{\bar\sigma(1)}$, \dots,
$f_n\dots a_{\bar\sigma(n)}$ for some permutation $\bar\sigma$. Form the loop
$$
a_1\dots f_{\sigma(1)}\dots a_{\bar\sigma\circ\sigma(1)}\dots
f_{\sigma\circ\bar\sigma\circ\sigma(1)}\dots a_1,
$$
starting from $a_1$ and alternating the paths of $S$ and $\bar S$ until the first return to $a_1$. Form analogous loops starting from the other not yet visited edges $a_k$. The resulting loops are in bijection with the loops in the loop decomposition of the permutation $\bar\sigma\circ\sigma$. Hence their total number is even if and only if $(-1)^n\mathrm{sgn}(\sigma)\mathrm{sgn}(\bar\sigma)=+1$.
Consider the set consisting of the resulting loops
(obtained by gluing the paths of $S$ and $\overline{S}$)
and the loops of $S$ and $\overline{S}$. The number of loops in the set is even if and only if $(-1)^n\mathrm{sgn}(s)\mathrm{sgn}(\bar s)=+1$ because $\mathrm{sgn}(s)=(-1)^l\mathrm{sgn}(\sigma)$.
On the other hand, by Lemma~\ref{l-complete-loop} this number
has the same parity as $\mathrm{turns}(s)$, because the total number of turns in $S$ and the total number of turns in $\overline{S}$ both equal $\mathrm{turns}(s)$. We get $\mathrm{sgn}(s)\mathrm{sgn}(\bar s)=(-1)^{n-\mathrm{turns}(s)}$.

It remains to notice that $n-\mathrm{turns}(s)=|\{k:a_k\parallel f_k\}|\mod 2$. Indeed, each loop in $S$ 
has an even number of turns, a path $a_k\dots f_{\sigma(k)}$ has an even number of turns if and only if $a_k\parallel f_{\sigma(k)}$, and the parity of $|\{k:a_k\parallel f_k\}|$ is invariant under a permutation of $f_1,\dots,f_n$.
\end{proof}

\begin{proof}[Second proof of Proposition~\ref{p-initial}]
Use Definition~\ref{def-anti-subgraphs} and~Proposition~\ref{p-equivalence}.
The result follows from
$$
\sum_{s\text{ with the source and sink $a$}}A(s
)
=\sum_{s\not\ni a}A(s
)
=\sum_{s\ni a}A(s
)
=\frac{1}{2}\left(\sum_{s\not\ni a}A(s
)+\sum_{s\ni a}A(s
)\right)
=\frac{1}{2}\sum_{s}A(s
).
$$
Here the sums are over currents $s$ (in the first sum --- with the source and the sink~$a$). The first equality holds because $s\mapsto s-\{a\}$ is a bijection between currents with the source and sink $a$ and currents (without sources and sinks) not containing $a$. This bijection preserves $A(s
)$ because $a$ does not belong to any node of $s$. The second equality holds because $s\mapsto\overline{s}$ is a bijection between the currents containing and not containing $a$. This bijection preserves $A(s
)$ by
Proposition~\ref{p-complement}. The third equality follows from the second one.
\end{proof}

\begin{proof}[Second proof of Proposition~\ref{p-symmetry}]
The map $s\mapsto\overline{s}$ is a bijection between the currents with the source $a$ and sink $f$ and the currents with the source $f$ and sink $a$. By
Proposition~\ref{p-complement}, this bijection preserves $A(s
)$ for $a\perp f$ and changes the sign of $A(s
)$ for $a\parallel f$. Summing over all $s$ and diving by the sum over all the currents
, we get the required assertion by Proposition~\ref{p-equivalence}.
\end{proof}



We conclude this section by restating Definition~\ref{def-anti-subgraphs}
informally in a self-contained way resembling \emph{exclusion process}.
(Cf. a different \emph{quantum exclusion process} \cite{Bauer-Bernard-Jin-17} defined
by a continuous-time stochastic differential equation.)


\begin{wrapfigure}[15]{r}{3.0cm}
\vspace{-0.9cm}
\includegraphics[width=2.5cm]{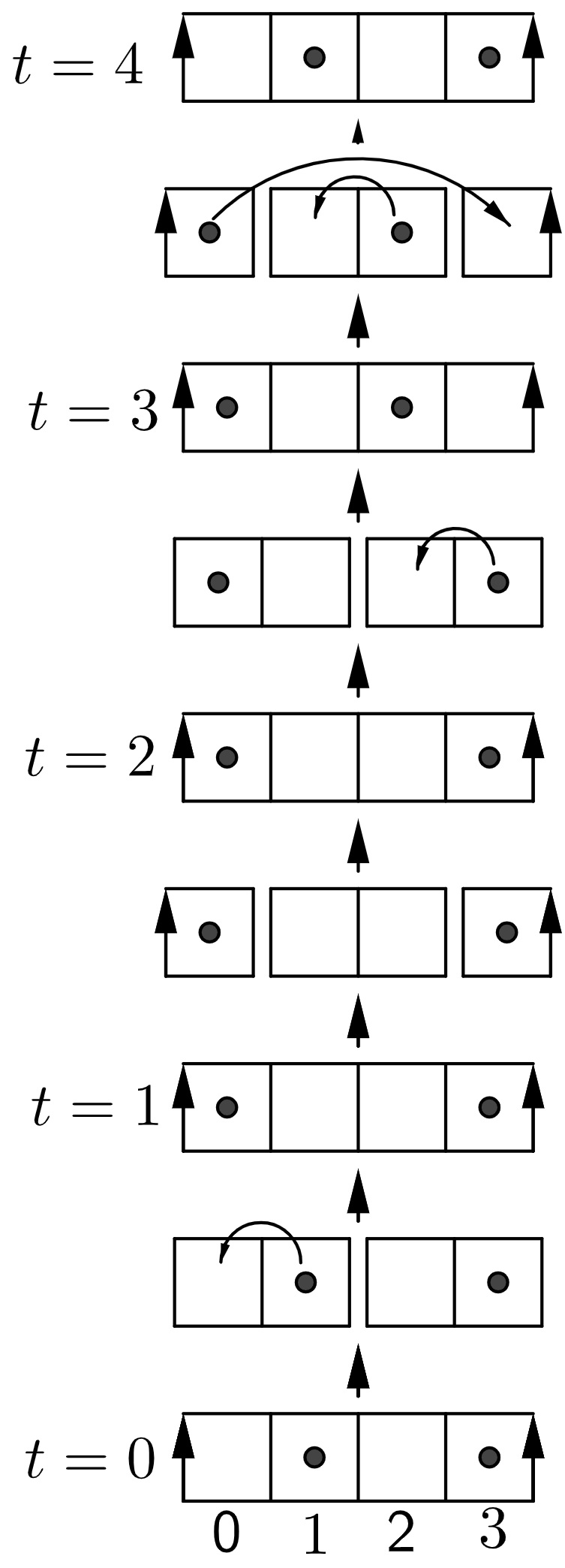}
\caption{A realization}
\label{fig-process}
\end{wrapfigure}
\smallskip
\noindent\textbf{Definition Sketch.}
(See Figure~\ref{fig-process}) Fix $T\in\mathbb{Z},\mu,\bluevar{\nu}>0$ called \emph{half-period, particle
mass, small imaginary mass} respectively. Take a checkered stripe $1\times 2T$ closed in a ring. Enumerate the $2T$ squares
by the numbers $0,\dots, 2T-1$ consecutively.

Define a realization of the exclusion process inductively.
At time $t=0$ some squares are occupied by identical particles,
at most one per square.

At time $t=k$ decompose the stripe into $T$ rectangles $1\times 2$
so that squares $k$ and $k+1$ form a rectangle. In each rectangle
with exactly $1$ particle, the particle is allowed to jump
into the empty square of the same rectangle.
In rectangles with $2$ or $0$ particles, nothing is changed.

Finally at time $t=2T$ it is requested that the particles occupy the same
set of squares as at $t=0$. The resulting sequence of $2T$
configurations of particles at times $t=0,1,\dots,2T-1$ is
a \emph{realization} of the exclusion process.

A \emph{realization 
with a source at $(0,0)$
and a sink at $(x,t)\in\bluevar{\left(\faktor{\mathbb{Z}}{2T\mathbb{Z}}\right)^2}$} is defined analogously, only:
\begin{itemize}
\item at time $0$ before any jumps the square $0$ is empty, and
a particle is added to the square;
\item at time $t$ after all jumps the square $x$ is occupied, and
the particle is removed from it.
\end{itemize}

To each realization $s$ 
(possibly with a source and a sink), assign a complex number
$A(s)$ as follows. Start with $A(s)=(-1)^n$, where $n$ is the number
of particles at time $0$. \remove{}
For each moment $t=0,1,\dots 2T-1$ and each $1\times 2$ rectangle containing two particles
at 
\blueit{time} $t$, multiply the current value
of $A(s)$ by $-1$. For each moment $t=0,1,\dots 2T-1$ and each $1\times 2$ rectangle
containing exactly one particle at 
\blueit{time} $t$, multiply the current value
of $A(s)$ by
$$
\begin{cases}
1/\sqrt{1+\mu^2}, &\text{if the particle jumps and $t$ is even};\\
-i\mu/\sqrt{1+\mu^2}, &\text{if the particle does not jump and $t$ is even};\\
1/\sqrt{1-\bluevar{\nu}^2}, &\text{if the particle jumps and $t$ is odd};\\
-\bluevar{\nu}/\sqrt{1-\bluevar{\nu}^2}, &\text{if the particle does not jump and $t$ odd}.
\end{cases}
$$
\blueit{Finally, multiply $A(s)$ by the sign of the permutation of the particles obtained at time $t=2T$ (if the particles added at the source and removed at the sink are distinct, then they are identified, and $A(s)$ is multiplied by an additional $-1$).}

The \emph{two-point function} 
is then
$
{A}^{\bluevar{\curvearrowright}}(0,0\to x,t
)
:=\frac{\sum\limits_{\substack{\text{realizations $s$}\\
\text{with the source $(0,0)$ and the sink $(x,t)$}}}A(s
)}
{\sum\limits_{\substack{\text{realizations $s$}}}A(s
)}.
$
\smallskip

Using Proposition~\ref{p-equivalence}, one can see that
the two-point function actually equals the finite-lattice propagator;
for instance,  for $x,t$ even it equals
$A(a_0\to f_{2}; \mu/2,2,\bluevar{\nu}/2,T)$.
Notice that if we restrict to just realizations without particles at time $0$,
drop space- and time-periodicity requirements,
and take $\bluevar{\nu}=0$, then the definition becomes equivalent to Definition~\ref{def-mass}.

\addcontentsline{toc}{myshrink}{}

\section{Wightman axioms}
\label{sec-axioms}


To put the new model in the general framework of quantum theory, we define the Hilbert space describing the states of the model along with the Hamiltonian and the field operators acting on this space. The definition is similar to (and simpler than) the continuum free spin $1/2$ field \cite[\S5.2]{Folland}, 
only we have unusual dispersion relation~\eqref{eq-omega} and smaller number of spin components (coming from smaller spacetime dimension).
It goes along the lines of \cite[\S IV]{Bender-etal-85}, where the \emph{massless} case is considered.
For the \emph{massive} case, we have not found the detailed construction in the literature (cf.~\cite[\S2.3]{Arrighi-et-al}).
Although the definition is self-contained, familiarity with the continuum analogue is desirable. 
We use notation $f^*$, $f^\dagger$, $\bar f$, $\langle f| g\rangle$ from~\cite[\S1.1]{Folland}
(introduced below) 
unusual in mathematics but common in physics.
For simplicity, we first perform the construction for the model with a fixed spatial size, then for the infinite lattice,
and finally discuss which Wightman axioms of quantum field theory are satisfied.

\subsection*{Informal motivation}

In quantum theory, a system is described by a Hilbert space encoding all possible \emph{states} of the system. Examples of states of a free field (in a box of fixed spatial size) are: the vacuum state without any particles at all; the state with one particle of given momentum $p$; the state with two particles of momenta $p_1$ and $p_2$; the state with one particle of momentum $p$ and one anti-particle of momentum $q$;
and so on. States of this kind actually form a basis of the Hilbert space. In general, a state is an arbitrary unit vector of the Hilbert space up to scalar multiples.

What quantum theory can compute is the expectation of observables such as the total energy of the system (the \emph{Hamiltonian}) or charge density at a particular point. In general, an observable is a self-adjoint operator on the Hilbert space.  The expectation of the observable in a given state equals the inner product of the state with its image under the operator.

Field operators are not observables but are building blocks for those. They are used to construct states such as the state with one right electron at position $x$ and time $t$, and useful functions such as the propagator.

\subsection*{Definition for fixed spatial size}

\begin{definition} \label{def-abstract-box}
Fix $X\in\mathbb{Z}$ called \emph{lattice spatial size} and $\varepsilon,m>0$.
Assume $X>0$. Define $\widetilde{A}_k(x,t,m,\varepsilon,X)$ analogously
to $\widetilde{A}_k(x,t,m,\varepsilon)$
(see Definition~\ref{def-anti-alg}), only take the quotient
$$
\faktor{
\{\,(x,t)\in[0,X\varepsilon]\times\mathbb{R}:
2x/\varepsilon,2t/\varepsilon,(x+t)/\varepsilon\in\mathbb{Z}\,\}
}
{
\forall t\,(0,t)\sim (X\varepsilon,t).
}
$$
The \emph{momentum space} is
$$
\mathcal{P}_X:=\left\{\frac{2\pi k}{X\varepsilon}:-\frac{X}{2}<k\le \frac{X}{2}, k\in\mathbb{Z}\right\}.
$$

Denote by $L_2(\mathcal{P}_X)$ the Hilbert space with the finite orthonormal basis formed by the functions $\chi_p\colon \mathcal{P}_X\to\mathbb{C}$ equal to $1$ at a particular element $p\in \mathcal{P}_X$ and vanishing at all the other elements.
Equip it with the natural inner product antilinear in the first argument.
Let $\otimes$ and $\wedge$ be respectively the tensor and exterior product over $\mathbb{C}$. An empty exterior product of vectors (respectively, spaces) is set to be $1$ (respectively, $\mathbb{C}$).

The \emph{Hilbert space of $X$-periodic Feynman anticheckers}
is the $2^{2X}$-dimensional Hilbert space
$$
\mathcal{H}_X:=\left(
\bigoplus_{n=0}^{X}\bigwedge^n
L_2\left(\mathcal{P}_X
\right)
\right)^{\otimes 2}.
$$
It has an orthonormal basis formed by the vectors
\begin{equation}\label{eq-basis}
\sqrt{n!l!}(\chi_{p_1}\wedge\dots\wedge \chi_{p_n})\otimes
(\chi_{q_1}\wedge\dots\wedge \chi_{q_l})
\end{equation}
for all integers $n,l$ from $0$ to $X$
and all $p_1,\dots,p_n,q_1,\dots,q_l
\in \mathcal{P}_X$
such that $p_1<\dots<p_n$ and $q_1<\dots<q_l$.
(Physically, the vectors mean the states
with $n$ electrons of momenta $p_1,\dots,p_n$ and $l$ positrons
of momenta $q_1,\dots,q_l$.)
The basis vector $1\in \mathbb{C}=\bigwedge^0 L_2\left(
\mathcal{P}_X
\right)$ obtained
for $n=l=0$ is 
the \emph{vacuum vector}. It is denoted by $|0\rangle$. The dual vector is denoted by $\langle 0|\in \mathcal{H}_X^*$.

The \emph{Hamiltonian of $X$-periodic Feynman anticheckers}
is the linear operator on $\mathcal{H}_X$ 
such that all basis vectors~\eqref{eq-basis} are eigenvectors with the eigenvalues (see notation~\eqref{eq-omega})
$$
\omega_{p_1}+\dots+\omega_{p_n}+\omega_{q_1}+\dots+\omega_{q_l}.
$$
(Physically, $\omega_{p}$ is viewed as the energy of a particle with
momentum $p$; hence the Hamiltonian eigenvalues mean total energy of the eigenstates.)

For $p\in \mathcal{P}_X$,
the \emph{creation operators} $a_p^\dagger $ and $b_q^\dagger$ \emph{of particles and antiparticles with momenta $p$ and $q$} respectively are the linear operators on $\mathcal{H}_X$ defined on basis vectors~\eqref{eq-basis} by
\begin{align*}
\hspace{-0.5cm}  a_p^\dagger \left(\sqrt{n!l!}(\chi_{p_1}\wedge\dots\wedge \chi_{p_n})\otimes
(\chi_{q_1}\wedge\dots\wedge \chi_{q_l})\right)&:= \sqrt{(n+1)!l!}(\chi_p\wedge\chi_{p_1}\wedge\dots\wedge \chi_{p_n})\otimes
(\chi_{q_1}\wedge\dots\wedge \chi_{q_l}),\\
\hspace{-0.5cm}  b_q^\dagger \left(\sqrt{n!l!}(\chi_{p_1}\wedge\dots\wedge \chi_{p_n})\otimes
(\chi_{q_1}\wedge\dots\wedge \chi_{q_l})\right)&:= (-1)^n\sqrt{n!(l+1)!}(\chi_{p_1}\wedge\dots\wedge \chi_{p_n})\otimes
(\chi_q\wedge\chi_{q_1}\wedge\dots\wedge \chi_{q_l}).
\end{align*}
Their adjoint operators are denoted by $a_p$ and $b_q$ respectively.
For each $t\in \varepsilon\mathbb{Z}$ and $x\in\faktor{\varepsilon\mathbb{Z}}{X\varepsilon\mathbb{Z}}$
define the \emph{field operator} $\psi_X(x,t)\colon\mathcal{H}_X\to \mathcal{H}_X\oplus\mathcal{H}_X$ by
$$
\psi_X(x,t):=\frac{1}{\sqrt{X}}
\sum_{p\in \mathcal{P}_X
}
\left(
\begin{pmatrix}
   -i\cos(\alpha_p/2) \\
  \sin(\alpha_p/2)
\end{pmatrix}
e^{ipx-i\omega_pt}a_p
+
\begin{pmatrix}
  i\cos(\alpha_p/2) \\
   \sin(\alpha_p/2)
\end{pmatrix}
e^{i\omega_pt-ipx}b_p^\dagger
\right),
$$
where $\alpha_p\in [0,\pi]$ is determined by the condition
$
\cot\alpha_p=\frac{\sin(p\varepsilon)}{m\varepsilon}.
$
(Informally, the first component of the field operator creates a right positron or annihilates a right electron at position $x$ and time $t$. The second component does the same for left particles.)

The propagator is defined through field operators as follows.
Denote by $\psi^\dagger_X(x,t)\colon \mathcal{H}_X\oplus\mathcal{H}_X\to \mathcal{H}_X$
the adjoint of the operator $\psi_X(x,t)$.
Define the \emph{Dirac adjoint} by
$\bar \psi_X(x,t):=\psi_X^\dagger(x,t)
\left(
\begin{smallmatrix}
0 & -i \\
i & 0
\end{smallmatrix}\right)$. \blueit{Let $\psi_X(0,0)^T$ be the transpose of $\psi_X(0,0)$.}
Define the \emph{time-ordered product}
$$
\mathrm{T}\psi_X(x,t)\bar \psi_X(0,0):=
\begin{cases}
\psi_X(x,t)\bar \psi_X(0,0), &\text{if }t\ge 0;\\
-\bar \psi_X(0,0)^T\psi_X(x,t)^T, &\text{if }t<0.
\end{cases}
$$
The \emph{Feynman propagator for $X$-periodic Feynman anticheckers} is
$\langle 0|\mathrm{T}\psi_X(x,t)\bar \psi_X(0,0)|0\rangle$.
\end{definition}

A direct checking using an analogue of Proposition~\ref{th-equivalence}
shows that the two constructions of the propagator are consistent:
for each $x,t\in\varepsilon\mathbb{Z}$ and positive $X\in \mathbb{Z}$
the propagator for $X$-periodic Feynman anticheckers equals
(cf.~Proposition~\ref{p-abstract} below)
$$
-\frac{i}{2} \begin{pmatrix}
\widetilde{A}_1(-x,t,m,\varepsilon,X)& \widetilde{A}_2(x,t,m,\varepsilon,X) \\
-\widetilde{A}_2(-x,t,m,\varepsilon,X) & \widetilde{A}_1(x,t,m,\varepsilon,X)
\end{pmatrix}.
$$


\begin{remark} In $1$ space dimension, unlike $3$ dimensions, there are \emph{no} states
such as ``one \emph{right} electron of momentum $p$'' or ``one \emph{left} electron of momentum $p$''.
What we do have is the state $a_p^\dagger|0\rangle$, ``one electron of momentum $p$'',
which is right with the probability $\cos^2(\alpha_p/2)$ and left with the probability $\sin^2(\alpha_p/2)$.
But there \emph{are} states ``one right (or one left) electron at position $x$ and time $t$''; they are the two components of
$\psi_X(x,t)|0\rangle$. The reason is that the Dirac equation in $1$ space- and $1$ time-dimension
(both lattice and continuum) has just \emph{one} (up to proportionality) solution for given momentum $p$ and positive
energy \cite[Proposition~16]{SU-22}.
\end{remark}

\subsection*{Definition for the infinite lattice}


\begin{definition} \label{def-abstract} Denote by $L_2[a;b]$ the Hilbert space of square-integrable functions $[a;b]\to\mathbb{C}$ with respect to the Lebesque measure up to changing on a set of measure zero.
Equip it with the inner product $\langle f| g\rangle:=\int_{[a;b]}f^*(p)g(p)\,dp$
antilinear in the first argument, where $^*$ denotes complex conjugation. Denote by $\bigoplus$, $\bigotimes$, and $\bigwedge$ the orthogonal direct sum, the tensor and exterior product of \emph{Hilbert spaces}, that is, \emph{completions} of the orthogonal direct sum, the tensor and exterior product of Hermitian spaces over $\mathbb{C}$.

Fix $\varepsilon,m>0$. The \emph{Hilbert space of Feynman anticheckers} is the Hilbert space
$$
\mathcal{H}:=\left(
\bigoplus_{n=0}^{\infty}\bigwedge^n
L_2\left[
-\frac{\pi}{\varepsilon};\frac{\pi}{\varepsilon}
\right]
\right)^{\otimes 2}.
$$
The vector $1\in \mathbb{C}=\bigwedge^0 L_2\left[
-\frac{\pi}{\varepsilon};\frac{\pi}{\varepsilon}\right]$ is denoted by $|0\rangle$. The dual vector is denoted by $\langle 0|\in \mathcal{H}^*$.
Denote by $\mathcal{H}^0\subset \mathcal{H}$ the (incomplete) linear span
of $\bigwedge^n L_2\left[-\frac{\pi}{\varepsilon};\frac{\pi}{\varepsilon}\right]
\bluevar{\otimes \bigwedge^l L_2\left[-\frac{\pi}{\varepsilon};\frac{\pi}{\varepsilon}\right]}$
for \bluevar{all} $n,\bluevar{l}$. 

The \emph{Hamiltonian of Feynman anticheckers} is the linear operator on
$\mathcal{H}^0$ given by (see~\eqref{eq-omega})
$$
H(u_1\wedge\dots \wedge u_n \otimes v_1\wedge\dots \wedge v_l):=
(\omega_p u_1)\wedge\dots \wedge u_n \otimes v_1\wedge\dots \wedge v_l+\dots+
u_1\wedge\dots \wedge u_n \otimes v_1\wedge\dots (\omega_pv_l)
$$
for all integers $n,l\ge 0$ and all $u_1,\dots,u_n,v_1,\dots,v_l\in L_2\left[
-\frac{\pi}{\varepsilon};\frac{\pi}{\varepsilon}
\right]$, where $\omega_p$ is understood as a function in $p$.
The \emph{evolution operator} is the bounded linear operator on
$\mathcal{H}$ given by
$$
e^{-iHt}(u_1\wedge\dots \wedge u_n \otimes v_1\wedge\dots \wedge v_l):=
(e^{-i\omega_pt} u_1)\wedge\dots \wedge (e^{-i\omega_pt}u_n)
\otimes (e^{-i\omega_pt}v_1)\wedge\dots \wedge (e^{-i\omega_pt}v_l).
$$
The operators 
$P$ and 
$e^{-iPx}$
are defined analogously, only $\omega_p$ and $t$ are replaced by $p$ and $x$.

For $f\in L_2\left[
-\frac{\pi}{\varepsilon};\frac{\pi}{\varepsilon}
\right]$, the \emph{creation operators} $a(f)^\dagger $ and $b(f)^\dagger$
\emph{of particles and antiparticles} respectively \emph{with momentum distribution} $f$
are the linear operators on $\mathcal{H}$ defined by
\begin{align*}
  a(f)^\dagger (u\otimes v)
  &:= \sqrt{n+1}\,(f\wedge u)\otimes v,\\
  b(f)^\dagger (u\otimes v)
  &:= (-1)^n\sqrt{l+1}\,u \otimes (f^*\wedge v)
\end{align*}
for all $u\in \bigwedge^n L_2\left[
-\frac{\pi}{\varepsilon};\frac{\pi}{\varepsilon} \right]$ and $v\in \bigwedge^l L_2\left[
-\frac{\pi}{\varepsilon};\frac{\pi}{\varepsilon} \right]$.
Their adjoint operators are denoted by $a(f)$ and $b(f)$ respectively.
For each $x,t\in \varepsilon\mathbb{Z}$
define the \emph{field operator} $\psi(x,t)\colon\mathcal{H}\to \mathcal{H}\oplus\mathcal{H}$ by
$$
\psi(x,t):=
\left(
\begin{matrix}
\psi_1(x,t)\\
\psi_2(x,t)
\end{matrix}
\right):=
\left(
\begin{matrix}
  a(f_{1,x,t})+b(f^*_{1,x,t})^\dagger\\
  a(f_{2,x,t})+b(f^*_{2,x,t})^\dagger
\end{matrix}
\right),
\quad\text{where}
\quad
\begin{aligned}
f_{1,x,t}&=\sqrt{{\varepsilon}/{2\pi}}\,i\cos(\alpha_p/2)e^{i\omega_pt-ipx},\\
f_{2,x,t}&=\sqrt{{\varepsilon}/{2\pi}}\,\sin(\alpha_p/2)e^{i\omega_pt-ipx}.
\end{aligned}
$$
(The creation, annihilation, and field operators are bounded;
see \cite[(4.57)]{Folland}.)
The \emph{Feynman propagator for Feynman anticheckers}
is defined through them analogously to Definition~\ref{def-abstract-box}.
\end{definition}

This construction of the propagator is consistent with the
ones from Definitions~\ref{def-anti-alg}, \ref{def-anti-combi}, and~\ref{def-anti-subgraphs}:

\begin{proposition} \label{p-abstract}
Let
$\mathrm{T}\langle f_{k,x,t}|f_{l,0,0}\rangle
:=\begin{cases}
\langle f_{k,x,t}|f_{l,0,0}\rangle,
& \mbox{if }t\ge 0; \\
-\langle f_{l,0,0}|f_{k,x,t}\rangle, & \mbox{if }t<0.
\end{cases}$
Then for all $x,t\in\varepsilon\mathbb{Z}$ we get
$$
\langle 0|\mathrm{T}\psi(x,t)\bar \psi(0,0)|0\rangle
=
i\,\mathrm{T}\begin{pmatrix}
\langle f_{1,x,t}|f_{2,0,0}\rangle & -\langle f_{1,x,t}|f_{1,0,0}\rangle\\
\langle f_{2,x,t}|f_{2,0,0}\rangle & -\langle f_{2,x,t}|f_{1,0,0}\rangle
\end{pmatrix}
=-\frac{i}{2}\begin{pmatrix}
\widetilde{A}_1(-x,t,m,\varepsilon)& \widetilde{A}_2(x,t,m,\varepsilon) \\
-\widetilde{A}_2(-x,t,m,\varepsilon) & \widetilde{A}_1(x,t,m,\varepsilon)
\end{pmatrix}.
$$
\end{proposition}

\begin{proof}
Case 1: $t\ge 0$. By definition, for each $k,l\in\{1,2\}$ we have
\begin{multline*}
\langle 0|\mathrm{T}\psi_k(x,t)\bar \psi_l(0,0)|0\rangle=
\langle 0|\psi_k(x,t)\bar \psi_l(0,0)|0\rangle=
-i(-1)^{l}\langle 0|\psi_k(x,t)\psi_{3-l}^\dagger(0,0)|0\rangle
=\\=
-i(-1)^{l}\psi_{k}^\dagger(x,t)|0\rangle\cdot
\psi_{3-l}^\dagger(0,0)|0\rangle
=-i(-1)^{l}\langle f_{k,x,t}|f_{3-l,0,0}\rangle,
\end{multline*}
where $\cdot$ denotes the inner product in $\mathcal{H}$.
The latter equality follows from
\begin{align*}
\psi_{k}^\dagger(x,t)|0\rangle  &=
\left(a(f_{k,x,t})^\dagger+b(f^*_{k,x,t})
\right)|0\rangle 
=f_{k,x,t}\otimes 1
\in \bigwedge^1 L_2\left[
-\frac{\pi}{\varepsilon};\frac{\pi}{\varepsilon}
\right] \otimes \bigwedge^0 L_2\left[
-\frac{\pi}{\varepsilon};\frac{\pi}{\varepsilon}
\right],
\end{align*}
where $b(f)|0\rangle=0$ because
for each $u\in \bigwedge^n L_2\left[
-\frac{\pi}{\varepsilon};\frac{\pi}{\varepsilon} \right]$ and $v\in \bigwedge^l L_2\left[
-\frac{\pi}{\varepsilon};\frac{\pi}{\varepsilon} \right]$ we have
$$
b(f)|0\rangle\cdot(u\otimes v)=
|0\rangle\cdot b(f)^\dagger(u\otimes v)=
|0\rangle\cdot (-1)^n\sqrt{l+1}\,u \otimes (f^*\wedge v)=0
$$
by the condition $\bigwedge^0 L_2\left[
-\frac{\pi}{\varepsilon};\frac{\pi}{\varepsilon}
\right]\otimes \bigwedge^0 L_2\left[
-\frac{\pi}{\varepsilon};\frac{\pi}{\varepsilon}
\right]\perp \bigwedge^n L_2\left[
-\frac{\pi}{\varepsilon};\frac{\pi}{\varepsilon}
\right]\otimes \bigwedge^{l+1} L_2\left[
-\frac{\pi}{\varepsilon};\frac{\pi}{\varepsilon}
\right]$.
Thus by the definitions of $f_{k,x,t}$, $\alpha_p$ and Proposition~\ref{th-equivalence}
we get
\begin{multline*}
\langle 0|\mathrm{T}\psi(x,t)\bar \psi(0,0)|0\rangle=
\frac{-i\varepsilon}{2\pi}
\int_{-\pi/\varepsilon}^{\pi/\varepsilon}
\begin{pmatrix}
  i\,\cos(\alpha_p/2)\sin(\alpha_p/2) & \cos^2(\alpha_p/2) \\
  -\sin^2(\alpha_p/2) & i\,\cos(\alpha_p/2)\sin(\alpha_p/2)
\end{pmatrix}e^{ipx-i\omega_pt}\,dp
=\\=
\frac{-i\varepsilon}{4\pi}
\int\limits_{-\pi/\varepsilon}^{\pi/\varepsilon}
\begin{pmatrix}
  \frac{im\varepsilon}{\sqrt{m^2\varepsilon^2+\sin^2 p\varepsilon}}
  & 1+\frac{\sin p\varepsilon}{\sqrt{m^2\varepsilon^2+\sin^2 p\varepsilon}} \\
  \frac{\sin p\varepsilon}{\sqrt{m^2\varepsilon^2+\sin^2 p\varepsilon}}-1 & \frac{im\varepsilon}{\sqrt{m^2\varepsilon^2+\sin^2 p\varepsilon}}
\end{pmatrix}e^{ipx-i\omega_pt}\,dp
=
-\frac{i}{2} \begin{pmatrix}
\widetilde{A}_1(-x,t,m,\varepsilon)& \widetilde{A}_2(x,t,m,\varepsilon) \\
-\widetilde{A}_2(-x,t,m,\varepsilon) & \widetilde{A}_1(x,t,m,\varepsilon)
\end{pmatrix}.
\end{multline*}

Case 2: $t<0$. An analogous computation shows that
$$
\langle 0|\mathrm{T}\psi_k(x,t)\bar \psi_l(0,0)|0\rangle=
i(-1)^{l}\psi_{3-l}(0,0)|0\rangle\cdot
\psi_{k}(x,t)|0\rangle
= i(-1)^{l}\langle f_{3-l,0,0}|f_{k,x,t}\rangle
= \left(-i(-1)^{l}\langle f_{k,x,t}|f_{3-l,0,0}\rangle\right)^*.
$$
Thus 
we get the same integral formula
as in Case~1, only the whole expression is conjugated.
The change of the variables $p\mapsto \pi/\varepsilon-p$
and Proposition~\ref{th-equivalence}
complete the proof.
\end{proof}

Formula~\eqref{eq-q} for the expected charge is consistent with the
expression through field operators, which we briefly recall now
(this paragraph is for specialists).
Under notation from \cite[\S6.4]{Folland},
$:\!\psi^\dagger(x,t)\psi(x,t)\!:$ is
the \emph{charge density operator} for the 
field
\cite[(3.113)]{Peskin-Schroeder}. Thus
minus the expected charge density in the state
$\psi_1^\dagger(0,0)|0\rangle$
(meaning one right electron at the origin) is
$$
\frac{\langle 0|\psi_1(0,0)
:\!\psi^\dagger(x,t)\psi(x,t)\!:
\psi_1^\dagger(0,0)|0\rangle}
{\langle 0|\psi_1(0,0)\psi_1^\dagger(0,0)|0\rangle}
=\frac{\left|\langle 0|\psi(x,t)\psi_1^\dagger(0,0)|0\rangle\right|^2}
{\langle 0|\psi_1(0,0)\psi_1^\dagger(0,0)|0\rangle}
=\frac{1}{2}|\widetilde{A}_1\left(x,t,m,\varepsilon\right)|^2+
\frac{1}{2}|\widetilde{A}_2\left(x,t,m,\varepsilon\right)|^2,
$$
which coincides with~\eqref{eq-q}.
Here the first equality can be deduced, for instance,
from a version of Wick's theorem \cite[(6.42)]{Folland}, and the second one
--- from Proposition~\ref{p-abstract}.

\subsection*{Wightman axioms}


The continuum limit of the new model is the well-known free spin-$1/2$
quantum field theory which of course satisfies Wightman axioms
\cite[\S5.5]{Folland}. One cannot expect the discrete model to satisfy all the axioms
before passing to the limit because they are strongly tied to
continuum spacetime and Lorentz transformations. Remarkably, some of them
still hold on the lattice. 

\begin{proposition}[Checking of Wightman axioms]\label{p-wightman}
The objects introduced in Definition~\ref{def-abstract}
satisfy the following conditions:
\begin{description}
\item[axiom 1:] $\psi_k(x,t)$ is a bounded
linear operator on $\mathcal{H}$
for each $(x,t)\in\varepsilon\mathbb{Z}^2$ and $k\in\{1,2\}$;
\item[axiom 2:] $|0\rangle$ is the unique up to proportionality
  vector in $\mathcal{H}$ such that $e^{iH\varepsilon}|0\rangle=e^{iP\varepsilon}|0\rangle=|0\rangle$;
\item[axiom 3:] the vectors
$\psi_{k_1}(x_1,t_1)\dots\psi_{k_l}(x_l,t_l)
\psi^\dagger_{k_{l+1}}(x_{l+1},t_{l+1})\dots \psi^\dagger_{k_{l+n}}(x_{l+n},t_{l+n})|0\rangle$
for all $0\le n,l\in\mathbb{Z}$,
$k_j\in\{1,2\}$, $(x_j,t_j)\in \varepsilon\mathbb{Z}^2$
span a dense linear subspace in $\mathcal{H}$;
\item[axiom 4 (weakened):] 
$e^{biH-aiP}\psi_k(x,t)e^{aiP-biH}=\psi_k(x+a,t+b)$
for each $(a,b)\in\varepsilon\mathbb{Z}^2$;
\item[axiom 5 (weakened):]
$H\ge 0$ and $H^2-(1-\omega_0\varepsilon/\pi)^2P^2\ge 0$ \blueit{(here $\omega_0$ is $\omega_p$ for $p=0$)};
\item[axiom 6:]
$[\psi_{k}(x,t),\psi_{k'}(x',t')]_+
=[\psi_{k}(x,t),\psi^\dagger_{k'}(x',t')]_+
=[\psi^\dagger_{k}(x,t),\psi^\dagger_{k'}(x',t')]_+=0$
for all $k,k'\in\{1,2\}$ and $(x,t),(x',t')\in \varepsilon\mathbb{Z}^2$ such that $|x-x'|>|t-t'|$,
where $[a,b]_+:=ab+ba$.
\end{description}
\end{proposition}

Here Axiom~1 is weaker than the continuum one in the sense that the field operators are defined
only on the lattice, but stronger in the sense that they are genuine bounded operators rather than distributions.
The vectors in Axiom~3 mean states with electrons at the points
$(x_{l+1},t_{l+1}),\dots,(x_{l+n},t_{l+n})$ and positrons at the points
$(x_1,t_1),\dots,(x_l,t_l)$.
Axiom~4 is much weaker than the continuum one, which
involves general Lorentz transformations. On the lattice,
only translations remain, because (almost all) the other
Lorentz transformations do not preserve the lattice.
Axiom~5 in continuum theory asserts that $H\ge 0$ and $H^2-P^2\ge 0$, i.e.
the energy is positive in any frame of reference.
The inequality $H^2-P^2\ge 0$ is violated on the lattice
(but this does not mean negative energy because
Lorentz transformations do not preserve the lattice).
The weakened Axiom~5 shows that it still holds ``in the continuum limit''.
Axiom~6 is equivalent to vanishing of the real part of the
Feynman propagator outside the light cone; this is obvious in the original
Feynman model but nontrivial in the new one.

\begin{proof}
Axioms~1 and~2 hold by definition; recall that the operators are bounded by~\cite[(4.57)]{Folland}.
Axiom~3 holds because linear span of the functions
$(\psi_1(x,0)+\psi_2(x,0))|0\rangle
=i\sqrt{\varepsilon/\pi}\,e^{-i\alpha_p/2-ipx}
$, where $x$ runs through $\varepsilon\mathbb{Z}$,
is dense in $L_2\left[
-\frac{\pi}{\varepsilon};\frac{\pi}{\varepsilon}\right]$.
Axiom 4 is checked directly. 
Axiom~5 follows from the inequality
$\omega_p\ge |p|(1-\omega_0\varepsilon/\pi)$.

To check Axiom~6, recall that all the anticommutators of the operators
$a(f),a(f)^\dagger,b(f),b(f)^\dagger$ vanish except
$[a(f),a(g)^\dagger]_+=[b(f^*),b(g^*)^\dagger]_+=\langle f|g\rangle I$
\cite[(4.56) and p.96]{Folland}. This immediately
implies that $[\psi_{k}(x,t),\psi_{k'}(x',t')]_+
=[\psi^\dagger_{k}(x,t),\psi^\dagger_{k'}(x',t')]_+=0$.
Finally, by Proposition~\ref{p-abstract},
Theorem~\ref{th-well-defined}, and Definition~\ref{def-mass}
for $|x-x'|>|t-t'|$ we have
\begin{multline*}
[\psi_{k}(x,t),\psi^\dagger_{k'}(x',t')]_+
=[a(f_{k,x,t}),a(f_{k',x',t'})^\dagger]_+
+ [b(f^*_{k,x,t})^\dagger,b(f^*_{k',x',t'})]_+
=\\=\langle f_{k,x,t}|f_{k',x',t'}\rangle I
+\langle f_{k',x',t'}|f_{k,x,t}\rangle I
=2I\,\mathrm{Re}\langle f_{k,x-x',t-t'}|f_{k',0,0}\rangle
=0.\\[-1.55cm]
\end{multline*}
\end{proof}


\subsection*{Acknowledgements}

The work was prepared within the Russian Science Foundation grant N22-41-05001, \url{https://rscf.ru/en/project/22-41-05001/}. \mscomm{Check wording!!!}
The authors are grateful to P.~Arnault, C.~B\'eny, \blueit{J.~Bj\"ornberg,} A.~Borodin, M.~Drmota, \bluevar{A.~Glazman, V.~Gorin,} M.~Gualtieri, \bluevar{A.~Kholodenko, R.~Kenyon,} F.~Kuyanov, \bluevar{K.~Ryan, S.~Tata} for useful discussions.

\tobeadded
This work was presented as ...
\endtobeadded

{
\footnotesize

\subsection*{Statements and Declarations}

%
\textbf{Conflict of interests.} The authors have no relevant financial or non-financial interests to disclose.


\noindent
\textsc{Mikhail Skopenkov ORCID: 0000-0003-2453-0009\\
King Abdullah University of Science and Technology, \blueit{Saudi Arabia} and\\
HSE University (Faculty of Mathematics), \bluevar{Russia}} \remove{}
\\
\texttt{mikhail.skopenkov\,@\,gmail$\cdot $com} \quad \url{https://users.mccme.ru/mskopenkov/}

\vspace{0.3cm}
\noindent
\textsc{Alexey Ustinov ORCID: 0000-0002-0624-0406\\
HSE University (Faculty of Computer Science) and\\
Khabarovsk Division of the Institute for Applied Mathematics,\\
Far-Eastern Branch,
Russian Academy of Sciences, Russia} 
\\
\texttt{Ustinov.Alexey\,@\,gmail$\cdot $com} \quad
\blueit{\url{https://www.hse.ru/en/org/persons/530309935}}

}

\end{document}